\documentclass[letter,twoside,english,11pt]{article}
\usepackage[sumlimits]{amsmath}
\usepackage{amsfonts,amssymb,dsfont}
\usepackage{times}
\usepackage{type1cm,graphicx}

\newcommand{\vb}{\vspace{3mm}\noindent}

\usepackage[left=2cm,top=2cm,right=2cm]{geometry}
\usepackage{babel}
 \usepackage[pdfauthor={Elad Eban, Dorit Aharonov,Michael
   Ben-Or,Urmila Mahadev},pdftitle={Interactive Proofs For Quantum Computation},
 pdfcreator={Elad Eban},colorlinks=false]{hyperref}

\usepackage[usenames]{color}
\usepackage{bbm,amssymb,amsmath}

\newcommand{\BQP}{{\sf BQP}}
\newcommand{\NP}{{\sf NP}}

\newcommand{\BPP}{{\sf BPP}}
\newcommand{\IP}{{\sf IP}}

\newcommand{\PSPACE}{{\sf PSPACE}}
\newcommand{\Eq}[1]{Eq.~\ref{#1}}
\newcommand{\Le}[1]{Lemma~\ref{#1}}
\newcommand{\Cl}[1]{Claim~\ref{#1}}
\newcommand{\Th}[1]{Theorem~\ref{#1}}
\newcommand{\Prot}[1]{Protocol~\ref{#1}}

\newcommand{\Def}[1]{Definition~\ref{#1}}

\newcommand{\Sec}[1]{Sec.~\ref{#1}}

\newcommand{\pen}[1]{Appendix~\ref{#1}}

\newcommand{\EqDef}{\stackrel{\mathrm{def}}{=}}

\newcommand{\bra}[1]{\left< #1\right|}
\newcommand{\ket}[1]{\left| #1\right>}

\newcommand{\tr}{\mbox{Tr}}

\newcommand{\mN}{\mathbbm{N}}
\newcommand{\mbP}{\mathbbm{P}}

\newcommand{\mcI}{\mathcal{I}}

\newcommand{\mcA}{\mathcal{A}}
\newcommand{\mcB}{\mathcal{B}}

\newcommand{\mcO}{\mathcal{O}}
\newcommand{\mcK}{\mathcal{K}}

\newcommand{\mcV}{\mathcal{V}}
\newcommand{\mcL}{\mathcal{L}}

\newcommand{\mfC}{\mathfrak{C}}

\newcommand{\ignore}[1]{}

\newtheorem{thm}{Theorem}[section]

\newtheorem{deff}{Definition}[section]

\newtheorem{protocol}{Protocol}[section]
\newtheorem{claim}[thm]{Claim}
\newtheorem{lem}[thm]{Lemma}

\newtheorem{corol}[thm]{Corollary}
\newtheorem{fact}{Fact}[section]

\newtheorem{remark}{Remark}[section]

\newenvironment{proofof}[1]{\noindent{\textit Proof} of $\bf{#1}$:\hspace*{1em}}{$\Box $}
\newenvironment{statement}[1]{\noindent{\textbf {#1}\hspace*{0.5em}}}{$\\$}

\def\hpic #1 #2 {\mbox{$\begin{array}[c]{l}
      \epsfig{file=#1,height=#2} \end{array}$}}

\def\vpic #1 #2 {\mbox{$\begin{array}[c]{l}
      \epsfig{file=#1,width=#2} \end{array}$}}

\newcommand{\dnote}[1]{\textcolor{red}{\small {\textbf{(Dorit:}
#1\textbf{) }}}}

\newcommand{\qc}{\mbox{\textsf{Q-CIRCUIT}}}
\newcommand{\QPIP}{\textsf{QPIP}}
\newcommand{\QAS}{\textsf{QAS}}

\newcommand{\wt}[1]{{\widetilde{#1}}}
\newcommand{\odots}{{\otimes{\ldots}\otimes}}

\newenvironment{proof}{\noindent\textit{Proof: }}{$\Box $}

\begin{document}

\author{Dorit Aharonov\footnote{School of Computer Science, The Hebrew
    University of Jerusalem, Israel. $\{$doria,benor,elade$\}$@cs.huji.ac.il}
  \and Michael Ben-Or$^*$ \and Elad Eban$^*$ \and Urmila Mahadev\footnote{Department of Computer Science, UC Berkeley, California. mahadev@cs.berkeley.edu}}

\title{Interactive Proofs For Quantum Computations}

\maketitle

\begin{abstract}
The widely held belief that \BQP\ strictly contains \BPP\ raises fundamental
questions: upcoming generations of quantum computers might already be too
large to be simulated classically. Gottesman asked (\cite{aaronsonpost}):
Is it possible to experimentally test that
these systems perform as they should, if we cannot efficiently compute
predictions for their behavior? As phrased by Vazirani in \cite{vaziranitalk}:
If computing predictions
for quantum mechanics requires exponential resources,
is quantum mechanics a falsifiable theory? In
cryptographic settings, an untrusted future company wants to sell a quantum
computer or perform a delegated quantum computation. Can the customer be
convinced of correctness without the ability to compare results to
predictions?

To provide answers to these questions, we define Quantum Prover Interactive
Proofs (\QPIP). Whereas in standard interactive proofs \cite{goldwasser1985kci} the prover is
computationally unbounded, here our
prover is in \BQP, representing a quantum computer.
The verifier models our current computational capabilities: it is a \BPP\
machine, with access to only a 
few qubits. Our main theorem can be roughly stated
as: ``Any language in \BQP\ has a \QPIP\ which hides the computation from the prover''. We provide two proofs. The
simpler one uses a new (possibly of independent interest) quantum
authentication scheme (\QAS) based on random Clifford elements. This \QPIP, however, involves two way quantum communication for polynomially many rounds.  
Our second protocol uses polynomial codes \QAS\
due to Ben-Or, Cr{\'e}peau, Gottesman, Hassidim, and Smith \cite{benor2006smq},
combined with secure multiparty quantum
computation techniques. This protocol involves quantum communication from the verifier to the prover at the start of the protocol, and classical communication throughout the rest of the protocol. 
Both protocols are inherently ``blind'': both the quantum
circuit and the input remain unknown to the prover.

This is the journal version of work reported in 2008 (\cite{abe2008}) 
and presented in ICS 2010. The protocols are slightly modified from the original version, whereas some of the proofs required major modifications and 
corrections. Notably, the claim that the polynomial \QPIP\ is fault 
tolerant was removed.

After deriving the results in \cite{abe2008}, we learnt that 
Broadbent, Fitzsimons, and Kashefi \cite{broadbent2008ubq} have independently 
suggested ``universal blind quantum computation'' using completely different 
methods (measurement based 
quantum computation). Their construction implicitly implies 
similar implications. The protocol in \cite{broadbent2008ubq} was flawed but based on similar ideas, Fitzsimons and Kashefi have provided a protocol and 
proof of blind verifiable computation in \cite{fk2012}. 
The initial independent works (\cite{abe2008},\cite{broadbent2008ubq}) 
ignited a long line of research of blind and verifiable quantum 
computation, which we survey here, along with connections to various  
cryptographic problems. Importantly, the problems of making the results fault tolerant, as well as removing the need for quantum communication altogether,
remain open.   
\end{abstract}
\newpage

\section{Introduction}
\subsection{Motivation}\label{sec:funda}
As far as we know today, the quantum mechanical description
of many-particle systems requires exponential resources to simulate.
This has the following fundamental implication:
the results of an experiment conducted on a many-particle physical system
described by quantum mechanics cannot be predicted (in general)
by classical computational devices in any reasonable amount of time.
This important realization (or belief),
which stands at the heart of the interest in quantum computation,
led Gottesman \cite{aaronsonpost} to ask: 
Can a classical verifier verify the correctness of quantum evolutions? 
The question was phrased by Vazirani \cite{vaziranitalk} as: Is quantum mechanics a falsifiable
physical theory? Assuming that small quantum systems obey quantum
mechanics to an extremely
high accuracy, it is still possible that the physical description of
large systems deviates significantly from quantum mechanics.
Since there is no efficient way to make the predictions of the experimental
outcomes for most large quantum systems, there is no way
to test or falsify this possibility experimentally, using the usual
scientific paradigm of predict and compare.
\\~\\
This question has practical implications.
Experimentalists who attempt to realize quantum computers
would like to know how to test that their systems indeed perform
the way they should. But most tests cannot be compared to any
predictions! The tests whose predictions can in fact be computed do not actually test the more interesting aspects of quantum mechanics, namely
those which cannot be simulated efficiently classically.
\\~\\
The problem arises in cryptographic situations as well. It is natural to expect that the first
generations of quantum computers will be extremely expensive,
and thus quantum
computations would be delegated to untrusted
companies. Is there any way for the customer to trust the
outcome, without the need to trust the company which performed the
computation, even though the customer cannot verify the
outcome of the computation (since he cannot simulate it)?
And even if the company is honest,
can the customer detect innocent errors in such a computation? Given the amounts of grant money and prestige involved,
the possibility of
dishonesty of experimentalists
and experimentalists' bias inside the academia should not be
ignored either \cite{roodman2003bap, BlindWiki}.
\\~\\
As Vazirani points out \cite{vaziranitalk}, an answer
to these questions is already given in the form of Shor's factoring algorithm \cite{shor1997pta}. Indeed, quantum mechanics does not seem to be falsifiable
using the {\it usual} scientific paradigm, assuming that \BQP\ is
strictly lager than \BPP. However, Shor's algorithm does provide a way
for falsification, by means of
an experiment which lies outside of the usual scientific paradigm:
though its result cannot be {\it predicted} and then compared to the
experimental outcome, it can be {\it verified} once the outcome
of the experiment is known (by simply taking the product of the factors
and checking that this gives the input integer).
\\~\\
This, however, does not fully address the issues raised above. Consider, for example, a company called \textit{Q-Wave} which is trying to convince a customer that it has managed to build a quantum computer of 100 qubits. Such a system is already too big to simulate classically. However, any factoring algorithm that is run on a system of a $100$ qubits
can be easily performed by today's classical technology.
For delegated quantum computations, how can Shor's algorithm help
in convincing a customer of correctness of, say, the computation of the \BQP\
complete problem of approximating the Jones polynomial
\cite{aharonov2006pqa,jonesHardness,fklw2001,bfld2009}?
As for experimental results, it is difficult
to rigorously state which aspects of quantum mechanics are exactly falsified or verified
by the possibility to apply Shor's algorithm.
Moreover, we are now facing a time in which small quantum computers of a few tens of qubits may very well be realizable; yet, factoring is still impossible in such systems, and we would nevertheless like to be able to test their evolution.
\\~\\
We thus pose the following main question:
Can one be convinced of the correctness of the computation of {\it any}
polynomial quantum circuit? Alternatively, can one be convinced of
the ``correctness'' of the quantum mechanical description of
any quantum experiment that can be conducted in the laboratory,
even though one cannot compute the predictions for the outcomes of this
experiment? In this paper, we address the above fundamental question
in a rigorous way and provide a positive answer to these questions, in a well defined framework. We do this by taking a computational point of
view on the interaction between
the supposed quantum computer, and the entity which attempts to
verify that the quantum computer indeed computes what it should.

\subsection{Quantum Prover Interactive Proofs (\QPIP)}
Interactive proof systems, defined by Goldwasser, Micali and Rackoff
\cite{goldwasser1985kci}, play a crucial role in the theory of computer
science.  Roughly, a language $\mcL$ is said to have an interactive proof if
there exists a computationally unbounded prover (denoted $\mathds{P}$) and a \BPP\
verifier ($\mathds{V}$) such that for any instance $x$ in $\mcL$, $\mathds{P}$ convinces $\mathds{V}$ of the
fact that $x\in\mcL$ with probability $\ge \frac 2 3$ (completeness).
Otherwise,
when $x\notin\mcL$, there does not exist a prover who can convince $\mathds{V}$ that $x\in\mcL$ with probability higher
than $\frac 1 3$ (soundness).
\\~\\
Quantum interactive proofs in which the prover is
an \emph{unbounded}
quantum computer, and the \emph{verifier} is a \BQP\ machine have previously been studied in \cite{watrous2003phc}. The starting point of this work is the observation that Shor's factoring algorithm \cite{shor1997pta} can be viewed as an interactive
proof of a very different kind: one between a classical \BPP\ verifier, and a
quantum \textit{polynomial time} (\BQP) prover, in which
the prover convinces the verifier of the factors of a given number
(this can be easily converted to the usual \IP\ formalism of membership
in a language by converting the search problem to a decision problem in the standard way).
\\~\\
One might suspect that such an interactive proof exists for all problems inside \BQP\ $\cap$ \NP\, by asking the \BQP\ prover to find the witness, which the classical verifier can then verify. We do not know this to be true; the trouble with this argument is that the fact that the problem is in \BQP\ $\cap$ \NP\ does not guarantee that the \BQP\ machine can also {\it find} a witness efficiently - decision to search reductions are known only for \NP-complete problems.  In any case, it is widely believed that \BQP\ is not contained in \NP\ (
and in fact not even in the polynomial hierarchy - see \cite{aaronson2009} and references therein). The main goal
of this paper is to generalize the interactive point of view
of Shor's algorithm, as mentioned above, in order to show that a \BPP\ verifier can be convinced of the result of {\it any} polynomial
quantum circuit, using interaction with the \BQP\ prover (the quantum
computer). In other words, we would like to extend the above interactive proof (which is specific to factoring) to a BQP complete problem. 
\\~\\
To this end we define a new model of quantum interactive proofs which we call
quantum prover interactive proofs (\QPIP). The simplest definition would be
an interactive proof in which the prover is a \BQP\ machine and the verifier a
\BPP\ classical machine.  In some sense, this model captures the
possible interaction
between the quantum world (for instance, quantum systems in the lab) and the
classical world. However, we do not know how to provide interactive proofs for all problems in BQP with only classical interaction; this is a major open problem (see Section \ref{sec:openquestions}).
We therefore modify the model a little, and
allow the verifier additional access to a constant number of qubits.
The verifier can be viewed as modeling our current computational
abilities, and so in some sense, the verifier represents ``us''.

\begin{deff}\label{def:QPIP} A Language $\mcL$ is said to have a 
Quantum Prover Interactive Proof ($\QPIP_{\kappa}$) with completeness $c$ 
and soundness $s$ (where $c-s$ is a constant) if there exists a pair of algorithms $(\mathds{P},\mathds{V})$, where $\mathds{P}$ is the prover and $\mathds{V}$ is the verifier, 
with the following properties:
\begin{itemize}
\item The prover $\mathds{P}$ is a \BQP\ machine, which also has access to a quantum 
channel which can transmit $\kappa$ qubits. 
\item The verifier $\mathds{V}$ is a hybrid quantum-classical machine. 
Its classical part is
      a \BPP\ machine. The quantum part is a register of $\kappa$ qubits, on which the verifier
 can perform arbitrary quantum operations and which has access to a quantum channel which can transmit $\kappa$ qubits. At
      any given time, the verifier is not allowed to possess
      more than $\kappa$ qubits. The 
      interaction between the quantum and classical parts of the verifier
 is the usual one: the
      classical part controls which operations are to be performed on the
      quantum register, and outcomes of measurements of the quantum register can
      be used as input to the classical machine.
\item There is also a classical communication channel between the 
prover and the verifier, which can transmit polynomially many bits at 
any step. 
\item At any given step, either the verifier or the prover perform computations on their registers and send bits and qubits through the relevant channels to the other party.
\end{itemize}
We require:
\begin{itemize}
\item \textbf{Completeness}: if $x\in\mcL$, then after interacting with $\mathds{P}$, $\mathds{V}$ accepts with probability $\ge c$. 
\item \textbf{Soundness}: if $x\notin \mcL$, then the verifier rejects with probability $\ge 1-s$ \it{regardless} of the prover $\mathds{P}'$ (who has the same description as $\mathds{P}$) with whom he is interacting. 
\end{itemize} 
\end{deff}
Abusing notation, we
denote the class of languages for which such a proof
 exists also
by $\QPIP_{\kappa}$. Throughout the paper, when we refer to \QPIP\ without a subscript, we are assuming the subscript is a constant $c$. We remark that our definition of $\QPIP_{\kappa}$ is asymmetric -
the verifier is ``convinced'' only if the quantum circuit
outputs $1$. This asymmetry seems irrelevant in
our context of verifying correctness of quantum computations.  Indeed, it is
possible to define a symmetric version of $\QPIP_{\kappa}$
(which we denote by $\QPIP_{\kappa}^{sym}$)
in which the verifier is convinced of {\it correctness} of the prover's outcome
whether or not $x\in \mcL$ rather than only if $x\in \mcL$; see Appendix \ref{sec:symmetric}
for the definition. 

\subsection{Main Results}
Our main results are phrased in terms of the \BQP\ complete problem 
$\qc_{\gamma}$: 

\begin{deff}\label{def:qcircuit} The promise problem $\qc_{\gamma}$ consists of a quantum circuit made of a
  sequence of gates, $U=U_N{\ldots}U_1$, acting on $n$ input bits. The task is
  to distinguish between two cases for all $x\in \{0,1\}^n$:
\begin{eqnarray*}
  \qc_{\textmd{YES}}&:
  \|(\left(\ket 1 \bra 1 \otimes \mcI_{n-1}\right)U\ket{x} \|^2\ge 1-\gamma\\
  \qc_{\textmd{NO}}\;\,&:
  \|(\left(\ket 1 \bra 1 \otimes \mcI_{n-1}\right)U\ket{x} \|^2\le \gamma
\end{eqnarray*}
when we are promised that one of the two cases holds.
\end{deff}
$\qc_{\gamma}$ is a \BQP\ complete problem as long as $1 - 2\gamma >\frac 1 {poly(n)}$. Throughout this paper, if we refer to $\qc$ (without the parameter $\gamma$), we are assuming that $\gamma$ satisfies the above inequality.
Our main result is:
\begin{thm}\label{thm:qcircuit}
 For $0<\epsilon<1$ and $\gamma < 1 - \epsilon$, the language $\qc_{\gamma}$ has a $\QPIP_{O(\log(\frac{1}{\epsilon}))}$ with completeness $1 - \gamma$ and soundness $\epsilon + \gamma$.
\end{thm}
We note that although we provide \QPIP\ protocols only for the language \qc\ (for which the initial state is always a standard basis state), our proofs can be extended to security for a modified language for which the initial state is an arbitrary quantum state. In addition, we prove soundness against an unbounded prover, rather than a \BQP\ prover (as given in Definition \ref{def:QPIP}). By setting $\epsilon$ to a constant, we obtain a $\QPIP_c$ for a constant $c$, which gives our main theorem:
\begin{thm} \label{thm:main}
There exists a constant $c$ for which $\BQP\ = \QPIP_c$.
\end{thm}
\begin{proof}
Since \qc\ is \BQP\ complete, the fact that \BQP\ is in $\QPIP_c$ follows from Theorem \ref{thm:qcircuit}. $\QPIP_c$ is trivially in \BQP\ since the \BQP\ machine can simulate both prover, verifier and their interactions.  
\end{proof}
\\~\\
Thus, a \BQP\ prover can convince the verifier of any language he can
compute. Since \BQP\ is closed under  completion, we also get equality 
to the symmetric version of $\QPIP_c$ (see Appendix \ref{sec:symmetric} for the proof): 
\begin{corol}\label{thm:mainsym}
There exists a constant $c$ for which 
$\BQP=\QPIP_c^{sym}$
\end{corol}
Our main tools for the proof of Theorem \ref{thm:qcircuit}
are quantum authentication schemes (\QAS) \cite{barnum2002aqm}.
Roughly, a \QAS\ allows two parties to communicate in the following way. First, Alice
sends an encoded quantum state to Bob. Then Bob decodes the state and decides whether or not it is valid. If the state was not altered along the way, then upon decoding, Bob gets the same state that Alice had sent (and declares it valid). If the
state was altered, the scheme is $\epsilon$-secure if Bob declares a wrong state valid with probability at most $\epsilon$.  The basic idea used to extend a \QAS\ to a \QPIP\ is that similar 
security can be achieved, even if the
state needs to be rotated by unitary gates, as long as the verifier can control
how the unitary gates affect the authenticated states.

\subsubsection{Clifford \QAS\ based \QPIP} We start with a simple \QAS\ (which we extend to a \QPIP) based on random Clifford group operations (it is reminiscent of Clifford
based quantum $t$-designs \cite{ambainis2007qtd,ambainis2008tre}). The Clifford \QAS\ based \QPIP\ demonstrates some key ideas and might be of interest on
its own due to its simplicity. However, the \QPIP\ has the disadvantage that it 
requires two way quantum communication between the prover and the verifier. 
\\~\\
We first describe the Clifford \QAS. To encode a
state of $n$ qubits, Alice tensors the state with $e$ qubits in the state $\ket{0}$,
and applies a random Clifford operator on the $n+e$ qubits. To decode, Bob removes the Clifford operator chosen by Alice and checks if the $e$ auxiliary qubits are in the state $\ket{0}^{\otimes e}$ (see Protocol \ref{protocol:cliffordqas} for a complete description of the \QAS). We prove the following theorem:
\begin{thm}{\label{thm:CliffordAuth}}
The Clifford scheme given in Protocol \ref{protocol:cliffordqas} is a \QAS\ with security
  $\epsilon=2^{-e}$.
\end{thm}
This \QAS\ might be interesting
in its own right due to its simplicity. To construct a \QPIP\ using this \QAS, we simply use the prover as an untrusted
storage device: the verifier asks the prover for the authenticated qubits on
which he would like to apply the next gate, decodes them by applying the appropriate inverse Clifford operators, applies the gate, applies new random Clifford operators and sends the resulting qubits to the prover. As we show in the following theorem, this protocol (see \Prot{prot:CliffordIP} for full details) is a \QPIP:
\begin{thm}\label{thm:CliffordIP} For $0< \epsilon < 1$ and $\gamma < 1 - \epsilon$, \Prot{prot:CliffordIP} is a $\QPIP_{O(\log(\frac{1}{\epsilon}))}$ with completeness $1-\gamma$ and soundness $\gamma+\epsilon$ for
  $\qc_{\gamma}$.
\end{thm}

\subsubsection{Polynomial code \QAS\ based \QPIP\ }
Our second type of \QPIP\ uses a
\QAS\ due to Ben-Or, Cr\'epeau, Gottesman, Hassidim and Smith
\cite{benor2006smq}. This \QAS\ is based on signed quantum polynomial codes (defined in Definition \ref{def:SignedPolynomial}) ,
which are quantum polynomial codes \cite{aharonov1997ftq}
of degree at most $d$ multiplied by
some random sign ($1$ or $-1$) at every
coordinate (this is called the sign key $k$)
and a random Pauli at every coordinate (the Pauli key). The \QAS\ simply consists of Alice encoding a single qudit using the signed polynomial code and Bob checking if the received state is indeed encoded under the signed polynomial code (described in further detail in Protocol \ref{protocol:polynomialqas}). We prove the following theorem: 
\begin{thm}{\label{thm:PolynomialAuth}} The polynomial authentication scheme as described in Protocol \ref{protocol:polynomialqas}
is a \QAS\ with security $\epsilon = 2^{-d}$.
\end{thm}
The security proof of the polynomial code based \QAS\ is subtle, and 
was missing from the original paper \cite{benor2006smq}; we provide it here. 
\\~\\
To extend the polynomial based \QAS\ to a \QPIP, we first note that performing Clifford gates in this setting 
is very easy. Due to its algebraic structure, the signed polynomial code
allows applying Clifford gates without knowing the sign key (if the same sign key is used for all registers).
This was used in \cite{benor2006smq} for
secure multiparty quantum computation; here we use it
to allow the prover to perform gates without
knowing the sign key or the Pauli key. To perform Toffoli gates, the verifier first creates authenticated magic states (used in \cite{shor1996},\cite{bravyi2005uqc},\cite{benor2006smq}) and sends them to the prover. The prover can apply a Toffoli gate by first applying Clifford operations between the computation qubits and a magic state, then measuring 3 of the computation qubits, and then adaptively applying a Clifford correction based on the measurement results (for more details on applying Toffoli gates using Toffoli states see Section \ref{app:toffoli}). Note that since the prover obtains a measurement result encoded under the polynomial \QAS, he must send it to the verifier to be decoded before he can perform an adaptive Clifford correction. It follows that with authenticated magic states and classical assistance from the verifier,
the prover can perform universal computation using
only Clifford group operations and measurements
(universality was proved for qubits in \cite{bravyi2005uqc}
and for higher dimensions it was shown in \cite{aharonov1997ftq}).
\\~\\
The polynomial \QPIP\ protocol (Protocol \ref{prot:PolynomialIP}) goes as follows. The prover receives
all authenticated qubits in the beginning. Those
include the inputs to the circuit, as well as
authenticated magic states required to perform
Toffoli gates. The prover can then perform universal computation as described above. Except for the first round, any further
communication between the verifier and prover (occuring when implementing
the Toffoli gates) is thus classical. We show that this protocol is a \QPIP\ in the following theorem: 
\begin{thm}\label{thm:PolynomialIP} For $0< \epsilon < 1$ and $\gamma < 1 - \epsilon$, Protocol \ref{prot:PolynomialIP} is a $\QPIP_{O(\log(\frac{1}{\epsilon}))}$ protocol with
  completeness $1-\gamma$ and soundness $\gamma + \epsilon$ for $\qc_{\gamma}$.
\end{thm}
We remark that in the study of the classical notion of \IP,
a natural question is to ask how powerful the prover must be to prove
certain classes of languages.
It is known that a \PSPACE\ prover is capable of proving any language
in  \PSPACE\ to a \BPP\ verifier, and similarly, it is known that
 \NP\ or \textsf{\#P}
 restricted provers
 can prove any language which they can compute to a \BPP\ verifier. This is not known for
\textsf{coNP}, \textsf{SZK} or \textsf{PH} \cite{arora:ccm}.
It is natural to ask what is the power of a \BQP\ prover;
our results imply that such a prover can prove the entire
class of \BQP\ (albeit to a verifier who is not entirely classical).
Thus, we provide a characterization of the power of a \BQP\ prover. We stress 
the open question of characterizing this power when the interaction between the prover and verifier is completely classical 
(discussed in Section \ref{sec:openquestions}).

\subsubsection{Blindness}
In the works \cite{childs2001saq,blind} a related
question was raised: in our cryptographic setting, if we distrust
the company performing the delegated quantum computation,
we might want to keep both the input and the
function which is being computed secret.
Can this be done while maintaining the
confidence in the outcome? A simple modification of our protocols to work on 
universal quantum circuits 
gives the following theorem, which we prove in Section \ref{app:blind}:

\begin{thm}\label{thm:blind}
\Th{thm:qcircuit} holds also in a blind setting,
namely, the prover does not get any
 information regarding the function being computed and its input.
\end{thm}
We note that an analogous result for \NP-hard problems
was shown already in the late $80$'s to be impossible 
(in the setting of classical communication)
unless the polynomial hierarchy collapses \cite{abadi1987hio}. 
\\~\\
To achieve \Th{thm:blind}, we modify our construction so that the circuit that the prover performs is a {\it universal quantum circuit}, i.e.,
a fixed sequence of gates which
gets as input a description of a quantum circuit (with gates from a constant size universal set of gates) plus an input string to that
circuit, and  applies the input quantum circuit
to the input string.  Since the universal quantum circuit is fixed, it reveals
nothing about the input quantum circuit or the input string to it. To prove blindness, we simply need to show that the input states provided to the prover by the verifier (and all messages provided to the prover during the protocol) do not leak information about the input to the universal circuit. This is done by showing that at all times, the prover's state is independent of the input to the universal circuit. 
\\~\\
Proving blindness of the Clifford scheme is quite straightforward and done in the following theorem (which we prove in Section \ref{app:blind}):
\begin{thm}[\bfseries Blindness of the Clifford Based \QPIP]\label{thm:blindclifford}
The state of the prover in the Clifford based \QPIP\ (Protocol \ref{prot:CliffordIP}) is independent of the input to the circuit which is being computed throughout the protocol.
\end{thm}
On the other hand, proving blindness of the polynomial scheme is a bit more involved, due to the classical interaction: 
\begin{thm}[\bfseries Blindness of the Polynomial Based \QPIP]\label{thm:blindpolynomial}
The state of the prover in the polynomial based \QPIP\ (Protocol \ref{prot:PolynomialIP}) remains independent of the input to the circuit which is being computed throughout the protocol.
\end{thm}

\subsubsection{Interpretation} \label{sec:interpretation}
We will now present some corollaries which clarify the connection between the
results and the motivating questions, and show that
one can use the \QPIP\ protocols designed here to address the
various issues raised in \Sec{sec:funda}. 
\\~\\
We start with some basic questions. Conditioned that the verifier
does not abort, does he know that
the final state of the machine is very close to
the correct state that was supposed to be
the outcome of the computation?
This unfortunately is not the case. It may be that
the prover can make sure that the verifier
aborts with very high probability, but when he does {\it not} abort, 
the outcome of the computation is wrong. However a weaker form of the above
result (which achieves what is reasonable to hope for) does hold: if we know that the probability of not aborting is high,
then we can deduce something about the probability of the final state being very close to the correct state.

\begin{corol} \label{corol:confidence}
For the Clifford based \QPIP\ protocol with security parameter 
$\epsilon$, if the
verifier does not abort with probability $\ge \beta$
then the trace distance between the final density matrix conditioned on the verifier's accepting 
and the correct final state,
is at most $\frac {\epsilon} \beta$.
\end{corol}
The proof is given in Section \ref{app:interpretation} - it is simple for the Clifford \QPIP. For the polynomial scheme a similar corollary holds, 
with a proof which is more involved: 
\begin{corol} \label{corol:polyconfidence}
For the polynomial based \QPIP\ protocol with security parameter $\epsilon$, if the
verifier does not abort  with probability $\ge \beta$, and the correct final state is a standard basis state, then the trace distance between the final density matrix conditioned on the verifier's acceptance and the correct final state
is at most $\frac {\epsilon} \beta$. 
\end{corol}
The following corollary (which we prove in Appendix \ref{sec:symmetric}) contains another implication of \Th{thm:main}. We show that under a somewhat
stronger assumption than \BQP\ $\ne$ \BPP, but still a widely
believed assumption, it is possible to lower bound the computational power of the prover (and deduce that the prover is not within \BPP) by efficiently testing the prover (assuming the prover passes the test with high probability). 
\begin{corol}\label{corol:bqpsepbpp} 
Assume that there is a language $L \in$ \BQP\
and there is a polynomial time sampleable distribution $D$ on the instances of $L$ on which any \BPP\  machine errs with non negligible probability (e.g. the standard cryptographic assumption about the hardness of Factoring or Discrete Log). If the verifier runs a \QPIP\ (with soundness $\gamma + \epsilon$ and completeness $1 - \gamma$) on $\qc_{\gamma}$ on an instance drawn from $D$ and does not abort with probability $\geq \beta$ (where $\beta - \frac{4\epsilon}{1-2\gamma}$ is a constant), then the prover's
computational power cannot be simulated by a \BPP\ machine.
\end{corol}

\subsection{Proofs Overview}\label{sec:proofsoverview}
As mentioned, we rely on two different quantum authentication schemes (and their proofs of security) in order to derive and prove completeness and soundness of our two \QPIP s. 
\subsubsection{Clifford \QAS}
To prove the security of the Clifford \QAS\ (as stated in Theorem \ref{thm:CliffordAuth}) 
we first prove in \Le{CliffordDiscretization} that any non identity attack of
Eve is mapped by the random Clifford operator to a uniform mixture over all non identity Pauli operators. We call this property of Clifford operators (as stated in \Le{CliffordDiscretization}) operator decohering by Cliffords or in short, {\it Clifford decoherence}. Next, we show that the uniform mixture over Pauli operators changes the 
auxiliary 0 states used in the Clifford authentication scheme
with high probability, and the non identity attack is therefore likely to be detected by Bob's decoding procedure. 
\subsubsection{Clifford \QPIP}
To prove the soundness of the Clifford based \QPIP\ (stated in Theorem \ref{thm:CliffordIP}), we use Clifford decoherence to reduce the soundness of the \QPIP\ to the 
security of the \QAS. We do this by showing in Claim \ref{claim:stateformat} that Clifford decoherence (\Le{CliffordDiscretization}) allows us to shift all attacks of the prover to the end of the protocol (at which point we can simply apply the security proof of the \QAS). This shifting is clearly possible if the prover's attack is the identity operator. If the prover's attack is non identity, Clifford decoherence maps the prover's attack to a uniform mixture over all non identity Pauli operators. It follows that after the verifier decodes the state sent by the prover, the state will essentially be maximally mixed, and whatever the verifier applies at this point commutes with the prover's attack which is currently acting on the state (this is shown in \Le{clifSimplify}). Note that it is important that the verifier does not check the states for correctness in each round (by measuring the auxiliary qubits and checking whether they are 0); instead, he only checks in the final round. If the verifier had instead checked the states for correctness in each round, we could not have used the shifting technique due to Clifford decoherence, since the verifier would be applying a non unitary operator (measurement). We thereby obtain a simple \QPIP, with a rather short proof. 
\\~\\
The key disadvantage of this protocol is the two way quantum communication 
required in each round. 

\subsubsection{Polynomial based \QAS}
To strengthen the results, we use a polynomial based \QPIP\ instead. 
The proof of the security of the corresponding \QAS\ (Theorem \ref{thm:PolynomialAuth}) requires some care,
due to a subtle point which was not addressed in \cite{benor2006smq}.
To prove \Th{thm:PolynomialAuth}, we first prove in \Le{lem:PauliPolyAuthSec} that no non identity Pauli attack can preserve the signed polynomial code 
for more than 2 of the sign keys, and thus the sign key suffices in order 
to protect
against any non identity Pauli attack of Eve's. Next, we need to show that the scheme
is secure against general attacks. This, surprisingly, does not follow
by linearity from the security against Pauli attacks (as is the case
in quantum error correcting codes):
if we omit the Pauli key we get an authentication
scheme which is secure against Pauli attacks but not against general
attacks\footnote{Without Pauli keys, the sign key can be determined up to $\pm 1$ from a measurement of the state. This follows from the uniqueness of signed polynomials, which is proven in Fact \ref{signkeyfact}.}. We proceed by showing (with some similarity to
the Clifford based \QAS) that the random Pauli key
effectively translates Eve's attack to a mixture
(not necessarily uniform like in the Clifford
case) of Pauli operators acting on a state encoded by a random signed
polynomial code. We call this property of random Pauli keys 
{\it Pauli decoherence} (see \Le{PauliDiscretization}). 

\subsubsection{Polynomial based \QPIP}
We proceed to proving the soundness of the polynomial based \QPIP\  
(as stated in Theorem \ref{thm:PolynomialIP}).   
Unlike in the Clifford case, the soundness of the polynomial \QPIP\ cannot be directly reduced to the security of the polynomial \QAS, because the prover's attack cannot be shifted to the end of the protocol in the polynomial \QPIP. This is due to the weakness of Pauli decoherence relative to Clifford decoherence. In more detail, Clifford decoherence first maps the prover's attack to a convex sum over Pauli operators, and then further maps each non-identity Pauli operator to a uniform mixture over all non identity Pauli operators. Pauli decoherence only performs the first step: it maps the prover's attack to a convex sum over Pauli operators (which are weighted according to the original attack). This does not create a maximally mixed state, and therefore does not allow the same shifting of the prover's attack to the end of the protocol. Thus, the proof of \Th{thm:PolynomialIP} does not use the proof of the polynomial \QAS\ (\Th{thm:PolynomialAuth}) as a black box, and in fact that proof is not strictly needed for the proof of Theorem \ref{thm:PolynomialIP}. However, we included \Th{thm:PolynomialAuth} and its proof for 
completeness (as it was not written before), 
and mainly because the two key ideas used in that proof will also be used in the proof of the polynomial based \QPIP. Recall that these two key ideas are Pauli decoherence from Lemma \ref{PauliDiscretization} and security of the sign key against Pauli attacks from Lemma \ref{lem:PauliPolyAuthSec} (which takes up most of the technical effort involved in proving \Th{thm:PolynomialAuth}). 
\\~\\
To prove soundness of the polynomial \QPIP, we first note that if all of the classical messages sent from the verifier to the prover were fixed ahead of time, shifting the prover's attacks to the end of the protocol would be fine, as the verifier's messages to the prover do not depend on the prover's measurement results. Once we shift the prover's attacks, we can apply the two main ideas (Pauli decoherence from Lemma \ref{PauliDiscretization} and security of the sign key from Lemma \ref{lem:PauliPolyAuthSec}) used in the proving the security of the polynomial \QAS\ to obtain soundness of the polynomial \QPIP. However, in the actual polynomial \QPIP\ protocol, the classical interaction does depend on the prover's messages. We employ an idea from \cite{fk2012} (see Figure 7 in their paper): as part of the analysis, we fix the interaction transcript at the start of the protocol (to allow shifting of the prover's attacks) and then project onto this fixed interaction transcript at the end of the protocol to enforce consistency. This technique (which is formalized in \Cl{claim:polystateformat}) essentially partitions the prover's Hilbert space according to the interaction transcript, and we can then apply the two key ideas used in proving the security of the polynomial \QAS\ in each partition. 

\subsubsection{Blindness}
Proving blindness of the Clifford scheme (as stated in \Th{thm:blindclifford}) is quite straightforward and is done by showing that, due to the randomness of the Clifford encoding operator, the prover's state is maximally mixed at all times. This is because applying a random Clifford operator on a state results in a maximally mixed state (see \Le{cliffordmix}). See Section \ref{app:blind} for the full proof. 
\\~\\
Proving blindness of the polynomial scheme (as stated in \Th{thm:blindpolynomial}) is a bit more involved due to the classical interaction - see Section \ref{app:blind}. Without the classical interaction, we could use the randomness of the Pauli keys to show that the prover's state is always maximally mixed (this relies on the fact that applying a random Pauli operator to a state results in a maximally mixed state, as stated in \Le{paulimix}). When we include classical interaction (for the purpose of decoding measurement results), we need to show that the measurement results (even if altered by a malicious prover) do not leak information about the input state. This is due to the fact that measurement results are initially distributed uniformly at random, and even if a malicious prover attacks, his attack can be reduced to a convex sum over Pauli operators (due to \Le{PauliDiscretization}), which preserves the uniform distribution. Note that we could have simplified this proof significantly by including extra randomness in the magic states, which would serve essentially as a one time pad for the decoded measurement results (this would have complicated the description of the polynomial \QPIP\ protocol, Protocol \ref{prot:PolynomialIP}) . However, it is interesting to note that this extra randomness is not needed for blindness, and that the protocol is blind due to the randomness of the measurement results and the Pauli keys.  

\subsubsection{Interpretation}
We prove the corollaries (Corollary \ref{corol:confidence} and Corollary \ref{corol:polyconfidence}) given in Section \ref{sec:interpretation} in Section \ref{app:interpretation}. Both proofs rely on using the format of the prover's state, as shown in \Cl{claim:stateformat} for the Clifford \QPIP\ and \Cl{claim:polystateformat} for the polynomial \QPIP, to first determine what the prover's state will look like conditioned on the verifier's acceptance. In the Clifford case, the trace distance can then be determined quite easily, due to the simplicity of \Cl{claim:stateformat}. The format of the prover's final state in the polynomial case is significantly more involved. The proof of Corollary \ref{corol:polyconfidence} proceeds by analyzing the effect of two different types of Pauli attack operators (trivial and non trivial) in order to show that the trace distance between the final state after acceptance and the correct final state is correlated to the probability of acceptance and the security parameter.

\subsection{Changes from Conference Version}
This journal version is a corrected and elaborated version of the 
conference version, and a new author (U.M) was added.  We describe here in detail the differences from the conference version. 
\subsubsection{Soundness}\label{sec:changessoundness}
\paragraph{Clifford Scheme}
The main difference between the Clifford \QPIP\ protocol in this version and theconference version is that in this version the verifier checks correctness (by measuring the auxiliary 0 states) only in the final round, whereas in the conference version, the verifier checked correctness each time he received qubits from the prover. The conference version of the protocol is actually not sound; the prover can cheat by deviating only slightly in each round. Since there are polynomially many rounds, this can add up to a significant deviation in the 
final state, without being detected, using essentially the 
zeno effect. 
\\~\\
The security proof in the original version assumed that all attacks of the provercould be shifted to the end of the protocol; namely, 
that the prover only deviated at the end of the protocol. 
While this does not hold in the original protocol, in the new scheme this 
can be proven, which is what makes the proof go through. 
The final proof eventually turns out to go along similar lines to 
the original one, except for this change in the protocol.   

\paragraph{Polynomial Scheme}
The protocol for the polynomial \QPIP\ remained the same. However, the security proof needed to be changed dramatically, due to the same incorrect assumption used in the Clifford \QPIP\ regarding shifting the prover's attacks to the end of the protocol (as mentioned above). Whereas in the Clifford scheme a minor change in the protocol sufficed to guarantee that this assumption actually holds, we did not have such a simple solution in the polynomial scheme. As described in Section \ref{sec:proofsoverview}, this is because the weakness of Pauli decoherence relative to Clifford decoherence: Pauli decoherence does not map 
the attack to a {\it uniform} mixture over Paulis, 
thus preventing shifting the prover's attacks to the 
end of the protocol. The polynomial \QPIP\ proof therefore required major 
revisions, because we could not simply reduce to the security of the polynomial \QAS.

\subsubsection{Fault Tolerance}\label{sec:previousfaulttolerance}
In the original version of the paper a scheme for making the 
protocol fault tolerant was proposed and was claimed to be secure. 
Unfortunately, there is a fatal flaw in the proof; we retract the claim about 
fault tolerance (see open questions in Section \ref{sec:openquestions} for possible approaches left for future work).  
\\~\\
We describe below the proposal for fault tolerance of \cite{abe2008}
and the bug. The proposed protocol was: at the first stage of the protocol,
authenticated qudits are sent from the verifier to
the prover, one by one. As soon as the prover receives an
authenticated qudit, he protects his qudits using his own
concatenated error correcting codes so that the effective error
in the encoded authenticated qudit
is constant. This constant accuracy can be maintained for a long time by
the prover, by performing error correction with respect to
{\it his} error correcting code (see \cite{aharonov1997ftq}). Thus, polynomially many such authenticated states can be passed to the prover
in sequence. A constant effective error is not good enough,
but can be amplified to an arbitrary inverse polynomial by
purification. Indeed, the prover cannot perform
purification on his own since the purification
compares authenticated qudits and the prover does not know the
authentication code. However, the verifier can help the prover by
using classical communication. This way the prover can reduce the effective
error on his encoded authenticated qudits to inverse polynomial,
and perform the usual fault tolerant construction
of the given circuit, with the help of the verifier in performing
the gates.
\\~\\
The problem with this approach is that the purification protocol could leak information about the sign key; during the purification protocol
the verifier tells the prover which states are good enough for him to keep and which he should throw away. A cheating prover could lie on all of his messages to the verifier; eventually, he will figure out which of his lies will lead to the verifier accepting, and this should give him information about the sign key chosen by the verifier. Once the sign key is no longer hidden from the prover, the \QPIP\ protocol is no longer secure. The problem seems to be difficult. 
In Section \ref{sec:openquestions}   
we describe why several other possible avenues we tried, in order to 
achieve fault tolerance, failed; 
It remains open to achieve blind verifiable \QPIP s in the noisy setting, even when we allow the verifier 
to hold a polylogarithmic quantum register, rather than a constant one. 

\subsection{Related Work}\label{sec:relatedwork}
\paragraph{Related Work in Blindness and Verifiability}
The question of delegated blind computation was asked by Childs in
\cite{childs2001saq} and by Arrighi and Salvail in \cite{blind}, who proposed
schemes to deal with such scenarios.  However \cite{childs2001saq} does not deal
with a cheating prover, so the protocol is not verifiable. Also, the setting is somewhat different; rather than limiting the quantum space of the verifier, the verifier is limited to only performing Pauli gates. In \cite{blind}, Arrighi and Salvail provide a blind interactive quantum protocol in this setting for a restricted set of functions, and prove its security against a restricted set of attacks.

\paragraph{Independent work}
After  deriving  the results  of  the first version of this paper,  
we learned that
Broadbent, Fitzsimons,  and  Kashefi  \cite{broadbent2008ubq}  have claimed
related results. Using measurement based quantum computation, they
construct a protocol for universal blind quantum computation.
In their case, it suffices that the verifier's
register consists of a single qubit. Their results have
similar implications to ours in terms of the \QPIP\ notion, though 
these are implicit in
\cite{broadbent2008ubq}. However, their protocol was not secure against general attacks (as noted in \cite{fk2012}). 
However, based on similar ideas, Fitzsimons and Kashefi 
suggested a measurement based protocol which is both verifiable and blind, 
and prove its security in \cite{fk2012} (a key idea they used to prove 
security was also useful in our proof of the polynomial \QPIP, as described in 
Section \ref{sec:proofsoverview}).

\paragraph{Follow-up Work}
Since the results presented here were first posted \cite{abe2008}, 
together with the \cite{broadbent2008ubq} paper, there had been a surge of
results investigating the notions of blind quantum computation, 
verifiable quantum computation, 
the ability to perform those in a noisy environment fault tolerantly,
as well as experimental demonstrations. 
\\~\\
As mentioned, Fitzimons and Kashefi gave a different  
\QPIP\ protocol which is both verifiable and blind, based on measurement 
based quantum computation \cite{fk2012}.
Our protocol seems to be simpler to state, but the \cite{fk2012} 
has the advantage of only requiring a single qubit at the verifier's end.  
\\~\\
In \cite{broadbent2012}, Broadbent, Gutoski and Stebila provided a framework for analysing blind \QPIP s in the context of one time quantum
programs; a sketch for a proof of the blindness (but not of the verifiability) of a protocol very similar to our polynomial based protocol (Protocol \ref{prot:PolynomialIP}) is given in that paper in Section 6.1. In \cite{mf2016}, Morimae and Fitzsimons proposed a very nice and simple \QPIP\ protocol
which is just verifiable but not blind. Moreover, it requires the verifier only to be able to measure qubits in the standard or Hadamard basis; it is based on the idea of the prover generating the history state known from
Kitaev's QMA proof (\cite{kitaev2002caq}). Additional blind QPIP protocols were proposed in \cite{hm2015}, \cite{morimae2014} and \cite{broadbent2016}. 
\\~\\
A very interesting question which branched out from the results presented 
here, was taken by Reichardt, Unger and Vazirani\cite{ruv2012}. 
In their work, they 
 proposed a protocol in which a \BPP\ verifier could verify a \BQP\ computation by only classical interaction, when interacting with {\it two} \BQP\ 
entangled provers \cite{ruv2012}. Since then, there have been several 
papers which have explored the model of multiple \BQP\ provers and a single \BPP\ verifier (such as \cite{mckague2013}, \cite{gkw2015}, \cite{hpf2015}, \cite{hh2016}). 
\\~\\
The difficulties in providing a fault tolerant blind verifiable 
protocol with only $O(1)$ qubits at the verifier's end 
seem hard to get around. 
There have been several attempts to suggest solutions, 
including the conference version of this paper \cite{abe2008} as well 
as \cite{fk2012, tfmi2016} but to the best of our knowledge this problem 
remains importantly open, even when the verifier is allowed to hold 
a polylogarithmic quantum register. 
\\.\\
As for experimental demonstrations, we mention a few: 
of blind computing in \cite{bkbfzw2012},\cite{grbmw2016} and of 
verifiable computing in \cite{bfkw2013}).

\subsection{Fault Tolerance Open Questions and Attempts}\label{sec:openquestionsfaulttolerance}
Technically, the main open question raised by this work is to 
provide a fault tolerant version of these results.
In work yet to be published (\cite{verifiableftqpip}), it is shown that fault tolerance can 
be achieved if only one of the tasks (blindness or verification) is required.   
However, we do not know how to achieve fault tolerance for both tasks 
simultanuously. Moreover, we do not even know how to do this
when allowing the
verifier a quantum register of polylogarithmic size. There seems to be an inherent problem in any of the straightforward 
approaches to making our schemes fault tolerant, which we now explain.
\\~\\
We already discussed above why the approach 
which we suggested in the original paper, of purification with the 
help of the verifier (see Section \ref{sec:previousfaulttolerance}), failed. Another attempt is to create a fault tolerant version of the Clifford protocol by running a fault tolerant circuit, which involved the prover passing the qubits back to the verifier for correction at every step of the circuit. This seemed to require the verifier to measure and check for errors when correcting, which compromises the soundness of the Clifford protocol as explained in Section \ref{sec:changessoundness} (recall from Section \ref{sec:proofsoverview} that the verifier only checks for errors at the end of the current Clifford protocol).
\\~\\
We also attempted to create a fault tolerant version of the polynomial protocol
by using blind computation to allow the prover to create the authenticated states on his own; the prover can then do everything in his lab fault tolerantly.  
This idea seemed troublesome because the prover did not have to honestly run the blind computation in order to create the authenticated states, and his dishonesty during the state creation phase could potentially compromise security later on in the protocol.
\\~\\
Finally, we attempted to use standard fault 
tolerant techniques (e.g. \cite{aharonov1997ftq}) in order to simulate 
the \QPIP\ protocol, both by the verifier and the prover. 
The protocol will start with the verifier (who now has a quantum register of polylogarithmic size) creating a fault tolerant encoding of
his authenticated states, and sending those to the prover. The prover 
will act as expected by the \QPIP\ protocol, but will keep correcting 
the state with respect to the code used for fault tolerance. 
Unfortunately, we do not yet know how to extend our security proofs to 
hold for this protocol, 
though it may be secure. A natural attempt to prove security  
would be to reduce the security in the noisy case to that of the ideal case.  
In other words, we would like to claim that if the protocol is 
insecure in the noisy setting, then 
the prover can also cheat in the noiseless setting (by simulating the noise). 
Unfortunately, we do not know how to claim that the prover can
simulate the effect of the noise acting 
on the authenticated states, since the noise may depend on the private keys of the verifier (this is because the verifier's circuit to create the authenticated states depends on these keys). 
One might hope to use error correction techniques to remove the dependence 
of the noise on the keys, but this approach turns out to fail due to a 
very subtle issue - namely, due to teleportation-like effects, dependencies on the keys may 
propagate through the error correction to qubits which were previously
subject to errors independent of the keys. Hence, we leave this approach for 
future investigation. 

\subsection{Conclusion and Open Questions}\label{sec:openquestions}
The results presented here introduced the notion of interactive proofs with quantum provers and this journal version provides rigorous proofs of the two \QPIP s presented in \cite{abe2008}. These results show that the fundamental questions regarding the falsifiability of the high complexity regime of quantum mechanics, the ability to delegate quantum computations to untrusted servers, and the ability to test that experimental quantum systems behave as they should can all be done using interactive protocols between a \BQP\ prover and a classical (\BPP) verifier augmented with $O(1)$ qubits.
\\~\\
This work has revolutionary implications in the context of philosophy of science. It suggests that experiments can be conducted in a structured adaptive way, along the lines of interactive proofs \cite{goldwasser1985kci}; this can be called "interactive experiments" and suggests a new approach to confirmation of physical theories. Following discussions with us at preliminary stages of this work,
Jonathan Yaari has studied ``Interactive proofs with Nature''
from the philosophy of science perspective \cite{JonathanThesis}. 
The philosophical 
aspects of this possibility of interactive experiments suggested by 
our \QPIP\ protocols  
were also discussed by Aharonov and Vazirani in \cite{av2012}. A very interesting question is whether interactive experiments can be designed to test conjectured physical theories, even in the absence of full control of the physical system as is required in our protocols. 
A particularly interesting example is high $T_c$ superconductivity, in which guesses regarding the governing Hamiltonian exist. 
It would be extremely interesting to be able to test the correctness of 
the Hamiltonian using such interactive techniques, without resorting to full fledged quantum computational power.  
\\~\\
Perhaps the most important and intriguing open question that emerges from this 
work is whether it is possible to
remove the necessity for even a small quantum register, and
achieve similar results in the more natural \QPIP\ model
in which the verifier is
entirely classical. This would have interesting fundamental implications
regarding the ability of a classical system to learn and test a quantum
system; it is likely that such a protocol might also have 
implications on the major open problem of quantum PCP \cite{qpcp}. 
\\~\\
Finally, we can also ask whether it is possible to achieve blind 
(rather than verifiable) computation, in two different settings. The question of blind computation involves a client who would like to ask a \BQP\ server to run a \BQP\ circuit. The client does not wish to verify the result of the computation, but just to ensure that the server does not learn anything about the computation, even though he is able to  run the computation. If the client is a \BQP\ machine (but does some amount of work which is independent of the size of the computation)
and there is only one round of interaction,
this problem is known as quantum fully homomorphic encryption. While there have been several results exploring this question, such as 
\cite{dss2016}, \cite{bj2014} and \cite{ypf2014} (which is an impossibility result regarding information theoretically secure quantum homomorphic encryption), quantum fully homomorphic encryption remains an open question. We can also change the model slightly by allowing classical interaction and restricting the client to be a \BPP\ machine. This variant also remains open. 
\\~\\
We remark that this area is notorious for the difficulty in providing rigorous protocols and proofs of security, as the arguments involved are very delicate and subtle.  We hope that this journal version makes a useful contribution in this direction. We believe that the techniques presented here will be very useful in the vastly growing area of delegated quantum computation and quantum cryptographic protocols.


\vb\textbf{Paper Organization}
We start with some notations and background in \Sec{Back}. In Section \ref{sec:CliffordAuth}, we present the Clifford \QAS\ and prove its security. In Section \ref{sec:cliffordIP}, we present the Clifford \QPIP\ and prove security. Sections \ref{sec:PolyAuth} and \ref{sec:IPQ} present the polynomial \QAS\ and \QPIP. Blind delegated quantum computation is proved in Section \ref{app:blind}.
The corollaries related to the interpretations of the results
are proven in Section \ref{app:interpretation}. Appendix \ref{sec:symmetric} contains the definition of $\QPIP_{\kappa}^{sym}$ as well as the proofs of Corollary \ref{thm:mainsym} and Corollary \ref{corol:bqpsepbpp}. Appendix \ref{app:backgroundcliffordpauli} contains useful lemmas about Clifford and Pauli operators, Appendix \ref{sec:cliffordtechnical} contains proofs of the technical lemmas required in Sections \ref{sec:CliffordAuth} and \ref{sec:cliffordIP} and Appendix \ref{app:poly} contains proofs of correctness of the logical operators on signed polynomial codes. Finally, in Appendix \ref{app:tables} we provide a notation table; this is especially helpful in reading Section \ref{sec:IPQ}, as we introduce a significant amount of notation in that section. 

\section{Background}\label{Back}
\subsection{Quantum Authentication}\label{sec:authentication}
Quantum authentication is a protocol by which a sender $\mcA$ 
and a receiver $\mcB$ are capable of verifying that the state sent by 
$\mcA$ had not been 
altered while transmitted to $\mcB$.

\subsubsection{Quantum Security}
If $\mcB$ is a quantum machine, we would like our authentication definition to capture the following
two requirements. On the one hand, in the absence of
intervention, the received state should be the same as the sent
state and moreover, $\mcB$ should not abort.
On the other hand, we want that when the adversary does intervene,
then with all but a small probability (or in fact, distance in terms of 
density matrices), 
either $\mcB$ rejects or his received state is the same as that sent by
$\mcA$.

This is formalized below for pure states; one can deduce
the appropriate statement about fidelity of mixed states, or for
states that are entangled to the rest of the world (see
\cite{barnum2002aqm} Appendix B).

\begin{deff} \label{def:qas} (adapted from Barnum et. al. \cite{barnum2002aqm}).
A quantum authentication scheme (\QAS) from $l$ to $m = l+e$ qubits, 
with security $\epsilon$,   
is a pair of polynomial time quantum
algorithms $\mcA$ and $\mcB$ together with a set of 
classical keys $\mcK$ such that:
\begin{itemize}
\item $\mcA$ takes as input a state $\ket{\psi}$ on $l$ qubits and chooses $k\in \mcK$ uniformly at random. $\mcA$ then applies a unitary operator $A_k$ on the state of $m$ qubits $\ket{\psi}\ket{0}^{\otimes e}$ obtaining:
\begin{equation}
A_k(\ket{\psi}\ket{0}^{\otimes e})
\end{equation}
\item $\mcB$ takes as input a state of $m$ qubits 
and a classical key $k \in \mcK$. He applies a unitary operator $B_k$ to the input state to obtain an output state of $m$ qubits. $\mcB$ declares the state valid if the last $e$ qubits of the output state lie in the space $\ket{0}\bra{0}^{\otimes e}$ and declares the state erroneous if the last $e$ qubits lie in the space $\Pi_{ABR} = \mcI - \ket{0}\bra{0}^{\otimes e}$. 
\end{itemize}
We require: 
\begin{itemize}
\item {\bf Completeness}: For all keys $k\in \mcK$, 
\begin{equation}
B_kA_k(\ket{\psi}\bra{\psi}\otimes\ket{0}\bra{0}^{\otimes e}) A_k^\dagger B_k^\dagger = \ket{\psi}\bra{\psi}\otimes
\ket{0}\bra{0}^{\otimes e}
\end{equation}
\end{itemize} 

To quantify soundness, define the projections: 
\begin{eqnarray}\label{authenticationprojections1}
\Pi_1^{\ket{\psi}} & = &\ket{\psi}\bra{\psi} \otimes I^{\otimes e}+ \left(I^{\otimes l}- \ket{\psi}\bra{\psi}\right) \otimes
\Pi_{ABR} 
\end{eqnarray}
\begin{eqnarray}\label{authenticationprojections0}
\Pi_0^{\ket{\psi}} & = &(I^{\otimes l} -
\ket{\psi}\bra{\psi}) \otimes \ket{0}\bra{0}^{\otimes e}
\end{eqnarray}
Then
\begin{itemize}
\item
{\bf Soundness}: For any super-operator $\mcO$ (representing a possible intervention by the adversary), let $\rho_B$ be defined by
\begin{equation}
\rho_B = \frac{1}{|\mcK|}\sum_k B_k(\mcO(A_k(\ket{\psi}\bra{\psi}\otimes\ket{0}\bra{0}^{\otimes e})A_k^\dagger))B_k^\dagger
\end{equation}
Then the quantum authentication scheme is $\epsilon$-secure if:
\begin{equation}
\tr{(\Pi_1^{\ket{\psi}}\rho_B)} \ge 1- \epsilon
\end{equation}
\end{itemize}
\end{deff}

\subsection{Pauli and Clifford Gates in $F_2$}\label{sec:PauliNClifford}
The $n$-qubits Pauli group consists of all elements of the form
$P=P_1\otimes P_2\odots  P_n$ where $P_i \in \{\mcI,X,Y,Z\}$, together with the multiplicative factors $-1$ and $\pm i$. We will use a subset of this group, which we denote as $\mbP_n$, which includes all operators $P = P_1\otimes P_2\odots P_n$ but not the multiplicative factors. 

The Pauli group $\mbP_n$ is a basis to the
matrices acting on n-qubits. We can write any matrix $U$ over a vector space $A\otimes B$ (where $A$ is the space of $n$ qubits)
as $\sum_{P\in \mbP_n}P\otimes U_P$ where $U_P$ is some (not necessarily unitary) matrix on
$B$. 

Let $\mfC_n$ denote the $n$-qubit Clifford group. Recall that it is
a finite subgroup of unitaries acting on $n$ qubits generated by the Hadamard matrix-H, by $K=\left(\begin{array}{ll}
1&0\\0&i\end{array}\right)$, and by controlled-NOT.
The Clifford group is characterized by the property that it maps the
Pauli group $\mbP_n$ to itself, up to a phase $\alpha\in\{\pm 1,\pm i\}$. That is:
$\forall C\in\mfC_n ,  P\in \mbP_n: ~\alpha CPC^\dagger \in \mbP_n$

\begin{fact}\label{fa:randomclifford}\cite{dlt2002}
A random element from the Clifford group on $n$ qubits can be
sampled efficiently by choosing a string $k$ of $poly(n)$ bits uniformly
at random. The map from $k$ to the group element represented as a product
of Clifford group generators can be computed in classical polynomial
time.
\end{fact}

\subsection{Generalized Gates over $F_q$}\label{defgeneralized}
\begin{deff}\label{defpauli}
The generalized Pauli operators over $F_q$ perform the following maps:
\begin{eqnarray}
X \ket{a} &=& \ket{\left(a+1\right)\mod q}\\
Z \ket{a} &=& \omega_q^a\ket{a}\\
Y &=& XZ
\end{eqnarray}
where $\omega_q = e^{2\pi i /q}$ is the primitive q-root of the unity.
\end{deff}
We note that $ZX=\omega_q XZ$. The generalized Pauli group consists of generalized Pauli operators, together with the multiplicative factor $\omega_q$. We use the same notation, $\mbP_n$, for the standard and generalized
Pauli groups, as it will be clear by context which one is being used.

\begin{deff}
For vectors $x,z$ in $F_q^m$, we denote by $P_{x,z}$ the Pauli
operator $Z^{z_1}X^{x_1}\odots Z^{z_m}X^{x_m}$.
\end{deff}
We now define the other generalized gates we will need:
\begin{deff}{\textbf{Generalized Gates}}
\leavevmode
\begin{enumerate}
\item The generalized Fourier transform over $F_q$ performs the following map on $a\in F_q$: 
\begin{eqnarray}\label{deffourier}
F \ket{ a} &\EqDef& \frac 1 {\sqrt q} \sum_b \omega_q^{ab}\ket b
\end{eqnarray}
\item The generalized $r$- variant of the Fourier transform over $F_q$ performs the following map on $a\in F_q$: 
\begin{eqnarray}
F_r \ket{a} &\EqDef& \frac 1 {\sqrt q} \sum_b \omega_q^{rab}\ket b
\end{eqnarray}
\item The generalized CNOT gate, which we denote as \textit{SUM}, performs the following map on $a,b\in F_q$:
\begin{eqnarray}
SUM \ket{ a}\ket{b} &\EqDef& \ket{a}\ket{(a+b) \mod q}
\end{eqnarray}
\item The generalized Toffoli gate $T$ performs the following map on $a,b,c\in F_q$:
\begin{equation}
T\ket{a}\ket{b}\ket{c} \EqDef \ket{a}\ket{b}\ket{c+ab}
\end{equation}
\item The multiplication gate $M_r$ (for $r\in F_q$, $r\neq 0$) performs the following map on $a\in F_q$:
\begin{equation}
M_r \ket{a}\EqDef \ket{ra}
\end{equation}
\item The generalized controlled phase gate, which we denote as \textit{CPG}, performs the following map on $a,b\in F_q$:
\begin{equation}
CPG \ket{a}\ket{b} = \omega_q^{ab}\ket{a}\ket{b}
\end{equation}
\end{enumerate}
\end{deff}

\subsubsection{Toffoli Gate by Teleportation}\label{app:toffoli}
If given a resource state (which we will also refer to as a magic state or a Toffoli state) of the following form:
\begin{equation}
\frac{1}{q}\sum_{a,b\in F_q}\ket{a,b,ab}
\end{equation}
it is possible to apply a Toffoli gate using only Clifford operations and measurements. This can be done as follows. Assume we would like to apply the Toffoli gate to the state $\ket{c,d,e}$, resulting in $\ket{c,d,e + cd}$. We start with the following state:
$$
\frac{1}{q}\sum_{a,b\in F_q}\ket{a,b,ab,c,d,e}
$$
We then perform the following Clifford entangling operations: a SUM gate from register 6 to register 3, inverse sum gates from register 1 to 4 and register 2 to 5, and an inverse Fourier gate on register 6 resulting in:
\begin{equation}\label{beforemeasurement}
\frac{1}{q}\sum_{a,b\in F_q}\ket{a,b,ab} \longrightarrow\frac{1}{\sqrt{q^3}}\sum_{a,b,l\in F_q}\omega^{-l e}\ket{a,b,ab+e,c-a,d-b,l}
\end{equation}
We then measure registers 4,5, and 6 obtaining measurement results $x,y,z$ where $x$ corresponds to register 4, etc.. The renormalized state after measurement on the unmeasured registers (the first three registers) is then: 
\begin{equation}\label{teleportationstateaftermeas}
\omega^{-ze}\ket{c-x,d-y,(c-x)(d-y)+e}
\end{equation}
Then we apply the following correction to the state (on the first three remaining registers):
\begin{equation}\label{cliffordcorrection}
C_{x,y,z} \EqDef T(X^x\otimes X^y\otimes Z^z)T^\dagger = (X^xZ^{-yz}\otimes X^yZ^{-xz}\otimes X^{xy}Z^z) SUM_{1,3}^ySUM_{2,3}^xCPG_{1,2}^{-z}
\end{equation}
where the subscript denotes the registers (the first is the control and second is the target). We note that the above correction involves Toffoli gates, but since they are acting by conjugation on Pauli operators, the expression is actually a Clifford operator. It is easy to check that after applying $C_{x,y,z}$ to the state in equation \ref{teleportationstateaftermeas}, the resulting state is: 
\begin{eqnarray}
\ket{c,d,e+cd}
\end{eqnarray}


\subsection{Conjugation Properties of Generalized Gates}\label{sec:conjugationproperties}
In this section, we describe how the gates above conjugate operators in the Pauli group. We begin with the \textit{SUM} gate. It is easy to check that:
\begin{eqnarray} {\textit{SUM}}\label{sumconjugation}
  (Z^{z_A}X^{x_A}\otimes Z^{z_B}X^{x_B}){\textit{SUM}}^\dagger &=&
  (Z^{z_A-z_B}X^{x_A}\otimes Z^{z_B}X^{x_B+x_A})
\end{eqnarray}
Next, the Fourier gate swaps the roles of the $X$ and $Z$ Pauli operators; i.e. for $r\in F_q$ ($r\neq 0$)
\begin{equation}\label{fourierconjugation}
F_r Z^zX^x F_r^\dagger = X^{-r^{-1}z}Z^{rx}
\end{equation}
Finally, the multiplication gate $M_r$ (again for $r\in F_q$ where $r\neq 0$) has the following conjugation behavior:
\begin{equation}\label{multiplyconjugation}
M_r Z^zX^x M_r^\dagger = Z^{r^{-1}z}X^{rx}
\end{equation}

\subsection{Signed Polynomial Codes}\label{sec:SignedPolynomial}
We first define polynomial codes:
\begin{deff} Polynomial error correction code
  \cite{aharonov1997ftq}. Given $m,d,q$ and $\{\alpha_1,\ldots,\alpha_m\}$ where $\alpha_i$
  are distinct non zero values from $F_q$, the encoding of $a\in F_q$ is
  $\ket{S_a}$
\begin{equation}
\ket{S_a} \EqDef \frac{1}{\sqrt{q^d}}\sum_{f:def(f)\le d ,
f(0)=a}\ket{ f(\alpha_1),{\ldots} , f(\alpha_m)}
\end{equation}
\end{deff}
We use here $m=2d+1$, in which case the code subspace is its own dual.
It is easy to see that this code can detect up to $d$ errors
\cite{aharonov1997ftq}. In this paper, we will be using signed polynomial codes:
\begin{deff}\label{def:SignedPolynomial} {\bf (\cite{benor2006smq})}
  The signed polynomial code with respect
 to a string $k\in\{\pm 1 \}^m$ is defined by:
\begin{equation}
\ket{S_a^k} \EqDef \frac{1}{\sqrt{q^d}}\sum_{f:deg(f)\le d ,
f(0)=a}\ket{k_1\cdot f(\alpha_1){\ldots} k_m\cdot f(\alpha_m)}\nonumber
\end{equation}
\end{deff}
We again use $m=2d+1$. Similar to the polynomial code, the signed polynomial code can detect $d$ errors and is self dual \cite{benor2006smq}. We will require the following encoding circuit:
\begin{deff}\label{def:encodingcircuit}
Let $E_k$ be a unitary operator such that: 
$$
E_k\ket{a}\ket{0}^{\otimes m-1} = \ket{S_a^k}
$$
\end{deff}
$E_k$ is the encoding circuit, which we describe in further detail in Section \ref{sec:encodingcircuit}. We will write $\rho^k$ to denote that a density matrix $\rho$ is encoded with the signed polynomial code with respect to $k$; i.e. if $\rho$ is one qudit, then 
\begin{equation}\label{def:encodedstate}
\rho^k \EqDef E_k (\rho\otimes\ket{0}\bra{0}^{\otimes m-1}) E_k^\dagger
\end{equation}
\subsubsection{Signed Polynomial Code Logical Gates}\label{sec:logicalgates}
For proofs of all claims and lemmas below, see \pen{app:poly}. We first provide the logical $X$ operator:
\begin{claim}\label{claim:logicalx}
For $x\in F_q$ and $k\in\{-1,1\}^m$, the logical $X$ operator $\wt{X}_k^x$ obeys the following identity:
\begin{equation}
\wt{X}_k^x\ket{S_a^k} \EqDef (X^{k_1x}\otimes\cdots\otimes X^{k_mx})\ket{S_a^k}= \ket{S_{a+x}^k}
\end{equation}
\end{claim}
Similarly for logical \textit{SUM}, we consider the transitive
application of controlled-sum, that is a \textit{SUM} operation applied
between the $j$'th register of $\ket{S_a}$ and $\ket{S_b}$.
\begin{claim}\label{claim:logicalsum}
For all $k\in\{-1,1\}^m$, the logical \textit{SUM} operator $\wt{\textit{SUM}}$ obeys the following identity:
\begin{equation}
\wt{\textit{SUM}}\ket{S_a^k}\ket{S_b^k}\EqDef (SUM)^{\otimes m}\ket{S_a^k}\ket{S_b^k} = \ket{S_a^k}\ket{S_{a+b}^k}
\end{equation}
where each SUM gate in the tensor product acts between registers $i$ and $m+i$ for $1\leq i\leq m$. 
\end{claim}

Showing what is the logical Fourier transform on the signed polynomial code requires more work. We need the following lemma:
\begin{lem}\label{inter} For any $m$ distinct numbers $\{\alpha_i\}_1^m$ there
exist \textbf{interpolation coefficients} $\{c_i\}_1^m$ such that \begin{eqnarray}
\sum_{i=1}^m c_if(\alpha_i) = f(0)
\end{eqnarray}
for any polynomial of degree $\le m-1$.
\end{lem}


We are now ready to define the logical Fourier transform.
\begin{claim}\label{claim:fourier} For all $k\in\{-1,1\}^m$, the logical Fourier operator $\wt F$ obeys
  the following identity:
\begin{eqnarray}
\wt{F} \ket{S_a^k} \EqDef F_{c_1}\otimes F_{c_2}\odots F_{c_m}\ket{S_a^k} =& \frac{1}{\sqrt{q}}\sum_b
\omega_q^{ab} \ket{\wt{S_b^k}}
\end{eqnarray}
where $\ket{\wt{S_b^k}}$ is the encoding of $b$ in a signed polynomial code of
degree $m-d$ on $m$ registers.
\end{claim}
Finally, we define the logical $Z$ operator.
\begin{claim}\label{claim:logicalz}
For $z\in F_q$ and $k\in\{-1,1\}^m$, the logical Pauli $Z$ operator $\wt{Z}_k^z$ obeys the following identity:
\begin{equation}
\wt{Z}_k^z\ket{S_a^k} \EqDef (Z^{k_1c_1z}\otimes\cdots\otimes Z^{k_mc_mz})\ket{S_a^k}= \omega_q^{za}\ket{S_a^k}
\end{equation}
\end{claim}
\subsubsection{Signed Polynomial Encoding Circuit}\label{sec:encodingcircuit}
The encoding circuit $E_k$ first applies a Fourier transform on the first $d$ 0 states, and then interpolates to fill in the rest of the state. The following unitary operator performs the interpolation:
\begin{deff}\label{def:interpolationcircuit}
Let $D_k$ be a unitary operator such that
$$
D_k\ket{a}\ket{k_2f(\alpha_2),\ldots,k_{d+1}f(\alpha_{d+1})}\ket{0}^{\otimes d} = \ket{k_1f(\alpha_1),\ldots,k_mf(\alpha_m)}
$$
such that $deg(f)\leq d$ and $f(0)$ = a. 
\end{deff}
Now if $F$ is the generalized Fourier transform (see equation \ref{deffourier}) it is easy to check that 
\begin{equation}\label{equation:encodingcircuit}
E_k = D_k (\mcI\otimes F^{\otimes d}\otimes \mcI)
\end{equation}
We now describe $D_k$ in further detail. 
\begin{claim}\label{decodingoperations}
The operator $D_k$ can be written as a product of \textit{SUM} operators controlled by registers $1,\ldots,d+1$ with target registers $1,d+2,\ldots,m$ and a multiplication operator on the first register. More explicitly:
\begin{equation}
D_k = \prod_iSUM_{i,1}^{h_i(\alpha_1)k_ik_1}(M_{k_1h_0(\alpha_1)}\otimes\mcI^{\otimes m-1})\prod_{i,j}SUM_{i,j}^{h_i(\alpha_j)k_ik_j}\prod_j SUM_{1,j}^{h_0(\alpha_j)k_j}
\end{equation}
where $2\leq i\leq d+1$, $d+2\leq j\leq m$ and for $i'\in\{0,2,...,d+1\}$ and $\alpha_0 = 0$
\begin{equation}
h_{i'}(x) = \prod\limits_{\substack{l\in\{0,2,\ldots,d+1\}\\l\neq i'}} \frac{x - \alpha_l}{\alpha_{i'}-\alpha_l}
\end{equation}
\end{claim}
\begin{proof}
Observe that
\begin{equation}\label{interpolation}
h_0(x) f(0) + \sum_{2\leq i\leq d+1} h_i(x) f(\alpha_i) = f(x) 
\end{equation}
It follows that for all $d+2\leq j\leq m$, register $j$ holds the following value after the controlled sum operations detailed above:
\begin{eqnarray}
k_0f(0)(h_0(\alpha_j)k_0k_j) + \sum_{2\leq i\leq d+1} k_if(\alpha_i)(h_i(\alpha_j)k_ik_j) &=& k_j(h_0(\alpha_j)f(0)+\sum_{2\leq i\leq d+1} h_i(\alpha_j)f(\alpha_i))\\
&=& k_jf(\alpha_j)
\end{eqnarray}
Now we can see that the controlled sum operations have performed the following mapping:
\begin{equation}
\ket{f(0),k_2f(\alpha_2),\ldots,k_{d+1}f(\alpha_{d+1}),0^{d}} \rightarrow \ket{f(0),k_2f(\alpha_2),\ldots,k_mf(\alpha_m)}
\end{equation}
The only thing left to do is map the first register from $f(0)$ to $k_1f(\alpha_1)$. To do this, first multiply the first register by $k_1h_0(\alpha_1)$ by using the multiplication operation $M_{k_1h_0(\alpha_1)}$, where
\begin{equation}
M_{k_1h_0(\alpha_1)}\ket{a} = \ket{k_1h_0(\alpha_1)a}
\end{equation}
The multiplication operator performs the following mapping:
\begin{equation}
 \ket{f(0),k_2f(\alpha_2),\ldots,k_mf(\alpha_m)}\rightarrow  \ket{k_1h_0(\alpha_1)f(0),k_2f(\alpha_2),\ldots,k_mf(\alpha_m)}
\end{equation}
Then apply controlled sum operations from registers $\hat{i}\in\{2,\ldots,d+1\}$ to register 1 $h_{\hat{i}}(\alpha_1)k_ik_1$ times. Due to equation \ref{interpolation}, the value in the first register after these operations is:
\begin{eqnarray}
k_1h_0(\alpha_1)f(0) + \sum_{2\leq i\leq d+2} k_if(\alpha_i)(h_i(\alpha_1)k_ik_1) &=& k_1(h_0(\alpha_1)f(0) + h_i(\alpha_1)f(\alpha_i))\\
&=& k_1f(\alpha_1)
\end{eqnarray}
It follows that the final controlled sum operations have performed the following mapping:
\begin{equation}
\ket{k_1h_0(\alpha_1)f(0),k_2f(\alpha_2),\ldots,k_mf(\alpha_m)} \rightarrow \ket{k_1f(\alpha_1),k_2f(\alpha_2),\ldots,k_mf(\alpha_m)}
\end{equation}
\end{proof}


\section{Clifford Authentication Scheme}\label{sec:CliffordAuth}
We now define a quantum authentication scheme based on Clifford operations. Let $\mcK_m$ be the set of authentication keys, consisting of succinct
  descriptions of Clifford operations in $\mfC_m$ (i.e. Clifford operations on $m$ qubits); these descriptions exist due to Fact \ref{fa:randomclifford}.

\begin{protocol}\textbf{Clifford based \QAS\ }\label{protocol:cliffordqas}:
  Given is a state $\ket\psi$ on $l$ qubits. Let $e\in\mN$ be such that 
$2^{-e}=\epsilon$. 
 We denote $m=l+e$.
We denote by $C_k$ the operator specified by a
  key $k\in \mcK_m$.
\begin{itemize}
\item \textbf{Encoding - $A_k$}: Alice applies $C_k$ on the state $\ket\psi \otimes
\ket{0}^{\otimes e}$.
\item \textbf{Decoding - $B_k$}: Bob applies $C_k^\dagger$ to the received
      state.  Bob measures the $e$ auxiliary registers and declares the state valid
      if they are all $0$, otherwise Bob aborts.
\end{itemize}
\end{protocol}

\begin{statement}{\Th{thm:CliffordAuth}} \textit{The Clifford scheme given in Protocol \ref{protocol:cliffordqas} is a \QAS\ with security
  $\epsilon=2^{-e}$.} 
\end{statement}

\subsection{The Overall Proof of \Th{thm:CliffordAuth}}
The completeness of this protocol is trivial. In the following proof, we show soundness by first showing that \emph{any} attack of Eve can be decomposed into a distribution over Pauli attacks. We then show that averaging over the random Clifford operators maps a Pauli operator to a uniform distribution over Paulis; the effective transformation on the original state is an application of a random
Pauli. These two facts are summarized by \Cl{claim:EveAttack}. We conclude the proof of \Th{thm:CliffordAuth} by showing that any Pauli attack is detected with high probability.

\begin{proofof}{\Th{thm:CliffordAuth}}
  We denote the space of the message sent from Alice to Bob as $M$.  Without
  loss of generality, we can assume that Eve adds to the message a system $E$
  (of arbitrary dimension) and performs a unitary transformation $U$ on the joint system.  We note that there is a representation of $U$ as $\sum_{P\in \mbP_m}P\otimes U_P$, where $P$ acts on the message space $M$, and $U_P$ is not necessarily unitary and acts on the environment $E$. This is because the Pauli matrices
  form a basis for the $2^m\times 2^m$ matrix vector space. We first
  characterize the effect that Eve's attack has on the unencoded message:
  $\ket\psi\otimes{\ket 0}^{\otimes e}$.

\begin{claim}\label{claim:EveAttack}
  Let $\rho=\ket{\psi}\bra{\psi}\otimes{\ket 0}^{\otimes e}$ be the state of Alice before
  the application of the Clifford operator. For any attack
  \mbox{$U=\sum_{P\in\mbP_m}P\otimes U_P$} by Eve, Bob's state after decoding is
  $s\rho + \frac{1-s}{4^m-1} \sum_{P\neq\mcI}P\rho P^\dagger$, where $s = \tr(U_I\rho_EU_I^\dagger)$.
\end{claim}

We proceed with the proof of the theorem. From the above claim we know what
Bob's state after Eve's intervention is and we would like to bound its projection on $P_0^{\ket\psi}$ (defined in equation \ref{authenticationprojections0}):
\begin{eqnarray}\label{tr}
\tr{\Big(\Pi_0^{\ket\psi} \big(s\rho + \frac{1-s}
{4^m -1}\sum_{Q\in\mbP_m\setminus\{\mcI\}}Q\rho
Q^\dagger\big)\Big)}
&= &s\tr(\Pi_0^{\ket\psi}\rho) + \frac{1-s}{4^m-1}
\sum_{Q\in\mbP_m\setminus\{\mcI\}}\tr{(\Pi_0^{\ket\psi}
Q\rho Q^\dagger)}\label{eq:ClifSec}
\end{eqnarray}
By definition of $\Pi_0^{\ket\psi}$ we see that $\tr(\Pi_0^{\ket\psi}\rho)=0$. On
the other hand: $\tr{(\Pi_0^{\ket\psi}Q\rho Q^\dagger)}\leq 1$ when $Q$ does not flip
any auxiliary qubits, and vanishes otherwise.  The Pauli operators that do not
flip auxiliary qubits can be written as $Q'\otimes Q''$ where \mbox{$Q'\in
  \mbP_l$} and \mbox{$Q'' \in \{\mcI,Z\}^{\otimes e}$}. It follows that the
number of such operators is exactly $4^l2^e$. Omitting the identity $\mcI_m$ we
are left with $4^l2^e-1$ operators which are undetected by our scheme. We return
to \Eq{eq:ClifSec}:
\begin{eqnarray}\label{aftertr}
{\ldots}& \le &\frac{(1-s)(4^l2^e-1)}{4^m-1}\\
& \le &\frac{1-s}{2^e} \label{goodProjection}
\end{eqnarray}
The security follows from the fact that $1-s\le 1$, and hence the
projection is bounded by $\frac{1} {2^e}$. This concludes the proof. 
\end{proofof}

We remark that the above proof in fact implies a stronger theorem:
interventions that are very close to $\mcI$ are even more
likely to keep the state in the space defined by $\Pi_1^{\ket\psi}$.

\subsection{Proof of \Cl{claim:EveAttack}}
Let $ U = \sum_{P\in \mbP_m} P\otimes U_P$ be the operator applied
by Eve. We denote $\rho=\ket\psi \bra\psi\otimes\ket {
0}\bra{0}^{\otimes e}$ the state of Alice prior to encoding. Let us now
write $\rho_{Bob}$, the state of Bob's system after decoding and before measuring
the $e$ auxiliary qubits.  For clarity of reading we omit the
normalization factor $|\mfC_m|$ and denote the Clifford operation
applied by Alice (Bob) $C$ ($C^\dagger$):
\begin{eqnarray}
\rho_{Bob}&=& \frac{1}{|\mfC_m|}\tr_E{\Big(\sum_{C\in\mfC_m}(C\otimes \mcI_E)^\dagger
U\left( (C\otimes \mcI_E)\rho (C\otimes \mcI_E)^\dagger \otimes
\rho_E\right)U^\dagger(C\otimes \mcI_E)\Big)} 
\end{eqnarray}
At this point, we require the following lemma which states that a random Clifford conjugating an operator has the effect of decohering (removing the cross terms) of the operator:
\begin{lem}[\bfseries Clifford Decoherence]\label{CliffordDiscretization}
Let $\rho'$ be a density matrix on $m'>m$ qubits and let $U = \sum\limits_{P\in\mbP_m}P\otimes U_P$ be a matrix acting on $\rho'$ by conjugation. Then
$$
\frac{1}{|\mfC_m|}\sum\limits_{C\in\mfC_m} (C\otimes\mcI)^\dagger U (C\otimes\mcI)\rho' (C\otimes\mcI)^\dagger U^\dagger (C\otimes\mcI) = (\mcI\otimes U_I)\rho' (\mcI\otimes U_I)^\dagger + \frac{1}{|\mbP_m|-1}\sum_{P,Q\in\mbP_m\setminus\{\mcI\}} (P\otimes U_Q) \rho' (P\otimes U_Q)^\dagger
$$
\end{lem}
The proof of the above lemma is in Section \ref{cliffordtechnicallemmas}. Applying \Le{CliffordDiscretization}, we have
\begin{eqnarray}
\rho_{Bob} &=& \tr_E[(\mcI\otimes U_I) \rho\otimes\rho_E (\mcI\otimes U_I)^\dagger + \frac{1}{|\mbP_m|-1}\sum_{P,Q\in\mbP_m\setminus\{\mcI\}} (P\otimes U_Q)\rho\otimes\rho_E (P\otimes U_Q)^\dagger]\\
&=&  \rho\cdot\tr(U_I\rho_EU_I^\dagger) + \frac{1}{|\mbP_m|-1}\sum_{P,Q\in\mbP_m\setminus\{\mcI\}} P\rho P^\dagger\cdot\tr(U_Q\rho_EU_Q^\dagger)\\
&=& \rho\cdot\tr(U_I\rho_EU_I^\dagger) + \frac{1}{|\mbP_m|-1}(\sum_{P\in\mbP_m\setminus\{\mcI\}} P\rho P^\dagger)\cdot(\sum_{Q\in\mbP_m\setminus\{\mcI\}} \tr(U_Q\rho_EU_Q^\dagger))
\end{eqnarray}
We now require the following lemma, which allows us to trace out the extra space a prover may use as part of his attack. See Section \ref{cliffordtechnicallemmas} for the proof:
\begin{lem} \label{Decompose} Let $U=\sum_{P\in\mbP_m} P\otimes U_P$ be a
  unitary operator acting on $m'>m$ qubits. For any density matrix $\tau$ acting on the last $m'-m$ qubits:
\begin{equation}\sum_{P\in\mbP_m}
\tr(U_P\tau U_P^{\dagger}) =\tr(\tau) = 1
\end{equation}
\end{lem}
We apply \Le{Decompose} to write Bob's state as:
\begin{equation}\label{attackProfile}
s\rho + \frac{(1-s)}{4^m-1}\sum_{P\in\mbP_m\setminus\{\mcI\}}\left(P\rho P^\dagger
\right)
\end{equation}
for $s=\tr({ U_{\mcI}\rho_E U_{\mcI}^\dagger})$, which concludes the proof of \Cl{claim:EveAttack}.

\subsection{Proofs of Technical Lemmas} \label{cliffordtechnicallemmas}
In this section we prove \Le{CliffordDiscretization} and \Le{Decompose}.
\subsubsection{Proof of \Le{CliffordDiscretization}}
To prove \Le{CliffordDiscretization}, we require the following three lemmata, which we prove in Appendix \ref{sec:cliffordtechnical}.
The lemma below states that applying a random Clifford operator in $\mfC_m$ (by conjugation) to a non identity Pauli operator $P\in\mbP_m$ maps it to a Pauli operator $Q\in\mbP_m$ chosen uniformly over all non-identity Pauli operators:
\begin{lem}[\bfseries Pauli Partitioning by Cliffords]\label{mix} For every
  $P,Q\in\mbP_m\setminus\{\mcI\}$ it holds that : $\left|\left\{C\in\mfC_m |
      C^\dagger PC =Q\right\}\right| = \frac {\left|\mfC_m\right|} {\left|\mbP_m
    \right| -1}= \frac {\left|\mfC_m\right|} {4^m -1}$.
\end{lem}
The following two lemmas describe the effect of conjugating an operator by a random Pauli or Clifford operator. The lemma below (Lemma \ref{pauliTw}) is used in the proof of the lemma after it (Lemma \ref{clifTw}), and since it will be useful also when we handle the polynomial codes based protocol later on, we state it here with generalized Pauli operators: 
\begin{lem}[\bfseries Pauli Twirl] \label{pauliTw} Let $P \ne P'$ be generalized Pauli operators. For
  any density matrix $\rho'$ on $m'>m$ qubits it holds that
  \begin{equation}
\sum\limits_{Q\in \mbP_m} (Q^\dagger P Q\otimes\mcI) \rho' (Q^\dagger (P')^\dagger Q\otimes\mcI) = 0\nonumber
\end{equation}
\end{lem}

\begin{lem}[\bfseries Clifford Twirl] \label{clifTw} Let $P\ne P'$ be Pauli operators. For
  any density matrix $\rho'$ on $m'>m$ qubits it holds that
  \begin{equation}
  \sum\limits_{C\in\mfC_m}(C^\dagger P C\otimes\mcI)\rho' (C^\dagger
    (P')^\dagger C \otimes\mcI)=0\nonumber
    \end{equation}
\end{lem}


We now proceed to the proof of \Le{CliffordDiscretization}:

\begin{proofof}{\Le{CliffordDiscretization}}
We start with
\begin{eqnarray}
\frac{1}{|\mfC_m|}\sum\limits_{C\in\mfC_m} (C\otimes\mcI)^\dagger U (C\otimes\mcI)\rho' (C\otimes\mcI)^\dagger U^\dagger (C\otimes\mcI) &=& 
\frac{1}{|\mfC_m|}\sum\limits_{P,P'\in\mbP_m}\sum\limits_{C\in\mfC_m} (C^\dagger PC\otimes U_P)\rho' (C^\dagger P'C\otimes U_{P'})^\dagger\nonumber
\end{eqnarray}
We use \Le{clifTw} and are left only with $P=P'$
\begin{eqnarray}\label{prev}
{\ldots}  =
\frac{1}{|\mfC_m|}\sum\limits_{P\in\mbP_m}\sum\limits_{C\in\mfC_m} (C^\dagger PC\otimes U_P)\rho' (C^\dagger PC\otimes U_{P})^\dagger
\end{eqnarray}
We first consider the case were $P=\mcI$, then:
\begin{equation}
\frac{1}{|\mfC_m|}\sum_{C\in\mfC_m}(C^\dagger
P C \otimes U_P) \rho' (C^\dagger
P C \otimes U_P)^\dagger =(\mcI\otimes U_I)\rho'(\mcI\otimes U_I)^\dagger
\end{equation}
On the other hand when, $P\ne\mcI$  by \Le{mix}:
\begin{equation}
\frac{1}{|\mfC_m|}\sum\limits_{C\in\mfC_m} (C^\dagger PC\otimes U_P)\rho' (C^\dagger PC\otimes U_{P})^\dagger = \frac {1}{|\mbP_m|-1}\sum_{Q\in\mbP\setminus\{\mcI\}}(Q\otimes U_P)\rho' (Q\otimes U_P)^\dagger
\end{equation}
\end{proofof}

\subsubsection{Proof of \Le{Decompose}}
\begin{proofof}{\Le{Decompose}}
  We analyze the action of $U$ on the density matrix $\frac{1}{2^m}\mcI\otimes\tau$. We first
  notice that $U$ is a trace preserving operator, that is:
$\frac{1}{2^m}\tr(U(\mcI\otimes\tau) U^\dagger)=\frac{1}{2^m}\tr(\mcI\otimes\tau) = \tr{(\tau)}$.
On the other hand it holds that:
\begin{eqnarray}
\frac{1}{2^m}\tr\big(U(\mcI \otimes \tau) U^{\dagger}\big) &=&
\frac{1}{2^m}\sum_{P,P'\in \mbP_m}\tr\big((P\otimes U_P)(\mcI\otimes\tau)
(P'\otimes U_{P'})^{\dagger}\big)\\
&=&
\frac{1}{2^m}\sum_{P,P'\in \mbP_m}\tr\big(P\mcI P'^{\dagger} \otimes U_P\tau U_{P'}^{\dagger}\big)\\
&=&
\frac{1}{2^m}\sum_{P,P'\in \mbP_m}\tr\big(PP'^{\dagger}\big) \tr\big( U_P\tau U_{P'}^{\dagger}\big)
\end{eqnarray}
If $P\ne P'$ then $\tr\big(PP'^{\dagger}\big)=0$, and therefore:
\begin{eqnarray}\begin{aligned}
\ldots &=
\frac{1}{2^m}\sum_{P\in \mbP_m}\tr\big(\mcI\big) \tr\big( U_P\tau U_{P}^{\dagger}\big)\\
&=
\sum_{P\in \mbP_m} \tr\big( U_P\tau U_{P}^{\dagger}\big)
\end{aligned}\end{eqnarray}
It follows that  $\tr(\tau)=\sum_{P\in \mbP_m} \tr( U_P\tau U_{P}^{\dagger})$, which concludes the proof.

\end{proofof}

\section{Quantum Interactive Proofs with Clifford Authentication}\label{sec:cliffordIP}
\begin{protocol} \textbf{Clifford based Interactive Proof for \qc}:
  \label{prot:CliffordIP}Fix a security parameter $\epsilon$.  Given is a
  quantum circuit consisting of two-qubit gates, $U=U_{N}{\ldots}U_1$, acting on
  $n$ input qubits with error
  probability reduced to $\le \gamma$. The verifier chooses $n$ authentication keys $k_1,\ldots,k_n\in\mcK_{e+1}$, where $e = \lceil{\log{\frac 1 \epsilon}}\rceil$. The verifier authenticates the input
  qubits of the circuit one by one using the Clifford \QAS\ ; that is qubit $j$ is authenticated using
  operation $C_{k_j}$ on $e+1$ qubits. The verifier sends the
  authenticated qubits to the prover $\mathds{P}$. In round $i$ (for $1\leq i \leq N$), the verifier asks $\mathds{P}$ to return the qubits on which $U_i$
  will act. The verifier decodes these qubits by applying the inverse Clifford operator. The verifier then applies $U_i$, authenticates the resulting qubits with new authentication keys and the same (unmeasured) auxiliary qubits and sends the qubits to $\mathds{P}$. In round $N+1$, the prover sends the verifier the first authenticatd qubit, which the verifier decodes and rejects if the auxiliary qubits are not valid. The verifier then measures the decoded qubit (which contains the result of the circuit) and accepts or rejects accordingly. 
\label{prot:finalMeasurement} In any
  case that the verifier does not get the correct number of qubits during the protocol he aborts.
\end{protocol}

\begin{statement}{\Th{thm:CliffordIP}}\textit{For $0< \epsilon < 1$ and $\gamma < 1 - \epsilon$, \Prot{prot:CliffordIP} is a $\QPIP_{O(\log(\frac{1}{\epsilon}))}$ with completeness $1-\gamma$ and soundness $\gamma+\epsilon$ for
  $\qc_{\gamma}$.}
  \end{statement}

The quantum communication is linear in the number of gates.  For $\epsilon
=\frac 1 2$, we get $e=1$, and so the verifier uses a register of $4$ qubits (2 per gate).
In fact $3$ is enough, since each of the authenticated qubits can be decoded (or
encoded and sent) on its own before a new authenticated qubit is handled.

\subsection{Overall Proof of \Th{thm:CliffordIP}}
Let us first analyze completeness: 
\begin{claim}[\bfseries Completeness]\label{cl:qcircuitcompleteness}
For any $\gamma > 0$, \Prot{prot:CliffordIP} is a \QPIP\ protocol with completeness $1 - \gamma$ for \qc.
\end{claim}
\begin{proof}
To prove completeness, we assume the prover is honest and we will show that if $x\in L$, the verifier accepts with probability $\geq 1- \gamma$. Since the prover is honest, the state at all times is indeed 
the correctly authenticated state of the circuit. 
Thus, the output qubit is indeed an authentication of a (possibly mixed)
state, which if measured, outputs $1$ with probability $\ge 1-\gamma$ (the error in the circuit is $\le \gamma$). The decoding of the output block received by the verifier 
will thus result in accept with probability $\ge 1-\gamma$. 
\end{proof}
\\~\\
To prove that Protocol \ref{prot:CliffordIP} has soundness $\gamma + \epsilon$, we will first observe that each round is essentially a run of the Clifford authentication scheme (the only difference is that the auxiliary qubits are not measured and checked by the verifier in each round). Let round $i$ be the round in which the verifier applies $U_i$. Each round can be seen as an authentication protocol on the 1 or 2 authenticated qubits requested by $V$; the rest of the qubits are independently authenticated and can therefore be thought of as the extra space of the verifier. Therefore, we can apply the main claim involved in the soundness proof of the authentication scheme (\Cl{claim:EveAttack}) in each round to just the qubits requested by $V$. 
\\~\\
We will use the claim below to prove soundness, and then we will prove the claim by induction:
\begin{claim}[\bfseries Clifford \QPIP\ State Evolution]\label{claim:stateformat}
The state held by the prover at the start of round $i$ (before sending qubits for $U_i$ to the verifier) can be written as:
$$
O^i (C_k \rho_i C_k^\dagger\otimes\rho_{E_i}) (O^{i})^\dagger
$$
where $C_k = C_{k_1}\otimes\cdots\otimes C_{k_n}$, $O^i$ is a unitary, $\rho_i = (U_{i-1}\cdots U_1) \rho (U_{i-1}\cdots U_1)^\dagger$ ($\rho$ is the initial density $n$ qubit matrix), and $\rho_{E_i}$ represents the prover's extra space.
\end{claim}
This is the claim, which, as mentioned in Section \ref{sec:proofsoverview} of the introduction, uses the strong properties of Clifford decoherence (\Le{CliffordDiscretization}) to show that the attack of the prover can be passed through all rounds to the end of the protocol. Using this claim, we will prove soundness: 
\begin{claim}[\bfseries Soundness]\label{cl:qcircuitsoundness}
For any $\epsilon,\gamma > 0$, \Prot{prot:CliffordIP} is a \QPIP\ protocol with soundness $\gamma + \epsilon$ for \qc.
\end{claim}
\begin{proof}
Assume \Cl{claim:stateformat} holds. Before the prover sends the verifier the first authenticated qubit in round $N+1$ (the final round), \Cl{claim:stateformat} implies that his state is:
$$
O^{N+1} (C_k \rho_{N+1} C_k^\dagger\otimes\rho_{E_{N+1}}) (O^{N+1})^\dagger
$$
Now we can average over all of the authentication keys except $k_1$, since they will not be used for decoding. After averaging, the prover's state can be written as:
$$
\frac{1}{|\mfC_{e+1}|^{n-1}}\sum_{k_2,\ldots,k_n} O^{N+1} (C_k \rho_{N+1} C_k^\dagger\otimes\rho_{E_{N+1}}) (O^{N+1})^\dagger 
$$
Here we require the following lemma, which states that random Clifford operators acting on a state turn it into a maximally mixed state: 
\begin{lem}[\bfseries Clifford Mixing]\label{cliffordmix}
For a matrix $\rho$ on spaces $A\otimes B$, where $A$ is the space of $n$ qubits and $n\in\mN$
$$
\frac{1}{|\mfC_n|}\sum_{C\in\mfC_n} (C\otimes\mcI_B) \rho (C\otimes\mcI_B)^\dagger = \frac{1}{2^n}\mcI_A\otimes\tr_A(\rho)
$$
\end{lem}
The proof of this lemma follows from a similar lemma which uses Pauli operators instead of Clifford operators:
\begin{lem}[\bfseries Pauli Mixing]\label{paulimix}
For a matrix $\rho$ on two spaces $A,B$
$$
\frac{1}{\mbP_n}\sum_{P\in\mbP_n} (P\otimes\mcI_B) \rho (P\otimes\mcI_B)^\dagger = \frac{1}{q^n}\mcI_A\otimes\tr_A(\rho)
$$
\end{lem}
The proofs of both Lemma \ref{cliffordmix} and Lemma \ref{paulimix} can be found in \pen{app:backgroundcliffordpauli}. By applying \Le{cliffordmix},we can see that the averaging changes the state in the final round to:
\begin{equation}
 O^{N+1} (C_{k_1}\tr_A(\rho_{N+1}) C_{k_1}^\dagger\otimes\rho_{AE_{N+1}})(O^{N+1})^\dagger
\end{equation}
where
\begin{equation}
\rho_{AE_{N+1}} = \frac{1}{2^{(n-1)(e+1)}}\mcI_A\otimes\rho_{E_{N+1}}
\end{equation}
and $A$ represents the space of all computational qubits other than the first. Now we proceed to write down the state after the verifier's decoding. Namely, the verifier will decode by applying $C_{k_1}^\dagger$, and then we can average over $k_1$, obtaining:
\begin{equation}
\frac{1}{|\mfC_{e+1}|} \sum_{k_1\in\mcK_{e+1}}(C_{k_1}^\dagger \otimes\mcI_{AE_{N+1}})O^{N+1} (C_{k_1}\tr_A(\rho_{N+1}) C_{k_1}^\dagger\otimes\rho_{AE_{N+1}})(O^{N+1})^\dagger (C_{k_1} \otimes\mcI_{AE_{N+1}})
\end{equation}
Now we can directly apply \Cl{claim:EveAttack} to obtain the state after decoding. This is done by replacing $\rho$ in the statement of \Cl{claim:EveAttack} with $\tr_A(\rho_{N+1})$ and replacing $\rho_E$ in the statement with $\rho_{AE_{N+1}}$. We would like to bound the projection of the state on $\Pi_0^{\ket{1}}$:
\begin{equation}\label{eq:cliffordfinalstate}
\tr{\Big(\Pi_0^{\ket{1}} \big(s\tr_A(\rho_{N+1}) + \frac{1-s}
{4^m -1}\sum_{Q\in\mbP_m\setminus\{\mcI\}}Q(\tr_A(\rho_{N+1}))
Q^\dagger\big)\Big)}
\end{equation}
\begin{equation}
= s\tr(\Pi_0^{\ket{1}}\tr_A(\rho_{N+1})) + \frac{1-s}{4^m-1}
\sum_{Q\in\mbP_m\setminus\{\mcI\}}\tr{(\Pi_0^{\ket{1}}
Q(\tr_A(\rho_{N+1})) Q^\dagger)}
\end{equation}
Due to the circuit error $\gamma$, $\tr(\Pi_0^{\ket{1}} \rho_{N+1})\leq\gamma$. The rest of the upper bound is exactly as in the proof of \Th{thm:CliffordAuth} (see the explanation linking equations \ref{tr} and \ref{aftertr}), and the final upper bound is:
\begin{eqnarray}
s\tr(\Pi_0^{\ket{1}}\rho_{N+1}) + \frac{1-s}{4^m-1}
\sum_{Q\in\mbP_m\setminus\{\mcI\}}\tr{(\Pi_0^{\ket{1}}
Q\rho_{N+1} Q^\dagger)}
&\le & s\gamma + \frac{1-s}{2^e}\\
&\le& \gamma + \frac{1}{2^e}
\end{eqnarray}
\end{proof}
\subsection{Proof of \Cl{claim:stateformat}}
We now proceed to the proof of the claim. We will require the following lemma, which we will prove after completing the current proof:
\begin{lem}[\bfseries Unitary Commutation] \label{clifSimplify}
For all unitaries $U$ acting on $A$, a space of $k$ qubits (for $k\in\mN$), and density matrices $\rho$ acting on $A\otimes B$, we have:
$$
\sum_{Q\neq\mcI\in\mbP_k} (UQ\otimes \mcI)\rho(UQ\otimes \mcI)^\dagger = \sum_{Q\neq\mcI\in\mbP_k} (QU\otimes \mcI)\rho(QU\otimes \mcI)^\dagger 
$$
\end{lem}
We prove \Cl{claim:stateformat} by induction. The base case is clear. For the inductive step, we assume the claim holds for round $i$ and show that it holds for round $i+1$. When the verifier requests the qubits needed for $U_i$, the prover sends back register $\mathcal{R}_i$, which contains the authenticated qudits required to apply $U_i$. Assume $\mathcal{R}_i$ is the first register of the state written below and that it contains 2 authenticated qudits. Then the prover sends back $\mathcal{R}_i$ from the state as given in the inductive step: 
\begin{equation}
O^{i} (C_k \rho_i C_k^\dagger\otimes\rho_{E_i}) (O^{i})^\dagger
\end{equation}
We now write the state after the verifier decodes register $\mathcal{R}_i$ (and after averaging over the Clifford keys for register $\mathcal{R}_i$):
\begin{equation}
\frac{1}{|\mfC_{e+1}|}\sum_{k_1,k_2\in\mcK_{e+1}} (C_{k_1}\otimes C_{k_2}\otimes\mcI)^\dagger O^{i} (C_k \rho_i C_k^\dagger\otimes\rho_{E_i}) (O^{i})^\dagger (C_{k_1}\otimes C_{k_2}\otimes\mcI)
\end{equation}
Next, we can decompose $O^{i}$ as $\sum_{P\in\mbP_{|\mathcal{R}_i|}}P\otimes O_P^{i}$, where $P$ acts on register $\mathcal{R}_i$ and $O_P^{i}$ acts on all other qubits (i.e. the remaining computational qubits as well as the extra space of the prover). Applying \Le{CliffordDiscretization}, we can write the state as:
\begin{equation}\label{clifforddecodedstate}
(\mcI \otimes O_\mcI^{i}) C_{k'}(\rho_i\otimes\rho_{E_i}) C_{k'}^\dagger(\mcI \otimes O_\mcI^{i})^\dagger
+ \frac{1}{|\mathbb{P}_{|\mathcal{R}_i|}| -1}\sum\limits_{P,Q\neq I\in\mathbb{P}_{|\mathcal{R}_i|}} (Q\otimes O_P^{i}) C_{k'}(\rho_i\otimes\rho_{E_i})C_{k'}^\dagger (Q\otimes O_P^{i})^\dagger
\end{equation}
where $C_{k'} = \mcI \otimes C_{k_3}\otimes\cdots\otimes C_{k_n}\otimes\mcI_E$. The verifier then applies the gate $U_i$ and authenticates register $\mathcal{R}_i$ with fresh keys $C_{\hat{k}_1}\otimes C_{\hat{k}_2}$. Let $U_i' = (C_{\hat{k}_1}\otimes C_{\hat{k}_2})U_i$; this is the operation applied to the decoded state (given in equation \ref{clifforddecodedstate}) by the verifier. First observe the action of this operation on the first term of the decoded state; it is now:
\begin{equation}
(\mcI \otimes O_\mcI^{i}) C_{\hat{k}}(\rho_{i+1}\otimes\rho_{E_i}) C_{\hat{k}}^\dagger(\mcI \otimes O_\mcI^{i})^\dagger
\end{equation}
where $C_{\hat{k}} = (C_{\hat{k}_1}\otimes C_{\hat{k}_2}\otimes\mcI) C_{k'}$. To determine what happens to the second term of the decoded state in equation \ref{clifforddecodedstate} after the verifier applies $U_i'$, we will apply the unitary commutation lemma, \Le{clifSimplify}. We first write the second term (after application of $U_i'$) in a way which makes it easier to apply \Le{clifSimplify}:
\begin{equation}
\frac{1}{|\mathbb{P}_{|\mathcal{R}_i|}| -1}\sum\limits_{P,Q\neq I\in\mathbb{P}_{|\mathcal{R}_i|}} (\mcI\otimes O_P^i)(U_i'Q\otimes \mcI) C_{k'}(\rho_i\otimes\rho_{E_i})C_{k'}^\dagger (U_i'Q\otimes \mcI)^\dagger (\mcI\otimes O_P^i)^\dagger
\end{equation}
Now we apply \Le{clifSimplify}, obtaining that the above expression is equal to:
\begin{equation}
\frac{1}{|\mathbb{P}_{|\mathcal{R}_i|}| -1}\sum\limits_{P,Q\neq I\in\mathbb{P}_{|\mathcal{R}_i|}} (\mcI\otimes O_P^i)(QU_i'\otimes \mcI) C_{k'}(\rho_i\otimes\rho_{E_i})C_{k'}^\dagger (QU_i'\otimes \mcI)^\dagger (\mcI\otimes O_P^i)^\dagger
\end{equation}
\begin{eqnarray}
 &=& \frac{1}{|\mathbb{P}_{|\mathcal{R}_i|}| -1}\sum\limits_{P,Q\neq I\in\mathbb{P}_{|\mathcal{R}_i|}} (Q\otimes O_P^i)(U_i'\otimes \mcI) C_{k'}(\rho_i\otimes\rho_{E_i})C_{k'}^\dagger (U_i'\otimes \mcI)^\dagger (Q\otimes O_P^i)^\dagger\\
&=& \frac{1}{|\mathbb{P}_{|\mathcal{R}_i|}| -1}\sum\limits_{P,Q\neq I\in\mathbb{P}_{|\mathcal{R}_i|}} (Q\otimes O_P^{i}) C_{\hat{k}}(\rho_{i+1}\otimes\rho_{E_i}) C_{\hat{k}}^\dagger (Q\otimes O_P^{i})^\dagger
\end{eqnarray}
It follows that the entire state is:
\begin{equation}
(\mcI \otimes O_\mcI^{i}) C_{\hat{k}}(\rho_{i+1}\otimes\rho_{E_i}) C_{\hat{k}}^\dagger (\mcI \otimes O_\mcI^{i})^\dagger
+ \frac{1}{|\mathbb{P}_{|\mathcal{R}_i|}| -1}\sum\limits_{P,Q\neq I\in\mathbb{P}_{|\mathcal{R}_i|}} (Q\otimes O_P^{i}) C_{\hat{k}}(\rho_{i+1}\otimes\rho_{E_i}) C_{\hat{k}}^\dagger (Q\otimes O_P^{i})^\dagger
\end{equation}
 We require one last observation: as it stands, the above state consists of a superoperator acting on $C_{\hat{k}}(\rho_{i+1}\otimes\rho_{E_i}) C_{\hat{k}}^\dagger$. To see that the above operation is a superoperator, note that it was obtained by conjugating a unitary by Clifford operators, and then averaging over the Clifford operators. By expanding the extra space from $\rho_{E_i}$ to $\rho_{E_{i+1}}$, we can instead assume we have a unitary $O^{i+1}$ acting on $C_{\hat{k}}(\rho_{i+1}\otimes\rho_{E_{i+1}}) C_{\hat{k}}^\dagger$.

\subsubsection{Proof of \Le{clifSimplify} (Unitary Commutation)}
We obtain the following equality from \Le{paulimix}:
$$
\frac{1}{|\mbP_k|}\sum_{Q\neq\mcI\in\mbP_k} (Q\otimes\mcI) \rho(Q\otimes\mcI)^\dagger = \frac{1}{2^k}\cdot\mcI\otimes\tr_A(\rho) - \frac{1}{|\mbP_k|}\rho
$$
where $\rho$ is a matrix on spaces $A,B$ and $Q$ acts on $A$. We have:
\begin{eqnarray}
\frac{1}{|\mbP_k|}\sum_{Q\neq\mcI\in\mbP_k} (UQ\otimes \mcI)\rho(UQ\otimes \mcI)^\dagger &=&
(U\otimes\mcI)(\frac{1}{2^k}\cdot\mcI\otimes\tr_A(\rho) - \frac{1}{|\mbP_k|}\rho) (U\otimes\mcI)^\dagger\\
&=& \frac{1}{2^k}\cdot\mcI\otimes\tr_A(\rho) - \frac{1}{|\mbP_k|}(U\otimes\mcI)\rho (U\otimes\mcI)^\dagger\\
&=& \frac{1}{2^k}\cdot \mcI\otimes\tr_A((U\otimes\mcI)\rho(U\otimes\mcI)^\dagger) - \frac{1}{|\mbP_k|}(U\otimes\mcI)\rho (U\otimes\mcI)^\dagger\\
&=& \frac{1}{|\mbP_k|}\sum_{Q\neq\mcI\in\mbP_k} (QU\otimes \mcI)\rho(QU\otimes \mcI)^\dagger 
\end{eqnarray}
This concludes the proof.

\section{Signed Polynomial Code Authentication Scheme}\label{sec:PolyAuth}

Let $\mcK_m$ be the set of pairs of Pauli and sign operators which will be used for authentication; i.e. $\mcK_m = \{(k,z,x)|k\in\{-1,1\}^m, x,z\in F_q^m\}$.
\begin{protocol}\label{protocol:polynomialqas} \textbf{Polynomial Authentication protocol }:
Alice wishes to send the state $\ket\psi$ of dimension $q$. She
chooses a security parameter $d$, a code length $m=2d+1$, and selects $k' = (k,z,x)\in\mcK_m$ at random.
\begin{itemize}
\item \textbf{Encoding - $A_{k'}$:}
 Alice applies $E_k$ to $\ket{\psi}\otimes\ket{0}^{\otimes m-1}$ to encode $\ket{\psi}$ using the signed quantum polynomial code of polynomial degree $d$ (see \Def{def:encodingcircuit}).
She then applies the Pauli $Z^zX^x$ defined by $x,z\in F_q^m$ (i.e., for $j \in \{1,..,m\}$ she applies  $Z^{z_j}X^{x_j}$ on the $j$'th qubit).

\item \textbf{Decoding - $B_{k'}$:}
 Bob applies the inverse of $A_{k'}$; he applies $(Z^zX^x)^\dagger$ followed by $E_k^\dagger$. Bob measures the $m - 1$ auxiliary registers and declares the state valid if they are all $0$, otherwise Bob aborts.

\end{itemize}
\end{protocol}
\begin{statement}{\Th{thm:PolynomialAuth}}\textit{The polynomial authentication scheme as described in Protocol \ref{protocol:polynomialqas}
is a \QAS\ with security $\epsilon = 2^{-d}$.}
\end{statement}

\subsection{The Overall Proof of \Th{thm:PolynomialAuth}}
The completeness of this protocol is trivial. We proceed to prove 
the security of the protocol. As in the proof of \Th{thm:CliffordAuth}, we first show that any intervention made by the adversary can be broken down into a distribution over generalized Pauli interventions. This is given by the Pauli decoherence lemma, \Le{PauliDiscretization}, which is the weaker analogue of the Clifford decoherence lemma (\Le{CliffordDiscretization}). We then state and prove the sign key security lemma (\Le{lem:PauliPolyAuthSec}), which states that non identity Pauli interventions by the adversary on states authenticated with the signed polynomial code are detected with high probability. We note that the main difference between the polynomial and Clifford \QAS\ proofs is the ease in which we can prove the security of each \QAS\ against Eve's non trivial Pauli interventions. In the Clifford case, this is easy because a random Pauli will change the auxiliary 0 states with high probability (see the explanation between equation \ref{tr} and equation \ref{aftertr} in the proof of \Th{thm:CliffordAuth}). In the polynomial case, the sign key security lemma (\Le{lem:PauliPolyAuthSec}) requires quite a few technical details. 
\\~\\
\begin{proofof}{\Th{thm:PolynomialAuth}}
We denote $\rho=\ket\psi \bra\psi\otimes\ket {
0}\bra{0}^{\otimes e}$ the state of Alice prior to encoding. Let $U$ be the attack made by Eve on the joint system, including the message space $M$ and Eve's environment $E$. Bob's state prior to measuring but after applying the decoding operators is:
\begin{eqnarray}\label{polynomialstatebefore}
\rho_{Bob} &=&  \frac{1}{2^m|\mbP_m|}\tr_E{\Big(\sum_{\substack{Q\in\mbP_m\\k\in\{-1,1\}^m}}(QE_k\otimes \mcI_E)^\dagger
U\left( ( QE_k\otimes \mcI_E)\rho \otimes\rho_E(QE_k\otimes \mcI_E)^\dagger
\right)U^\dagger(QE_k\otimes \mcI_E)\Big)^\dagger} 
\end{eqnarray}
At this point we require \Le{PauliDiscretization} (it is analogous to \Le{CliffordDiscretization}), which allows us to reduce general adversary interventions to adversary interventions which are generalized Pauli operators. The lemma states that a random Pauli conjugating an operator has the effect of decohering (removing the cross terms) of the operator (we will prove this lemma in Section \ref{subsection:proofofPauliDiscretization}):
\begin{lem}[\bfseries Pauli Decoherence]\label{PauliDiscretization}
Let $\rho$ be a matrix on $m'>m$ qudits and let $U = \sum\limits_{P\in\mbP_m}P\otimes U_P$ be a matrix acting on $\rho$. Then
$$
\frac{1}{|\mbP_m|}\sum\limits_{Q\in\mbP_m} (Q\otimes\mcI)^\dagger U (Q\otimes\mcI)\rho (Q\otimes\mcI)^\dagger U^\dagger (Q\otimes\mcI) = \sum_{P\in\mbP_m} (P\otimes U_P) \rho (P\otimes U_P)^\dagger
$$
\end{lem}
We decompose the attack $U$
  made by Eve to \mbox{$U=\sum_{P\in \mbP_m}P\otimes U_P$} (where $P$ acts on the space $M$ and $U_P$ acts on the space $E$) and then apply the lemma to equation \ref{polynomialstatebefore} by replacing $\rho$ in the lemma with $E_k\rho E_k^\dagger$, obtaining: 
\begin{eqnarray}
\rho_{Bob} &=& \frac{1}{2^m}\tr_E{\Big(\sum_{P\in\mbP_m,k\in\{-1,1\}^m} (E_k^\dagger PE_k\otimes U_P)\rho \otimes\rho_E(E_k^\dagger PE_k\otimes U_P)^\dagger\Big)} \\
&=& \frac{1}{2^m}\sum_{P\in\mbP_m,k\in\{-1,1\}^m} E_k^\dagger PE_k\rho E_k^\dagger P^\dagger E_k \cdot\tr(U_P\rho_EU_P^\dagger)
\end{eqnarray}
We set $\alpha_P = \tr{\big( U_P^\dagger U_P\rho_{E} \big)} $. Bob's state is now:
\begin{eqnarray}
\cdots &=& \alpha_{\mcI}\cdot \rho + \frac{1}{2^m}\sum_{P\in\mbP_m\setminus\{ \mcI\},k\in\{-1,1\}^m}
\alpha_P\cdot E_k^\dagger P
E_k\rho  E_k^\dagger P^\dagger E_k
\end{eqnarray}
Recall that we are interested in the projection of Bob's state onto $\Pi_0^{\ket{\psi}}$ (defined in equation \ref{authenticationprojections0}), which can now be written as: 
\begin{equation}
\tr(\Pi_0^{\ket{\psi}}\rho_{Bob}) = \alpha_{\mcI}\tr(\Pi_0^{\ket{\psi}}\rho) + \frac{1}{2^m}\sum_{P\in\mbP_m\setminus\{ \mcI\},k\in\{-1,1\}^m}
\alpha_P\tr(\Pi_0^{\ket{\psi}} E_k^\dagger PE_k\rho E_k^\dagger P^\dagger E_k)
\end{equation}
\begin{equation}\label{polynomialfinalsum}
= \frac{1}{2^m}\sum_{P\in\mbP_m\setminus\{ \mcI\},k\in\{-1,1\}^m}
\alpha_P\tr(\Pi_0^{\ket{\psi}} E_k^\dagger PE_k\rho E_k^\dagger P^\dagger E_k)
\end{equation}
Notice that each term in the above sum represents a generalized Pauli attack on the signed polynomial code. We now provide a lemma which states that the signed polynomial code allows detection of adversary interventions which are generalized Pauli operators (we will prove this lemma in Section \ref{subsection:proofofPauliPoly}):
\begin{lem}[\bfseries Sign Key Security]\label{lem:PauliPolyAuthSec} The signed polynomial code is $\frac{1}{2^{m-1}}$-secure against (generalized) Pauli attacks. More formally, for a density matrix $\rho = \ket{\psi}\bra{\psi}\otimes\ket{0}\bra{0}^{\otimes m-1}$ and a generalized Pauli operator $P\in\mbP_m\setminus\{\mcI\}$:
\begin{equation}
\frac{1}{2^m}\sum_{k\in\{-1,1\}^m}\tr(\Pi_0^{\ket{\psi}} E_k^\dagger PE_k\rho E_k^\dagger P^\dagger E_k)\leq\frac{1}{2^{m-1}}
\end{equation}
\end{lem}
We can now use the bound from \Le{lem:PauliPolyAuthSec} on each term in the sum in equation \ref{polynomialfinalsum} to obtain:
\begin{eqnarray}\label{eq:FinalCalc} 
\tr(\Pi_0^{\ket{\psi}}\rho_{Bob}) &=& \frac{1}{2^m}\sum_{P\in\mbP_m\setminus\{ \mcI\},k\in\{-1,1\}^m}
\alpha_P\tr(\Pi_0^{\ket{\psi}} E_k^\dagger PE_k\rho E_k^\dagger P^\dagger E_k)) \\
&\leq&\frac{1}{2^{m-1}}\sum_{P\in\mbP_m\setminus\{ \mcI\}}
\alpha_P\\
 &=& \frac{1 -\alpha_I}{2^{m-1}}\\
 &\leq& \frac{1}{2^{m-1}}
\end{eqnarray}
where the equality follows due to Lemma \ref{Decompose}, which provides the following equality:
\begin{equation}
\sum_{P\in\mbP_m}\alpha_P = 1
\end{equation}
\end{proofof}

Similarly to the random Clifford
authentication scheme, interventions that are very close to $\mcI$ are even more
likely to keep the state in the space defined by $P_1^{\ket\psi}$.

We notice that in this scheme a
$q$-dimensional system is encoded into a system of
dimension $q^m=q^{2d+1}$ and achieves security of $\frac{1}{2^{m-1}}$. The Clifford QAS encodes a 2-dimensional system into a system of dimension $2^{1+e}$ and achieves security of $2^{-e}$. The polynomial scheme is somewhat worse
in parameters (since $q$ must be at least 5), but still with exponentially good security.

To encode several registers, one can independently authenticate
each register as in the Clifford case, but in fact we will see that we can use the same sign key
$k$ for all registers, while still maintaining security.
This property will be extremely useful in applying gates as part of the polynomial QPIP protocol and we will use it in Section \ref{sec:IPQ}. For more details on how gates are applied on top of the signed polynomial code, see Section \ref{sec:logicalgates}.

\subsection{Proof of \Le{PauliDiscretization}}\label{subsection:proofofPauliDiscretization}
We will require the following lemma:
\begin{lem} \label{pauliID}For any two generalized Pauli operators $P$ and $Q$
\begin{eqnarray}
Q^\dagger P Q \rho
Q^\dagger P^\dagger Q = P\rho P^\dagger\nonumber 
\end{eqnarray}
\end{lem}

\begin{proofof}{\Le{pauliID}}
  From the observation about generalized Pauli operators in \Sec{Back} we know
  that for any two generalized Pauli operators $P,Q$ $PQ=\beta QP$ where
  $\beta$ is some phase (of magnitude 1) dependent on $P$ and $Q$.
\begin{equation}
Q^\dagger P Q \rho
Q^\dagger P^\dagger Q = Q^\dagger(\beta Q P)  \rho (\beta^* P^\dagger Q^\dagger) Q
= P\rho  P^\dagger
\end{equation}
\end{proofof}

We can now proceed to the proof:

\begin{proofof}{\Le{PauliDiscretization}}
We start with:
$$
\frac{1}{|\mbP_m|}\sum_{Q,P,P'\in\mbP_m} (Q\otimes\mcI)^\dagger (P\otimes U_P)(Q\otimes\mcI)\rho (Q\otimes\mcI)^\dagger (P'\otimes U_{P'})^\dagger (Q\otimes\mcI)
$$ 
We regroup elements to write the above expression as
\begin{eqnarray}
\ldots &=& \frac{1}{\mbP_m}\sum_{Q,P,P'\in\mbP_m} (\mcI\otimes U_P)(Q^\dagger PQ \otimes \mcI) \rho (Q^\dagger P'Q \otimes \mcI)^\dagger (\mcI\otimes U_{P'})^\dagger
\end{eqnarray}
We use \Le{pauliTw} and are left only with $P=P'$
\begin{eqnarray}\label{prev2}
{\ldots}  &=&
\frac{1}{\mbP_m}\sum_{P,Q\in\mbP_m} (Q^\dagger PQ \otimes U_P) \rho (Q^\dagger PQ \otimes U_P)^\dagger
\end{eqnarray}
Now we use \Le{pauliID} :
\begin{eqnarray}
{\ldots}  =
\sum_{P\in\mbP_m} (P \otimes U_P) \rho (P \otimes U_P)^\dagger
\end{eqnarray}
\end{proofof}

\subsection{Proof of \Le{lem:PauliPolyAuthSec}}\label{subsection:proofofPauliPoly}
In this section, we will prove \Le{lem:PauliPolyAuthSec} (security against Pauli attacks due to the sign key). 

\begin{proof}
Our goal is to show for a density matrix $\rho = \ket{\psi}\bra{\psi}\otimes\ket{0}\bra{0}^{\otimes m-1}$ and a generalized Pauli operator $P\in\mbP_m\setminus\{\mcI\}$:
\begin{equation}\label{mainequation}
\frac{1}{2^m}\sum_{k\in\{-1,1\}^m}\tr(\Pi_0^{\ket{\psi}} E_k^\dagger PE_k\rho E_k^\dagger P^\dagger E_k)\leq\frac{1}{2^{m-1}}
\end{equation}
Throughout this proof, we will ignore phases which come about from Cliffords conjugating Pauli operators or moving Pauli operators past each other. This is due to the format of $\rho_B$; whenever we manipulate $P$ to obtain a phase $\omega_q^a$, we obtain $\omega_q^{-a}$ by manipulating $P^\dagger$ in the same manner. We first need to develop tools to understand how a generalized Pauli attack affects a signed polynomial state. This is done in the following subsection, after which we will return to the proof of \Le{lem:PauliPolyAuthSec}.

\subsubsection{$k$-Correlated Pauli Operators}
We begin with definitions and their corresponding properties, and then proceed to analyze how generalized Pauli operators affect a signed polynomial state.
\paragraph{Definitions and Properties}
We now define what a correlated Pauli operator is:
\begin{deff}\label{deff:correlatedpauli}
For a sign key $k\in\{-1,1\}^m$, we will call a non identity Pauli operator $Q$ $k$-\textit{correlated} if there exist one qudit states $\ket{\psi}$ and $\ket{\phi}$:
\begin{equation}
QE_k\ket{\psi}\otimes\ket{0}^{m-1} = E_k\ket{\phi}\otimes\ket{0}^{m-1}
\end{equation}
\end{deff}
In other words, $Q$ maps a state encoded with the signed polynomial code to another state with the same encoding and therefore cannot be detected. We will show that a non identity generalized Pauli operator $Q$ can be $k$-correlated for at most 2 sign keys $k$ according to the above definition. We will then show that for all sign keys $k$, all non identity Pauli operators $Q$ which are \textit{not} $k$-correlated can be written as a product of a $k$-correlated Pauli operator $Q_k$ and a non identity, uncorrelated operator $\hat{Q}_k$ of a specific form. $\hat{Q}_k$ will always be detected by $B$'s decoding procedure, as it will change the auxiliary states. This implies that a non identity Pauli operator will be caught with probability $\frac{1}{2^{m-1}}$ (it will be caught for all but at most two sign keys for which it is $k$-correlated). 

Next, we describe what a $k$-correlated Pauli $X$ operator looks like. We will require the following fact.
\begin{fact}\label{signkeyfact}
For $k,\hat{k}\in\{-1,1\}^m$ (where $k\neq\hat{k}$), there exist polynomials $f,g$ of degree at most $d$ such that 
\begin{equation}
(k_1f(\alpha_1),\ldots,k_mf(\alpha_m)) = (\hat{k}_1g(\alpha_1),\ldots,\hat{k}_mg(\alpha_m))
\end{equation}
only if $k = -\hat{k}$. 
\end{fact}
\begin{proof}
There must either be at least $d+1$
indices on which $k$ and $\hat{k}$ agree 
or at least $d+1$ indices on which they differ. First consider the case in which there are at least $d+1$ indices where they agree. Since the values of $k$ and $\hat{k}$ at these indices uniquely define $f$ and $g$, $f$ and
$g$ must be equal. It follows that $k_i = \hat{k}_i$ for all $i$. If 
we instead consider the case when $k$ and $\hat{k}$ differ on at least $d+1$ indices, we obtain
that $f$ is equal to $-g$ and therefore $k_i =-\hat{k}_i$. 
\end{proof}
\begin{claim}\label{correlatedx}
A non identity Pauli operator $X = X^x$ is $k$-correlated if and only if it has the following form:
$$
X^x = \beta X^{k_1f(\alpha_1)}\otimes\cdots\otimes X^{k_mf(\alpha_m)}
$$
where $\beta$ is a phase with $|\beta|^2 = 1$, $f$ is a polynomial of degree at most $d$. The Pauli operator $X^x$ can be $k$-correlated for at most 2 sign keys $k$. 
\end{claim}
\begin{proof}
It follows by Definition \ref{deff:correlatedpauli} that
\begin{equation}
X^{k_1f(\alpha_1)}\otimes\cdots\otimes X^{k_mf(\alpha_m)}
\end{equation}
is $k$-correlated. We will now show that if $X^x$ is $k$-correlated it must have the form above. An $X$ Pauli operator can only be $k$-correlated if it adds a low degree polynomial signed with $k$ to the encoded state it is acting on. Therefore, if it is $k$-correlated, it must equal
\begin{equation}
X^{k_1f(\alpha_1)}\otimes\cdots\otimes X^{k_mf(\alpha_m)}
\end{equation}
for a polynomial $f$ of degree at most $d$. We now show that a non identity Pauli operator $X^x$ can be $k$-correlated for at most 2 sign keys. Assume now that $X$ is also $k'$-correlated. By the argument above, it follows that it must equal
\begin{equation}
X^{k'_1g(\alpha_1)}\otimes\cdots\otimes X^{k'_mg(\alpha_m)}
\end{equation}
for a polynomial $g$ of degree at most $d$. However, Fact \ref{signkeyfact} implies that either $k = k'$ or $k = -k'$. 
\end{proof}

Next, we describe what a $k$-correlated Pauli $Z$ operator looks like:
\begin{claim}\label{correlatedz}
A non identity Pauli operator $Z = Z^z$ is $k$-correlated if and only if it has the following form:
$$
Z^z = \beta Z^{c_1k_1f(\alpha_1)}\otimes\cdots\otimes Z^{c_mk_mf(\alpha_m)}
$$
where $\beta$ is a phase with $|\beta|^2 = 1$, $f$ is a polynomial of degree at most $d$ and $c_i$ is the interpolation coefficient defined in \Le{inter} with the following property:
$$
\sum_{1\leq i\leq m}c_if(\alpha_i) = f(0)
$$
The Pauli operator $Z^z$ can be $k$-correlated for at most 2 sign keys $k$. 
\end{claim}
\begin{proof}
We first show that if $Z^z$ is $k$-correlated it must also have the form above. Assume that $Z^z$ for $z\in F_q^m$ is $k$-correlated. Then it follows from Definition \ref{correlatedpauli} that there exist one qudit states $\ket{\psi}$ and $\ket{\phi}$ such that:
\begin{equation}\label{equationfromdef}
Z^z E_k\ket{\psi}\otimes\ket{0}^{m-1} = E_k\ket{\phi}\otimes\ket{0}^{m-1}
\end{equation}
Now if we apply a logical Fourier operator (described in Claim \ref{claim:fourier}), the left hand side of the above equation becomes:
\begin{eqnarray}
\tilde{F}Z^z E_k\ket{\psi}\otimes\ket{0}^{m-1} &=&
\tilde{F}Z^z \tilde{F}^{\dagger} \tilde{F} E_k\ket{\psi}\otimes\ket{0}^{m-1}\\
&=& \tilde{F}Z^z \tilde{F}^{\dagger} E_kF\ket{\psi}\otimes\ket{0}^{m-1}
\end{eqnarray}
where the last equality follows since $\tilde{F}$ is a logical operator (which is also proven in Claim \ref{claim:fourier}).  The right hand side of equation \ref{equationfromdef} becomes:
\begin{equation}
\tilde{F}E_k\ket{\phi}\otimes\ket{0}^{m-1} = E_kF\ket{\phi}\otimes\ket{0}^{m-1}
\end{equation}
Then (by setting the right and left hand side of the equations equal to each other) we have:
\begin{equation}\label{equationzcorrelation}
\tilde{F}Z^z \tilde{F}^{\dagger} E_kF\ket{\psi}\otimes\ket{0}^{m-1} = E_kF\ket{\phi}\otimes\ket{0}^{m-1}
\end{equation}
Due to the conjugation properties of $\tilde{F}$ (for more details about the conjugation behavior of $\tilde{F}$, see the Fourier description in Section \ref{sec:conjugationproperties}), we have (where $\alpha$ is a phase)
\begin{equation}\label{fourierz}
\tilde{F}Z^z \tilde{F}^{\dagger}= \alpha X^{-c_1^{-1}z_1}\otimes\cdots\otimes X^{-c_m^{-1}z_m}
\end{equation}
and due to equation \ref{equationzcorrelation} we can see that the above operator is $k$-correlated. By Claim \ref{correlatedx}, the $X$ Pauli operator in equation \ref{fourierz} can be $k$-correlated for at most 2 sign keys $k$. It follows that the Pauli operator $Z^z$ can be $k$-correlated for at most 2 sign keys. Finally, if $Z^z$ is indeed $k$-correlated, we can combine the fact that the $X$ operator in equation \ref{fourierz} is $k$-correlated and Claim \ref{correlatedx}, to write $Z^z$ as:
\begin{equation}
Z^{c_1k_1f(\alpha_1)}\otimes\cdots\otimes Z^{c_mk_mf(\alpha_m)}
\end{equation}
for a polynomial $f$ of degree at most $d$. We now need to show the opposite direction: if $Z^z$ can be written as
\begin{equation}
 Z^{c_1k_1f(\alpha_1)}\otimes\cdots\otimes Z^{c_mk_mf(\alpha_m)}
 \end{equation}
 it is $k$-correlated. To see this, we can obtain the following equality from equation \ref{fourierz}
 \begin{equation}
 \tilde{F}Z^z\tilde{F}^\dagger = \alpha X^{-k_1f(\alpha_1)}\otimes\cdots\otimes X^{-k_mf(\alpha_m)}
 \end{equation}
 Since this is a correlated $X$ operator, it follows by Definition \ref{deff:correlatedpauli} that $Z^z$ is a correlated $Z$ operator. 
\end{proof}

We can extend the claims for $X$ and $Z$ Pauli operators to general Pauli operators:
\begin{claim}\label{correlatedpauli}
A non identity Pauli operator can be $k$-correlated for at most 2 sign keys $k$. 
\end{claim}
\begin{proof}
Consider a non identity Pauli operator $Z^zX^x$. In order for it to be correlated, $X^x$ must add a low degree signed polynomial to a state encoded by $E_k$ which it is acting on. This means that $X^x$ is $k$-correlated. It follows by Claim \ref{correlatedx} that if $x\neq 0$, the Pauli operator $Z^zX^x$ can be $k$-correlated for at most 2 sign keys $k$. If $x = 0$, Claim \ref{correlatedz} implies that the Pauli operator $Z^zX^x$ can be $k$-correlated for at most 2 sign keys. 
\end{proof}

\paragraph{Correlation Properties of Generalized Pauli Operators}
Now that we have defined $k$-correlation, we can see how a generalized Pauli operator will behave on a signed polynomial state. We begin by showing that for a fixed sign key $k$, a Pauli operator $Q$ can be broken down into a product of a $k$-correlated Pauli operator and an uncorrelated Pauli operator:
\begin{claim}\label{correlateddecomposition}
Let $k$ be a sign key and $Q\in\mbP_m\setminus\{\mcI\}$ be an uncorrelated Pauli operator. Then $Q = Z^zX^x$ can be written as
\begin{equation}
Q = \hat{Q}_k Q_k
\end{equation}
where $Q_k$ is $k$-correlated and $\hat{Q}_k$ is uncorrelated (and in particular, non identity) and can be written (up to a phase) as:
\begin{equation}\label{uncorrelatedform}
\mcI\otimes Z^{\hat{z}_2}\otimes\cdots Z^{\hat{z}_{d+1}}\otimes X^{\hat{x}_{d+2}}\cdots\otimes X^{\hat{x}_{m}}
\end{equation}
where if $z = 0$, $(\hat{z}_2,\ldots,\hat{z}_{d+1}) = 0^d$ and if $x = 0$,  $(\hat{x}_{d+2},\ldots,\hat{x}_m)= 0^d$. 
\end{claim}
\begin{proof}
Observe that a signed low degree polynomial is determined by $d+1$ points. For a given sign key $k$ and $d+1$ points $y_{i_1},\ldots,y_{i_{d+1}}\in F_q$, where $i_1,\ldots,i_{d+1}\in\{1,\ldots,m\}$, let
\begin{equation}
s_k(y_{i_1},\ldots,y_{i_{d+1}}) = (k_1f(\alpha_1),\ldots,k_mf(\alpha_m))\in F_q^m
\end{equation}
be the signed polynomial that is obtained by interpolating the $d+1$ points $y_{i_1},\ldots,y_{i_{d+1}}$. Let: 
\begin{equation}
[s_k'(y_{i_1},\ldots,y_{i_{d+1}})]_i = c_i\cdot[s_k(y_{i_1},\ldots,y_{i_{d+1}})]_i
\end{equation}
For $Q = Z^zX^x = Z^{z_1}X^{x_1}\otimes\cdots\otimes Z^{z_m}X^{x_m}$, we claim that 
\begin{equation}\label{correlatedpart}
Q_k = (Z^{s_k'(c_1^{-1}z_1,c_{d+2}^{-1}z_{d+2},\ldots,c_m^{-1}z_m)})(X^{s_k(x_1,\ldots,x_{d+1})})
\end{equation}
is $k$-correlated. Claims \ref{correlatedx} and \ref{correlatedz} imply that both the $Z$ and $X$ operators of $Q_k$ are $k$-correlated. It follows by the definition of a $k$-correlated operator (Definition \ref{deff:correlatedpauli}) that the product of two $k$-correlated operators ($Q_k$ in this case) is $k$-correlated. Now we define $\hat{Q}_k$ such that:
\begin{equation}
\hat{Q}_k = Q Q_k^\dagger
\end{equation}
It can be readily checked, using the definition of $Q_k$ in equation \ref{correlatedpart}, that $\hat{Q}_k$ is of the following form (up to a phase):
\begin{equation}\label{equation:uncorrelatedform}
\hat{Q}_k \equiv \mcI\otimes Z^{\hat{z}_2}\otimes\cdots Z^{\hat{z}_{d+1}}\otimes X^{\hat{x}_{d+2}}\cdots\otimes X^{\hat{x}_{m}}
\end{equation}
Observe that $\hat{Q}_k$ as written in the above equation is uncorrelated. If it was $k$-correlated, then $Q$ would be a product of two $k$-correlated operators ($Q_k$ and $\hat{Q}_k$) which implies that $Q$ is also $k$-correlated, which contradicts our starting assumption that $Q$ is uncorrelated. Observe also that the last line of the claim (about the implication of $z = 0$ or $x = 0$) follows immediately from the definitions of $Q_k$ and $\hat{Q}_k$.
\end{proof}

Observe that an operator of the form of $\hat{Q}_k$ is always detected, as it will change the auxiliary qudits:
\begin{claim}\label{uncorrelated}
For all one qudit states $\ket{\psi}$ and $\ket{\phi}$ and a fixed sign key $k\in\{-1,1\}^m$, an uncorrelated operator $\hat{Q}_k$ of the form described in Claim \ref{correlateddecomposition} in equation \ref{uncorrelatedform} satisfies the following equation:
\begin{equation}
\tr(\Pi_0^{\ket{\psi}} E_k^\dagger \hat{Q}_k E_k \ket{\phi}\bra{\phi}\otimes\ket{0}\bra{0}^{\otimes m-1} E_k^\dagger \hat{Q}_k^\dagger E_k) = 0
\end{equation}
\end{claim}
\begin{proof}
We claim the following equality holds up to a phase:
\begin{equation}\label{conjugationproperties}
E_k^\dagger \hat{Q}_k E_k = \mcI\otimes X^{\hat{z}_2}\otimes\cdots X^{\hat{z}_{d+1}}\otimes X^{\hat{x}_{d+2}}\cdots\otimes X^{\hat{x}_{m}}
\end{equation}
Recall (from Section \ref{sec:encodingcircuit}) that
\begin{equation}
E_k = D_k(\mcI\otimes F^{\otimes d}\otimes\mcI)
\end{equation}
The conjugation behavior of $E_k^\dagger$ can be determined by looking at the conjugation properties of Clifford operators (see Section \ref{sec:conjugationproperties}). As a brief description, recall from Section \ref{sec:encodingcircuit} that $E_k^\dagger$ consists of the interpolation circuit ($D_k^\dagger$), which is a series of inverse controlled sum operations and an inverse multiplication operator on the first register (see Claim \ref{decodingoperations}). The final operation in $E_k^\dagger$ is an inverse Fourier transform ($(\mcI\otimes F^{\otimes d}\otimes\mcI)^\dagger$). Using the conjugation properties given in equations \ref{sumconjugation}, \ref{multiplyconjugation}, and \ref{fourierconjugation} in Section \ref{sec:conjugationproperties}, we obtain the following equalities. Inverse sum operations have the following conjugation behavior (up to a phase):
\begin{equation}\label{sumconjugation2}
SUM^\dagger (Z^{z_1}X^{x_1}\otimes Z^{z_2}X^{x_2}) SUM = Z^{z_1+z_2}X^{x_1}\otimes Z^{z_2}X^{x_2-x_1}
\end{equation}
where the \textit{SUM} operator above is controlled by the left register. The inverse multiplication operation, $M_r^\dagger$ (for $r\neq 0$), has the following conjugation behavior (up to a phase):
\begin{equation}
M_r^\dagger (Z^{z}X^x) M_r = Z^{rz}X^{r^{-1}x}
\end{equation}
Fourier operations have the following conjugation behavior (up to a phase):
\begin{equation}
F^\dagger Z^zX^x F = X^zZ^{-x}
\end{equation}
The inverse sum operations in $D_k^\dagger$ have no effect on $\hat{Q_k}$, since the target registers (registers $1,d+2,\ldots,m$) never have a non zero $Z$ coefficient in equation \ref{conjugationproperties} and the control registers (registers $1,\ldots,d+1$) never have a non zero $X$ coefficient in equation \ref{conjugationproperties}. In other words, the coefficients $x_1$ and $z_2$ in equation \ref{sumconjugation2} will be 0. The multiplication operation (which is in between the inverse sum operations) similarly has no effect on $\hat{Q}_k$; it is acting on the first register, for which both the $Z$ and $X$ coefficient are 0 (since it is $\mcI$). The inverse Fourier operation flips the $Z$ operators of registers $2,\ldots,d+1$ to $X$ operators. 

Now, returning to equation \ref{conjugationproperties}, since $\hat{Q}_k$ was not equal to the identity, it will make at least one of the auxiliary qudits nonzero (recall that the auxiliary qudits are contained in registers $2,\ldots,m$).
\end{proof}

\subsubsection{Proof of \Le{lem:PauliPolyAuthSec}}
We now use the concepts we developed in the above section to prove \Le{lem:PauliPolyAuthSec}. To summarize, what we have shown is that for each sign key $k$, a generalized non identity Pauli operator $Q$ can be broken down into a product of a $k$-correlated operator $Q_k$ and an uncorrelated operator $\hat{Q}_k$ (\Cl{correlateddecomposition}). The uncorrelated operator $\hat{Q}_k$ will always be detected (\Cl{uncorrelated}) and will only be non identity for at most 2 sign keys $k$ (\Cl{correlatedpauli}). Therefore, $Q$ can only preserve a signed polynomial state for at most 2 sign keys $k$. 

Recall that we would like to upper bound the following expression:
\begin{equation}
\frac{1}{2^m}\sum_{k\in\{-1,1\}^m}\tr(\Pi_0^{\ket{\psi}} E_k^\dagger PE_k\rho E_k^\dagger P^\dagger E_k)
\end{equation}
By Claim \ref{correlatedpauli}, $P$ can be $k$-correlated for at most 2 sign keys $k$. Consider one $k$ in the above sum for which $P$ is not $k$-correlated. We can now apply Claim \ref{correlateddecomposition}, to obtain that the term is equal to
\begin{eqnarray}
\cdots &=& \tr(\Pi_0^{\ket{\psi}} E_k^\dagger \hat{P}_kP_kE_k \ket{\psi}\bra{\psi}\otimes\ket{0}\bra{0}^{\otimes m-1} E_k^\dagger (\hat{P}_kP_k)^\dagger E_k) \\
&=& \tr(\Pi_0^{\ket{\psi}} E_k^\dagger \hat{P}_kE_k \ket{\psi_{P_k}}\bra{\psi_{P_k}}\otimes\ket{0}\bra{0}^{\otimes m-1} E_k^\dagger \hat{P}_k^\dagger E_k )\\
&=& 0 
\end{eqnarray}
where the second equality follows from the fact that $P_k$ is a $k$-correlated operator and the final equality follows from Claim \ref{uncorrelated}. Then we only obtain a non zero expression when $P$ is $k$-correlated. It follows that
\begin{equation}
\frac{1}{2^m}\sum_{k\in\{-1,1\}^m}\tr(\Pi_0^{\ket{\psi}} E_k^\dagger PE_k\rho E_k^\dagger P^\dagger E_k)\leq \frac{1}{2^{m-1}}
\end{equation}

\end{proof}




\section{Quantum Interactive Proofs with Polynomial Authentication}\label{sec:IPQ}
In this section, we give a \QPIP\ for \qc\ (providing another proof for \Th{thm:qcircuit}) using the signed polynomial encoding from the previous section. The key advantage of this protocol is that the prover can perform the gates on top of the encoding without knowing the encoding itself. This means that the prover does not need to hand back the qudits to the verifier in order for the verifier to perform the gates; the prover can perform them on his own. This also means that only one way quantum communication is required (the verifier only needs to send qudits at the start of the protocol, and the rest of the communication is classical). 

The key disadvantage of this protocol is the relative difficulty of proving soundness in comparison to the Clifford \QPIP\ protocol (Theorem \ref{thm:CliffordIP}). This difficulty arises due to the difference between \Le{PauliDiscretization} (Pauli decoherence) and \Le{CliffordDiscretization} (Clifford decoherence). 
The strength of Clifford decoherence allows us to prove \Cl{claim:stateformat} (the Clifford state evolution claim), which states that the prover's state throughout the protocol is the correct authenticated state (i.e. the state with the gates applied as requested by the verifier) with an attack independent of the authentication acting on top of it. Essentially, this is because \Cl{claim:stateformat} uses the unitary commutation lemma, \Le{clifSimplify}, and the Clifford decoherence lemma, \Le{CliffordDiscretization}, to change any logical attack (an attack acting \textit{inside} the authentication) to an attack \textit{outside} of the authentication (which no longer preserves the authenticated state). \Cl{claim:stateformat} then allows us to reduce the soundness of the Clifford \QPIP\ to the security of the Clifford \QAS.  

We cannot use Pauli decoherence (\Le{PauliDiscretization}) to prove a claim analogous to \Cl{claim:stateformat} in the polynomial case for the following reason. \Le{PauliDiscretization} shows that averaging over the Pauli conjugations of an operator removes cross terms, thereby mapping the operator to a convex sum over Pauli operators. \Le{CliffordDiscretization} shows that averaging over the Clifford conjugations of an operator not only maps the operator to a convex sum over Pauli operators, but goes one step further to map each non identity Pauli operator to a unifrom mixture over all non identity Pauli operators. This uniform mixture is crucial to the proof of \Cl{claim:stateformat}; the key part of the proof is the application of the unitary commutation lemma (\Le{clifSimplify}) to the mixture, which allows us to shift the prover's attacks to the end of the protocol. 

Since we do not have a claim analogous to \Cl{claim:stateformat} for the polynomial \QPIP, we instead have to monitor how the authenticated state changes throughout the protocol, as a function of the prover's deviation. At a high level, we do this by partitioning the Hilbert space of the prover according to the interaction transcript (as done in \cite{fk2012}). In each partition, the transcript is fixed at the start and then the measurement results of the state are projected onto the fixed transcript to enforce consistency. This method is formalized in \Cl{claim:polystateformat}, which describes the state shared by the verifier and prover throughout the protocol.  

Now since each partition has a fixed interaction transcript, we can shift the prover's attack to the end of the protocol (his attack no longer determines the interaction transcript). After shifting the prover's attack, we can analyze each partition using the same main ideas we used to prove security of the polynomial \QAS\ (Pauli decoherence from \Le{PauliDiscretization} and sign key security from \Le{lem:PauliPolyAuthSec}). 

We begin by discussing how to apply gates on top of the signed polynomial authentication (Section \ref{des:secapp}). We then describe the protocol, introduce necessary notation and assumptions and conclude with proving the soundness and completeness of the protocol.

\subsection{Application of Quantum Gates}\label{des:secapp}
We will describe how the prover performs
a set of universal gates (consisting of the Fourier transform and Toffoli gate) on authenticated qubits by applying only Clifford operators which do not require knowledge of the Pauli or sign keys. The prover does this by using classical
communication with the verifier and authenticated Toffoli states sent by the verifier. As described in Section \ref{app:toffoli}, if given an authenticated Toffoli state, a Toffoli gate can be applied using logical Pauli, \textit{SUM} and Fourier operations, along with measurement. We now describe how to apply these operations on authenticated states, which will complete our description of how the prover performs a universal set of gates. 

\subsubsection{Pauli Operations}
To apply Pauli $X$ and $Z$ operations, the verifier only needs to update his Pauli keys and the prover does not need to do anything. Recall from Section \ref{sec:logicalgates} that the logical $\wt{X}_k$ operator consists of an application of $X^{k_1}\otimes, {\ldots} \otimes X^{k_m}$ where $k\in\{-1,1\}^m$ is the sign key. We claim that this operation can be applied to the authenticated state by the verifier simply changing his Pauli key from $(x,z)$ to $(x - k,z)$. This is because:
\begin{eqnarray}
P_{x,z}\ket{S^k_a} &=& P_{x-k,z}   P_{x-k,z}^\dagger  P_{x,z}\ket{S^k_a}\\
&=&  P_{x-k,z} X^{-(x-k)}Z^{-z} Z^zX^x\ket{S^k_a}\\
&=& P_{x-k,z}(X^{k_1}\otimes, {\ldots} \otimes X^{k_m})\ket{S^k_a} \\
&=&  P_{x-k,z}\wt{X}_k \ket{S^k_a}
\end{eqnarray}
The $Z$ operator is performed in the same manner as the $X$ operator; all that is needed is a change of the Pauli key. We recall that $\wt{Z}_k=Z^{c_1k_1}\odots Z^{c_mk_m}$. We define the vector $\mathbf{t}$ to be $t_i=c_ik_i$. From the same argument as above, it holds that the change of keys must be $(x,z) \rightarrow (x,z -\mathbf{t})$.

\subsubsection{Fourier and \textit{SUM} Operations}
To apply Fourier and \textit{SUM} operations, the verifier needs to update his Pauli keys and the prover needs to apply the corresponding logical gate. For the \textit{SUM} gate, the prover applies the logical \textit{SUM} gate ($\wt{ \textit{SUM}}$ as given in Section \ref{sec:logicalgates}) and the verifier updates his pair of keys (for $x_A,z_A,x_B,z_B\in F_q^m$) from $(x_A,z_A), (x_B,z_B)$ to
$(x_A,z_A-z_B)$ and $ (x_B+x_A,z_B)$ where $A$ is the control register and $B$ is the target register. This is because the logical \textit{SUM} operator is applied on top of the Pauli keys, and must be shifted past. The update operations of the verifier essentially perform this shift:
\begin{eqnarray}
\wt{\textit{SUM}} (Z^{z_A}X^{x_A}\otimes Z^{z_B}Z^{x_B})\ket{S_a^k}\ket{S_b^k} &=& \wt{\textit{SUM}} (Z^{z_A}X^{x_A}\otimes Z^{z_B}Z^{x_B}) \wt{\textit{SUM}}^\dagger \wt{\textit{SUM}}\ket{S_a^k}\ket{S_b^k}\\
&=& (Z^{z_A- z_B}X^{x_A}\otimes Z^{z_B} X^{x_A + x_B}) \wt{\textit{SUM}}\ket{S_a^k}\ket{S_b^k}
\end{eqnarray}
where the last equality is up to a global phase and follows due to the conjugation properties given in Section \ref{sec:conjugationproperties}.

The Fourier gate is applied in a similar way; the prover applies the logical Fourier transform $\wt{F}$ given in Section \ref{sec:logicalgates} (Claim \ref{claim:fourier}) and the verifier updates his keys according to the conjugation behavior of $\wt{F}$, which we can determine from Section \ref{sec:conjugationproperties}. The following equality is up to a global phase:
\begin{equation}
\wt{F} (Z^zX^x) \wt{F}^\dagger = Z^{c_1x_1}X^{-c_1^{-1}z_1}\otimes{\ldots}\otimes
    Z^{c_mx_m}X^{-c_m^{-1}z_m}
\end{equation}
Therefore, for each register $i$, the verifier must  change the key from $(x_i,z_i)$ to $(-c^{-1}_iz_i,c_ix_i)$.

\subsubsection{Measurement}\label{sec:measurement}
 The prover measures the encoded state in the standard basis and sends the resulting string in $F_q^m$ to the verifier. The verifier first removes the entire Pauli key. Note that we are assuming a classical verifier can remove the $Z$ portion of the Pauli key; this is because the Pauli key is acting on a measured string, and phase gates have no effect on standard basis strings. Therefore, applying a $Z$ operator is the same as not applying it. We choose to assume the verifier does apply it because it simplifies the soundness proof of the protocol (specifically, it comes up in the proof of \Cl{claim:statesimplify}). The verifier then applies $D_k^\dagger$ (see Section \ref{sec:encodingcircuit}), obtaining a string $\delta\in F_q^{m}$. If the prover requires the decoded measurement result, the verifier sends the prover the first coordinate of $\delta$ (which should contain the value of the polynomial at 0). If the last $d$ coordinates of $\delta$ are not 0, the verifier records the measurement as invalid and aborts at the end of the protocol.
\\~\\
Observe that the verifier is not applying $E_k^\dagger$ (the full decoding circuit). It turns out that this is actually enough for the interactive protocol, since we only need to be able to catch attack operators involving Pauli $X$ deviations. Attack operators involving $Z$ deviations will not change measurement results. We will see below (in Corollary \ref{soundnesstrivialpaulis}) that  applying $D_k^\dagger$ and checking the appropriate auxiliary qudits allows the verifier to catch Pauli $X$ deviations.

\subsubsection{Conversion to Logical Circuit}\label{sec:conversiontological}
Now that we have described how to apply gates, we can describe how to convert a quantum circuit on $n$ qubits consisting of gates from the above universal set, $U = U_N\cdots U_1$, into a logical circuit acting on authenticated states. Assume $U$ contains $L$ Toffoli gates. Then 
\begin{equation}
U = A_LT_L\cdots A_1T_1A_0
\end{equation}
where $A_i$ is a Clifford circuit. To apply $U$ to authenticated states, we instead apply 
\begin{equation}
\tilde{A}_L\tilde{T}_L\cdots \tilde{A}_1\tilde{T}_1\tilde{A}_0
\end{equation}
where $\tilde{A}_L$ denotes a logical operation, as described above. Each $\tilde{T}_i$ involves Clifford entanglement operations (which we will denote by $\tilde{B}_i$) with a new magic state followed by a measurement, the results of which are sent to the verifier. Assume that the $i^{th}$ measurement result decodes to $\beta_i\in F_q^3$. Then $\tilde{T}_i$ consists of $\tilde{B}_i$, followed by measurement, followed by correction $\tilde{C}_{\beta_i}$ (which is the logical version of $C_{\beta_i}$ - see Section \ref{app:toffoli} for a reminder of how the Toffoli gate is applied). Now combine the Clifford entangling operators with the preceding Clifford operators in the circuit:
\begin{equation}
\tilde{Q}_i = \tilde{B}_{i+1}\tilde{A}_i
\end{equation}
where $B_{L+1} = \mcI$. Then to apply $U$ to authenticated states, we first apply $\tilde{Q_0}$. Then for $1\leq i\leq L$, we measure, obtaining $\beta_i\in F_q^3$, and then apply $\tilde{Q_i}\tilde{C}_{\beta_i}$. 

\paragraph{Properties of Toffoli Gate by Teleportation}\label{sec:teleportationproperties}
In order to prove soundness of the polynomial \QPIP, we will need to better understand the result of applying a circuit using Toffoli states (as described immediately above in Section \ref{sec:conversiontological}). More specifically, this understanding will come in to play when we are analyzing the behavior of Pauli attacks on the state at the end of the protocol (this is done in Claim \ref{claim:trivialsoundness} and Claim \ref{claim:nontrivialsoundness}). In this section, we will not work with logical operators and authenticated qudits, but with unauthenticated qudits. However, the analysis can immediately be extended to authenticated qudits. To begin, assume the measurement results $\beta_i\in F_q^3$ of each Toffoli gate are fixed beforehand. Then the circuit which will be applied (as described above) is:
\begin{equation}\label{equation:toffolicircuit}
Q_LC_{\beta_L}\cdots Q_1C_{\beta_1}Q_0
\end{equation} 
We will now provide a fact (used in Claim \ref{claim:trivialsoundness} and Claim \ref{claim:nontrivialsoundness}) which characterizes what the state looks like (including measurement results) after applying the circuit in equation \ref{equation:toffolicircuit} on $n + 3L$ qudits (the circuit acts on $n$ input qudits initially in state $\ket{\phi}$ and $L$ Toffoli states of 3 qudits each):
\begin{fact}\label{fact:teleportationstate}
For a string $\beta = (\beta_1,\ldots,\beta_L)\in F_q^{3L}$, where $\beta_i\in F_q^3$, the result of applying 
\begin{equation}
Q_LC_{\beta_L}\cdots Q_1C_{\beta_1}Q_0
\end{equation}
to 
\begin{equation}
\ket{\phi}(\frac{1}{q}\sum_{a,b\in F_q}\ket{a,b,ab})^{\otimes L}
\end{equation}
is
\begin{equation}
\frac{1}{\sqrt{q^{3L}}}\sum_{l\in F_q^{3L}}\ket{l}\ket{\psi}_{\beta,l}
\end{equation}
where $\ket{\phi}$ is a state on $n$ qudits and $\ket{\psi}_{\beta,l}$ is a state on $n$ qudits which equals $U\ket{\phi}$ if $\beta = l$. 
\end{fact}
Before proving the fact, observe that if we project the first register containing $l$ onto $\beta$, we obtain the state $U\ket{\psi}$. This makes sense; if the measurement results obtained are the same ones we fixed for the Clifford corrections, then each Toffoli gate is applied as intended. Moreover, note that without this projection, each $\l\in F_q^{3L}$ is equally probable. We now prove the fact.

\begin{proofof}{Fact \ref{fact:teleportationstate}}
First consider what happens if we would like to apply one Toffoli (as described above) to a 3 qudit state $\ket{\psi}$ using a magic state. After the Clifford operations entangling $\ket{\psi}$ and the magic state, but preceding the measurement (i.e. at the stage of equation \ref{beforemeasurement}), the state is:
\begin{eqnarray}
\frac{1}{\sqrt{q^3}}\sum_{a,b,l\in F_q}\omega^{-l e}\ket{a,b,ab+e,c-a,d-b,l} &=& \frac{1}{\sqrt{q^3}}\sum_{x,y,z\in F_q}\omega^{-ze}\ket{c-x,d-y,(c-x)(d-y)+e}\ket{x,y,z}\nonumber\\
&=& \frac{1}{\sqrt{q^3}}\sum_{x,y,z\in F_q} ((T(X^x\otimes X^y\otimes Z^z)T^\dagger)^\dagger T\ket{\psi})\ket{x,y,z}
\end{eqnarray}
We can write the state in this format because we know that when the measurement result is $x,y,z$, the operation $T(X^x\otimes X^y\otimes Z^z)T^\dagger$ corrects the state to $T\ket{\psi}$. Let's swap the first 3 registers with the last 3 registers (to be consistent with how the state is written in Fact \ref{fact:teleportationstate}). The state is now:
\begin{equation}
\frac{1}{\sqrt{q^3}}\sum_{x,y,z\in F_q} \ket{x,y,z}(T(X^x\otimes X^y\otimes Z^z)T^\dagger)^\dagger T\ket{\psi}
\end{equation}
Now let's assume the correction operator applied is $T(X^a\otimes X^b\otimes Z^c)T^\dagger$. The result is:
\begin{equation}
\frac{1}{\sqrt{q^{3}}}\sum_{x,y,z\in F_q} \ket{x,y,z} \otimes \ket{\psi}_{(a,b,c),(x,y,z)}
\end{equation}
where 
\begin{equation}
\ket{\psi}_{(a,b,c),(x,y,z)} = T(X^{a-x}\otimes X^{b-y}\otimes Z^{c-z}) \ket{\psi}
\end{equation}
Observe that 
\begin{equation}
\ket{\psi}_{(a,b,c),(a,b,c)} = T\ket{\psi}
\end{equation}
In other words, when the correction $\beta$ matches the measurement result $l$, we get the desired result: the application of a Toffoli gate to the initial state. To prove the fact, the above analysis is simply applied each time the Toffoli protocol is performed. 
\end{proofof}

\subsection{Protocol Description}
\begin{protocol}\textbf{Polynomial based Interactive Proof for
    \qc} \label{prot:PolynomialIP}
Fix a security parameter $\epsilon$.  Given is a
 quantum circuit on $n$ qubits
consisting of gates from the above universal set, $U=U_N{\cdots}U_1$, which can be converted to a logical circuit on authenticated qudits as in Section \ref{sec:conversiontological}. We assume there are $L$ Toffoli gates. We assume the circuit $U$ has error
 probability $\le \gamma$.
The verifier sets  $m = \lceil{\log{\frac 1 \epsilon}}\rceil + 1$, $d = \frac{m-1}{2}$ and
 uses $3$ registers of $m$ qudits each, where each qudit is
of dimensionality $q>m$.  The verifier uses the
polynomial \QAS\ with security parameter $d$ to authenticate
$n$ input qudits and $L$ Toffoli states and sends the authenticated states to the prover. The verifier uses
the same sign key (but independent Pauli keys) for each state.
This is done sequentially using $3m$ qudits at a time. Round 0 consists of the prover and verifier performing the Clifford gates $\tilde{Q}_0$. 
At the start of round $i$, for $1\leq i\leq L$, the prover and verifier perform the measurement (as described in Section \ref{sec:measurement}) on the $3m$ qudits as required for the $i^{th}$ Toffoli gate. The verifier sends the prover the decoded measurement result, and then they jointly perform the Clifford corrections required to complete the Toffoli gate and the Clifford circuit $\tilde{Q}_i$. In round $L+1$ (the final round), the verifier and prover perform the measurement of the first authenticated qudit (the verifier does not provide the prover with the decoded measurement result). The verifier aborts if the measurement results from any round were stored as invalid (see Section \ref{sec:measurement}). If he does not abort, he accepts or rejects
according to the final decoded measurement outcome.
\end{protocol}
\begin{statement}{\Th{thm:PolynomialIP}}\textit{For $0< \epsilon < 1$ and $\gamma < 1 - \epsilon$, Protocol \ref{prot:PolynomialIP} is a $\QPIP_{O(\log(\frac{1}{\epsilon}))}$ protocol with
  completeness $1-\gamma$ and soundness $\gamma + \epsilon$ for $\qc_{\gamma}$.}
  \end{statement}
  \\~\\
This theorem implies a second proof for \Th{thm:qcircuit}.
The size of the verifier's register is naively $3m$, but using the same idea
as in the Clifford case, $m+2$ suffice. As a reminder, the idea is to send qudits as they are encoded. For the Toffoli state, the verifier begins with 3 qudits, encodes the first one (using $m+2$ registers at this point), sends the first encoded qudit to the prover, and continues. With $\epsilon=1/2$, $m = 3$ (because $m = 2d+1$) giving a register size of 5 qudits of dimension 5 (since $q > m$). Before we provide the proof of the theorem, we introduce some necessary notation and make several observations about the protocol described above. 

\subsection{Assumptions} 
\begin{itemize}
\item{\textbf{The prover's messages are quantum states.}}
Note that although in the protocol the prover sends the verifier classical strings which the verifier then decodes, we can instead assume that the prover sends the verifier qudits, then the verifier decodes and finally measures. This is because the verifier's decoding operations (which consist of removing Pauli keys and applying $D_k^\dagger$, as described in Section \ref{sec:measurement}) commute with standard basis measurement. In other words, if you consider an $m$ qudit density matrix $\rho$, the following equality holds:
\begin{equation}\label{paulisigndecoding}
\sum_{j\in F_q^m} D_k^\dagger (Z^zX^x)^\dagger \ket{j}\bra{j}\rho \ket{j}\bra{j}Z^zX^xD_k  = \sum_{j\in F_q^m} \ket{j}\bra{j}D_k^\dagger (Z^zX^x)^\dagger \rho Z^zX^xD_k \ket{j}\bra{j} 
\end{equation}

\item{\textbf{The prover's deviation can be delayed until the end of each round.}}
In round $i$ (for $i\geq 1$), we can assume without loss of generality that the prover measures, sends the results to the verifier, receives the decoded measurement results $g(\delta_i)$ from the verifier, and then applies a unitary $\hat{V}_{g(\delta_i)}$ to the authenticated qudits and his extra space. Anything the prover does before the measurement can be shifted to the previous round. $\hat{V}_{g(\delta_i)}$ can be written as
\begin{equation}
\hat{V}_{g(\delta_i)} = \hat{V}_{g(\delta_i)}(\tilde{Q}_i\tilde{C}_{\beta_i})^\dagger \tilde{Q}_i\tilde{C}_{\beta_i} = V_{g(\delta_i)}\tilde{Q}_i\tilde{C}_{\beta_i}
\end{equation}
In other words, we can assume that the prover measures, applies the unitaries requested by the verifier in round $i$ ($\tilde{Q}_i\tilde{C}_{\beta_i}$) and then applies a unitary attack $V_{g(\delta_i)}$. Using similar reasoning, in round 0, we can assume the prover first applies $\tilde{Q}_0$ as requested and then applies a unitary attack. 
\end{itemize}

\subsection{Notation}\label{sec:notation}
Now we provide some of the notation that will be used in the proof. Please see the notation tables in Appendix \ref{app:tables} for all notations together, which hopefully will help in reading this part of the paper, since it is quite heavy on notation. First, throughout this protocol, we will refer to $Z$ and $\mcI$ Pauli operators as trivial and all other Pauli operators (Pauli operators containing the $X$ operator) as non trivial. This is because, as noted in Section \ref{sec:measurement}, trivial Pauli operators cannot change measurement results and non trivial Pauli operators can. Therefore, we only need to ensure that the verifier can detect non trivial Pauli operators. 
\\~\\
For strings $b_1,\ldots,b_j\in F_q^{m}$, let 
\begin{equation}\label{deffunctiong}
g((b_1,\ldots,b_j)) = (b_1(1),\ldots,b_j(1))
\end{equation}
where $b_i(1)\in F_q$ is the first value of the string. This is referred to as the decoded value of a measurement result, as it is the value the verifier will return to the prover after decoding the prover's measurement result.   
\\~\\
If the prover decides to deviate from the protocol, he can apply unitary operators to both the qudits sent by the verifier and his environment. We call the register corresponding to his environment $\mathcal{E}$. For convenience, we also label the other registers of the quantum state shared between the prover and verifier as follows. 
Note that at the start of round 0, the total number of qudits sent to the prover by the verifier is $m' = 3mL + mn$. In every following round, except round $L+1$, $3m$ qudits (which are to be measured) are sent to the verifier ($m$ qudits to be measured are sent in round $L+1$). For $i\geq 1$, we call the register containing all authenticated qudits left at the prover's hands at the beginning of round $i$ (equivalently at the end of round $i-1$) register $\mathcal{P}_i$ (it holds $m(3(L-i+1) + n)$ qudits). We call the register containing all the qudits sent to the verifier in rounds $1,\ldots,i-1$ register $\mathcal{V}_i$ (it holds $3m(i-1)$ qudits and is held by the verifier). We also add another register to the verifier's space (which we call the key register): at the start of round $i$, it contains the state $\tau_i(z,x,k)$, which carries the memory of the sign key $k$ and the Pauli keys $x,z\in F_q^{|\mathcal{P}_i|}$ for those qudits still held by the prover (qudits in register $\mathcal{P}_i$). More precisely:
$$
\tau_i(z,x,k) = \ket{z}\bra{z}\otimes\ket{x}\bra{x}\otimes\ket{k}\bra{k}
$$
Note that we are assuming the verifier no longer keeps record of the Pauli keys for qudits which were already sent to him by the prover; after the verifier uses these Pauli keys to decode, he traces them out of the key register.  
\\~\\
Given this notation, we can now provide Figure \ref{schematic} as an illustration of Protocol \ref{prot:PolynomialIP}. 
\begin{figure}[!ht]\label{protocolfigure}
\centering\includegraphics[scale=.4]{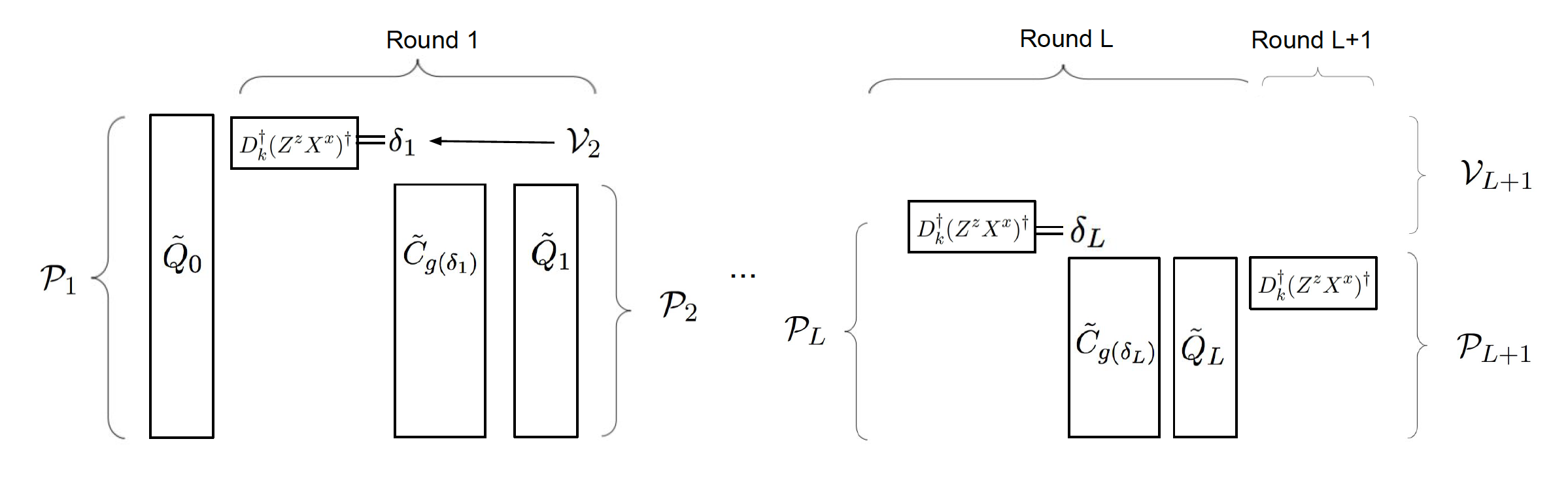}
\caption{This figure illustrates the gates an honest prover would apply during Protocol \ref{prot:PolynomialIP}, and which qubits are in which register during different rounds. In the figure, we are assuming that the measurement result in round $i$ (after the verifier removes the Pauli key and decodes with $D_k^\dagger$) is $\delta_i\in F_q^{3m}$. To simplify the illustration, we have left out the detail that the decoding Pauli keys $(Z^zX^x)^{-1}$ will be different for each register, and that the final decoding (at the start of round $L+1$) acts on $m$ qudits, while the previous $i$ decodings act on $3m$ qudits.  }\label{schematic}
\end{figure}

\subsection{Overall Proof of \Th{thm:PolynomialIP}}\label{sec:securitypolynomial}
\begin{proof}
 The completeness is trivial, similarly to the Clifford case (see \Th{thm:CliffordIP}). To prove soundness, recall that we begin with a \qc\ instance, $U$, which takes as input $\ket{y}^{\otimes n}$ (where $y$ is a classical $n$ bit string), and for soundness we would like to show that if the first qudit of $U\ket{y}^{\otimes n}$ is 0 with probability $1 - \gamma$, the verifier will either abort or not accept the final decoded measurement result with probability $\geq 1 - (\gamma + \epsilon)$, which gives the soundness parameter of $\gamma + \epsilon$. To do this, we will characterize how the prover's state evolves throughout the protocol. 

When each qudit is sent to the verifier at the start of round $i$ as part of the application of the Toffoli gate, the verifier will apply the inverse of the appropriate Pauli keys, interpolate with operator $D_k^\dagger$ (see \Def{def:interpolationcircuit}), and measure the $3m$ received qudits. Let the result of this measurement be $\delta_i\in F_q^{3m}$. We thus denote the effect of the measurement with this result by the projection $\ket{\delta_i}\bra{\delta_i}$ conjugating the density matrix (for $\delta_i\in F_q^{3m}$). Of course we will sum over all the different values of $\delta_i$. Next (in all rounds except round $L+1$), the verifier will send the prover the decoded measurement results $g(\delta_i)$ so the prover will be able to apply the Clifford correction $C_{g(\delta_i)}$, as written in equation \ref{cliffordcorrection}. The verifier will then instruct the prover to apply the next set of Clifford gates $\tilde{Q}_i$ in the circuit. Since the verifier sent the prover $g(\delta_i)$, the prover's next attack can be dependent on this value. 

We now provide the claim characterizing the state shared between the prover and the verifier, as a summation over all  of the measurement results from previous rounds ($\Delta_{i-1} = (\delta_1,\ldots,\delta_{i-1})$):
\begin{claim}[\bfseries Polynomial \QPIP\ State Evolution]\label{claim:polystateformat}
For $1\leq i\leq L+1$, the state shared by the prover and the verifier at the start of round $i$ can be written as:
\begin{equation}
\frac{1}{2^m|\mbP_{m'}|}\sum_{\substack{z,x\in F_q^{|\mathcal{P}_i|}\\k\in\{-1,1\}^m}}\tau_i(z,x,k)\otimes\sum_{\Delta_{i-1},z_1,x_1\in F_q^{|\mathcal{V}_i|}} W_{\Delta_{i-1},\hat{z},\hat{x},k}^i (\rho^k_{g(\Delta_{i-1})}\otimes\rho_{\mathcal{E}}) W_{\Delta_{i-1},\hat{z},\hat{x},k}^{i\dagger}
\end{equation}
where $m' = 3mL + mn$,
\begin{equation}
\hat{z} = (z_1,z), \hat{x} = (x_1,x),
\end{equation}
\begin{equation}
W_{\Delta_{i-1},\hat{z},\hat{x},k}^i = (\ket{\Delta_{i-1}}\bra{\Delta_{i-1}}(D_k^\dagger)^{\otimes |M_{i}|}(Z^{z_1}X^{x_1})^\dagger_{\mathcal{V}_i}\otimes\mcI_{\mathcal{P}_i,E}) U_{g(\Delta_{i-1})} ((Z^{z_1}X^{x_1})_{\mathcal{V}_i}\otimes (Z^{z}X^{x})_{\mathcal{P}_i}\otimes \mcI_{\mathcal{E}})
\end{equation}
where $U_{g(\Delta_{i-1})}$ is a unitary operator dependent on $g(\Delta_{i-1})$ and 
\begin{equation}\label{eq:finalstate}
\rho_{g(\Delta_{i-1})}^k = (\tilde{Q}_{i-1}\tilde{C}_{g(\delta_{i-1})}\cdots \tilde{Q}_1\tilde{C}_{g(\delta_1)}\tilde{Q}_0)\rho^k ((\tilde{Q}_{i-1}\tilde{C}_{g(\delta_{i-1})}\cdots \tilde{Q}_1\tilde{C}_{g(\delta_1)}\tilde{Q}_0)^\dagger
\end{equation}
for $\Delta_{i-1} = (\delta_1,\ldots,\delta_{i-1})\in F_q^{|\mathcal{V}_i|}$, where $\rho$ is the initial state on $3L + n$ qubits (consisting of $L$ Toffoli states and an $n$ qudit input state), $\rho^k$ indicates authentication as described in equation \ref{def:encodedstate}, and $\rho_{\mathcal{E}}$ is the initial state of the prover's environment. 
\end{claim}
The projection $\ket{\Delta_{i-1}}\bra{\Delta_{i-1}}$ denotes the verifier's measurement (it acts on register $\mathcal{V}_i$), part of which has been sent back to the prover in the form of $g(\Delta_{i-1})$ (hence the dependence of $U$ and $\rho^k$ on $g(\Delta_{i-1})$). 

As a brief aside, recall that (as mentioned at the start of Section \ref{sec:IPQ}) one key difference between the Clifford and polynomial protocols is that the authenticated state throughout the polynomial protocol is not necessarily the correct authenticated state (i.e. the authentication of the state which would result by applying the \qc\ instance $U$). This can be seen by observing the form of $\rho^k_{g(\Delta_{i-1})}$. Note that if the projection $\ket{\Delta_{i-1}}\bra{\Delta_{i-1}}$ acted directly on the state, it would indeed be the correct state. However, the projection acts after the attack $U_{g(\Delta_{i-1})}$, which implies that if $U_{g(\Delta_{i-1})}$ acts non trivially on register $\mathcal{V}_i$, $\rho^k_{g(\Delta_{i-1})}$ will not necessarily be the correct authenticated state.

Before we proceed with the proof of soundness, we will write down the state at the start of round 1 as an example of how \Cl{claim:polystateformat} works. At the start of round 1, the state shared between the verifier and the prover is:
\begin{eqnarray}
\frac{1}{2^m|\mbP_{m'}|}\sum_{\substack {z,x\in F_q^{|R_1|} \\ k\in\{-1,1\}^m}} \tau_{1}(z,x,k)\otimes  V_0 (Z^zX^x\tilde{Q}_0\otimes\mcI_{\mathcal{E}}) \rho^k\otimes\rho_{\mathcal{E}} (Z^zX^x\tilde{Q}_0\otimes\mcI_{\mathcal{E}})^\dagger V_0^\dagger\label{eq:startstate}
\end{eqnarray}
where $\rho^k$ is the initial state of the qudits sent to the prover and $V_0$ is the unitary attack of the prover applied at the end of round 0. Note that as the prover and verifier performed the Clifford operator $\tilde{Q}_0$, the verifier updated his initial Pauli keys to account for this operator (as described in Section \ref{des:secapp}). This is why the Pauli operator $Z^zX^x$ acts after $\tilde{Q}_0$ on $\rho^k$ in equation \ref{eq:startstate}. As you can see, \Cl{claim:polystateformat} holds for $i = 1$. 

We now proceed with the proof of soundness. \Cl{claim:polystateformat} implies that at the start of the final round, round $L+1$, the joint state of the prover's registers, $\mathcal{P}_{L+1}$ and the environment $\mathcal{E}$, and the verifier's registers, $\mathcal{V}_{L+1}$ and the key register containing the sign key and Pauli keys of qudits in $\mathcal{P}_{L+1}$, is:
\begin{equation}\label{startoffinalround}
\frac{1}{2^m|\mbP_{m'}|}\sum_{\substack{z,x\in F_q^{|\mathcal{P}_{L+1}|}\\k\in\{-1,1\}^m}}\tau_{L+1}(z,x,k)\otimes\sum_{\Delta_L,z_1,x_1\in F_q^{|\mathcal{V}_{L+1}|}} W_{\Delta_L,\hat{z},\hat{x},k}^{L+1} (\rho^k_{g(\Delta_L)}\otimes\rho_{\mathcal{E}} )(W_{\Delta_L,\hat{z},\hat{x}}^{L+1})^\dagger
\end{equation}
As in previous rounds, the verifier decodes the final authenticated qudit sent by the prover, with both the Pauli and sign key. Let $\mathcal{F}$ denote the register containing the final authenticated qudit. Let $\mathcal{P}_{final}$ denote the register of the remaining authenticated qudits (this contains all qudits in $\mathcal{P}_{L+1}$ except those in register $\mathcal{F}$). Let $\mathcal{V}_{final} = \mathcal{F}\cup \mathcal{V}_{L+1}$ be the register containing all qudits sent to the verifier during the protocol. 
\begin{corol}\label{finalstate}
The state shared between the prover and verifier after the decoding of register $\mathcal{F}$ is:
\begin{equation}\label{corfinalstateequation}
\rho_{L+1} \EqDef \frac{1}{2^m|\mbP_{m'}|}\sum_{\substack{z,x\in F_q^{|\mathcal{P}_{final}|}\\k\in\{-1,1\}^m}}\tau_{final}(z,x,k)\otimes\sum_{\substack{\Delta_L\in F_q^{|\mathcal{V}_{L+1}|}\\z_1,x_1\in F_q^{|\mathcal{V}_{final}|}}}  V_{\Delta_L,\hat{z},\hat{x},k}(\rho^k_{g(\Delta_L)}\otimes\rho_{\mathcal{E}})V_{\Delta_L,\hat{z},\hat{x},k}^\dagger
\end{equation}
where $\hat{z} = (z_1,z)$, $\hat{x} = (x_1,x)$ and
$$
V_{\Delta_L,\hat{z},\hat{x},k} = ((\ket{\Delta_L}\bra{\Delta_L}\otimes\mcI_{\mathcal{F}})(D_k^\dagger)^{\otimes |\mathcal{V}_{final}|}(Z^{z_1}X^{x_1})_{\mathcal{V}_{final}}^\dagger\otimes\mcI_{\mathcal{P}_{final},{\mathcal{E}}}) U_{g(\Delta_L)} ((Z^{z_1}X^{x_1})_{\mathcal{V}_{final}}\otimes (Z^zX^x)_{\mathcal{P}_{final}}\otimes \mcI_{\mathcal{E}})
$$
\end{corol}
\begin{proof}The only change between this state and equation \ref{startoffinalround} is the decoding of register $F$ (by applying the inverse Pauli key and the signed polynomial decoding $D_k^\dagger$).
Observe that $W_{\Delta_L,\hat{z},\hat{x}}^{L+1}$ and $V_{\Delta_L,\hat{z},\hat{x},k}$ differ only to the left of $U_{g(\delta)}$; nothing changes to the right. This is because in the state above, we are averaging over all Pauli operators acting on registers $\mathcal{V}_{final}$ and $\mathcal{P}_{final}$ and in equation \ref{startoffinalround} we are averaging over all Pauli operators acting on registers $\mathcal{V}_{L+1}$ and $\mathcal{P}_{L+1}$. To see that this is the same, observe that 
\begin{equation}
\mathcal{V}_{final}\cup \mathcal{P}_{final} = \mathcal{V}_{L+1}\cup \mathcal{P}_{L+1}
\end{equation}
To the left of $U_{g(\delta)}$, one additional register ($F$) is decoded first by the corresponding Pauli keys (which is reflected by the replacement of $(Z^{z_1}X^{x_1})_{\mathcal{V}_{L+1}}^\dagger$ with $(Z^{z_1}X^{x_1})_{\mathcal{V}_{final}}^\dagger$) and then by $D_k^\dagger$. The projection does not change (as indicated by $\mcI_{\mathcal{F}}$ in the projection) as we are only decoding register $\mathcal{F}$.  
\end{proof}

Note that the verifier only holds the first key register (containing $\tau_{final}(z,x,k)$) and register $\mathcal{V}_{final}$. Recall that our goal is to show that for the \qc\ instance $U$, if the first qudit of $U\ket{0}^{\otimes n}$ is 0 with probability $1 - \gamma$, the verifier will either abort or not accept the final decoded measurement result with probability $\geq 1 - (\gamma + \epsilon)$. For this purpose, we define the following projection on $\mcV_{final} = \mathcal{F}\cup\mcV_{L+1}$:
\begin{equation}\label{acceptingprojection}
\hat{\Pi}_0 \EqDef (\mcI^{\otimes d+1}\otimes\ket{0}\bra{0}^{\otimes d})^{\otimes 3L}_{\mathcal{V}_{L+1}}\otimes (\ket{1}\bra{1}\otimes\mcI^{\otimes d}\otimes \ket{0}\bra{0}^{\otimes d})_{\mathcal{F}}\EqDef (\hat{\Pi}_0)_{\mathcal{V}_{L+1}}\otimes(\hat{\Pi}_0)_{\mathcal{F}}
\end{equation}
The first term in the above projection describes the space of valid measurement results (i.e. strings which interpolate to low degree polynomials). The second term describes the space of a final qudit which is accepted and decodes to 1. We would like to show that
\begin{equation}\label{initialsoundnessequation}
\tr(\hat{\Pi}_0\rho_{L+1}|_{\mathcal{V}_{final}})\leq \gamma + \epsilon
\end{equation}
In other words, if the decoded measurement result of the final qudit does not yield 1, the verifier rejects or aborts with high probability. Each block of $m$ qudits in the register $\mathcal{V}_{L+1}$ is projected onto $\mcI^{\otimes d+1}\otimes\ket{0}\bra{0}^{\otimes d}$ as we are not looking for a specific decoded value in the measurement registers; we are only checking that the measurement results are valid.

Observe that in order to bound soundness, we only need to look at $\rho_{L+1}$ on $\mathcal{V}_{final}$; the key register containing $\tau_{final}(z,x,k)$ was unnecessary. This is because the keys $z,x$ acting on $\mathcal{P}_{final}$ will not be used; that register is never sent to the verifier. Also, the verifier no longer needs to remember the sign key, since it has already been used to decode the qudits in $\mathcal{V}_{final}$. Therefore, the verifier can trace out the first register containing $\tau_{final}(z,x,k)$. 

Before continuing to prove equation \ref{initialsoundnessequation}, we can simplify $\rho_{L+1}|_{\mathcal{V}_{final}}$:
\begin{claim}[\bfseries Final State]{\label{claim:statesimplify}}
$\rho_{L+1}|_{\mathcal{V}_{final}}$ is equal to
\begin{equation}
\frac{1}{2^m}\sum_{k\in\{-1,1\}^m}\sigma_k = \frac{1}{2^m}\sum_{k\in\{-1,1\}^m}\sum\limits_{P\in\mbP_{|\mathcal{V}_{final}|}}\sigma_k^P
\end{equation}
where
\begin{equation} \label{eq:sigmadef}
\sigma_k^P = \sum\limits_{\Delta_L\in F_q^{|\mathcal{V}_{L+1}|}} \alpha_{P,g(\Delta_L)}\cdot((\ket{\Delta_L}\bra{\Delta_L}\otimes\mcI_{\mathcal{F}})(D_k^\dagger)^{\otimes|\mathcal{V}_{final}|}P)\sigma_{g(\Delta_L)}^k((\ket{\Delta_L}\bra{\Delta_L}\otimes\mcI_{\mathcal{F}})(D_k^\dagger)^{\otimes|\mathcal{V}_{final}|}P)^\dagger
\end{equation}
and $\tr_{\mathcal{P}_{final}}(\rho_{g(\Delta_L)}^k) = \sigma_{g(\Delta_L)}^k$,
\begin{equation}
\alpha_{P,g(\Delta_L)} = \frac{1}{q^{|\mathcal{P}_{final}|}}\tr(U_{g(\Delta_L)}^P(\mcI_{\mathcal{P}_{final}}\otimes\rho_{\mathcal{E}})(U_{g(\Delta_L)}^P)^\dagger)
\end{equation}
and
\begin{equation}
U_{g(\Delta_L)} = \sum_{P\in\mbP_{|\mathcal{V}_{final}|}}P\otimes U_{g(\Delta_L)}^P
\end{equation}
\end{claim}
Starting from the form of $\rho_{L+1}$ in equation \ref{corfinalstateequation} in Corollary \ref{finalstate}, we show this claim by first summing over $z,x$ (this can be done since $\tau_{final}(z,x,k)$ is traced out), which has the effect of mixing register $\mathcal{P}_{final}$, as shown in the Pauli mixing lemma, \Le{paulimix} (which is analogous to \Le{cliffordmix} and also proven in Appendix \ref{app:backgroundcliffordpauli}). Next, we can use $z_1,x_1$ to decohere (or remove all cross terms of) the part of $U_{g(\Delta_L)}$ acting on register $\mathcal{V}_{final}$ (by applying \Le{PauliDiscretization}). 

Now let's return to our goal of proving equation \ref{initialsoundnessequation}. With the above state simplification, we are now proving:
\begin{equation}\label{soundnessequation}
\frac{1}{2^m}\tr(\hat{\Pi}_0(\sum_k \sigma_k)= \sum_{P\in\mbP_{\mathcal{V}_{final}}}\frac{1}{2^m}\tr(\hat{\Pi}_0\sum_k\sigma_k^P)\leq\gamma + \epsilon 
\end{equation}
We first consider terms $\sigma_k^P$ for which $P$ is trivial (i.e. $P$ consists only of $Z$ and $\mcI$ operators). To prove the following claim, we first observe that trivial Pauli operators have no effect on measurement results, since they commute with the verifier's decoding process (application of $D_k^\dagger$ and the inverse Pauli keys). Given this observation, we can see that the prover's decoded final answer will be 0 with probability $1 - \gamma$ (as it should be), and therefore we can upper bound the projection of the state onto $\hat{\Pi}_0$:
\begin{claim}[\bfseries Trivial Deviation]{\label{claim:trivialsoundness}}
For trivial $P$, 
$$
\frac{1}{2^m}\tr(\hat{\Pi}_0\sum_k\sigma_k^P) \leq \frac{\gamma}{q^{3L}}\sum\limits_{a\in F_q^{3L}} \alpha_{P,a}
$$
\end{claim}
Next, we consider terms $\sigma_k^P$ for which $P$ is non trivial. By using \Le{lem:PauliPolyAuthSec}, which implies that $P$ can produce a non zero trace (after the state is projected onto $\hat{\Pi}_0$) for at most 2 values of $k$, we show:
\begin{claim}[\bfseries Nontrivial Deviation]{\label{claim:nontrivialsoundness}}
For non trivial $P$, 
$$
\frac{1}{2^m}\tr(\hat{\Pi}_0\sum_k\sigma_k^P) \leq \frac{1}{q^{3L}2^{m-1}}\sum_{a\in F_q^{3L}}\alpha_{P,a}
$$ 
\end{claim}
By combining both claims, we obtain:
\begin{eqnarray}
\frac{1}{2^m}\tr(\hat{\Pi}_0\sum_k\sigma_k) &=& 
\frac{1}{2^m}\sum_{P\in\mbP_{\mathcal{V}_{final}}}\tr(\hat{\Pi}_0\sum_k\sigma_k^P)\\ 
&\leq&
\max(\gamma,\frac{1}{2^{m-1}})\frac{1}{q^{3L}}\sum_{a\in F_q^{3L}}(\sum_{P\in\mbP_{\mathcal{V}_{final}}}\alpha_{P,a})\\
&=&\max(\gamma,\frac{1}{2^{m-1}})\\
&\leq& \gamma + \frac{1}{2^{m-1}}
\end{eqnarray}
The final equality follows because:
$$
\sum_{P\in\mbP_{\mathcal{V}_{final}}}\alpha_{P,a} = 1 
$$
by \Le{Decompose}. 
\end{proof}
\subsection{Proof of \Cl{claim:polystateformat} (Polynomial \QPIP\ State Evolution)}
\begin{proof}
We will prove this claim by induction. The base case (round 1) is proven already in Section \ref{sec:securitypolynomial}, equation \ref{eq:startstate}. We assume the claim holds in round $i$ and show that it holds in round $i+1$. By the inductive hypothesis, we have the state shared by the prover and verifier in round $i$ is:
\begin{equation}
\frac{1}{2^m|\mbP_{m'}|}\sum_{\substack{z,x\in F_q^{|\mathcal{P}_i|}\\ k\in\{-1,1\}^m}}\tau_i(z,x,k)\otimes \sum_{\Delta_{i-1},z_1,x_1\in F_q^{|\mathcal{V}_i|}} W_{\Delta_{i-1},\hat{z},\hat{x},k}^i (\rho^k_{g(\Delta_{i-1})}\otimes\rho_{\mathcal{E}}) W_{\Delta_{i-1},\hat{z},\hat{x},k}^{i\dagger}
\end{equation}
where $\hat{z} = (z_1,z)$, $\hat{x} = (x_1,x)$ and
$$
W_{\Delta_{i-1},\hat{z},\hat{x},k}^i = (\ket{\Delta_{i-1}}\bra{\Delta_{i-1}}(D_k^\dagger)^{\otimes |\mathcal{V}_i|}(Z^{z_1}X^{x_1})^\dagger_{\mathcal{V}_i}\otimes\mcI_{\mathcal{P}_i,E}) U_{g(\Delta_{i-1})} ((Z^{z_1}X^{x_1})_{\mathcal{V}_i}\otimes (Z^{z}X^{x})_{\mathcal{P}_i}\otimes \mcI_{\mathcal{E}})
$$
Recall that the verifier holds register $\mathcal{V}_i$ and registers $\mathcal{P}_i$ and $E$ are held by the prover. When the prover measures and sends the verifier his measurement results, the verifier decodes them with both the Pauli keys and the sign key (as in equation \ref{paulisigndecoding}) to obtain $\delta_i\in F_q^{3m}$. The shared state at this point is:
\begin{equation}\label{sharedstatebeforesimp}
\frac{1}{2^m|\mbP_{m'}|}\sum_{\substack{z,x\in F_q^{|R_{i+1}|}\\ k\in\{-1,1\}^m}}\tau_{i+1}(z,x,k)\otimes\sum_{\Delta_{i},z_1,x_1\in F_q^{|M_{i+1}|}} T_{\Delta_{i},\hat{z},\hat{x},k}(\rho^k_{g(\Delta_{i-1})}\otimes\rho_{\mathcal{E}}) T_{\Delta_{i},\hat{z},\hat{x},k}^\dagger
\end{equation}
where
$$
T_{\Delta_i,\hat{z},\hat{x},k} = (\ket{\Delta_i}\bra{\Delta_i}(D_k^\dagger)^{\otimes |M_{i+1}|}(Z^{z_1}X^{x_1})_{M_{i+1}}^\dagger\otimes\mcI_{R_{i+1},E})U_{g(\Delta_{i-1})}((Z^{z_1}X^{x_1})_{M_{i+1}}\otimes (Z^{z}X^{x})_{R_{i+1}}\otimes \mcI_{\mathcal{E}})
$$
Note that $\Delta_i = (\delta_i,\Delta_{i-1})$, where $\delta_i\in F_q^{3m}$ ($\delta_i$ is the measurement result obtained in round $i$) and $\Delta_{i-1}\in F_q^{|\mathcal{V}_i|}$ (measurement results from previous rounds). The key difference here is that we have taken $3m$ qudits from register $R_{i}$ and added them to register $\mathcal{V}_i$ to create registers $R_{i+1}$ and $M_{i+1}$. We have also removed the Pauli keys corresponding to the newly measured qudits from the first register; the verifier traces out these keys after decoding as he no longer needs them. 

The remainder of round $i$ consists of the prover and verifier performing the Clifford gate $\tilde{Q}_i\tilde{C}_{g(\delta_i)}$. To show that the shared state in round $i+1$ is of the form described in \Cl{claim:polystateformat}, we need to replace $\rho^k_{g(\Delta_{i-1})}$ with $\rho^k_{g(\Delta_{i})}$ in equation \ref{sharedstatebeforesimp}. This can be done by determing how the application of $\tilde{Q}_i\tilde{C}_{g(\delta_i)}$ changes the state. Recall that:
\begin{eqnarray}
\rho_{g(\Delta_{i})}^k &=& (\tilde{Q}_{i}\tilde{C}_{g(\delta_{i})}\cdots \tilde{Q}_1\tilde{C}_{g(\delta_1)}\tilde{Q}_0)\rho^k ((\tilde{Q}_{i}\tilde{C}_{g(\delta_{i})}\cdots \tilde{Q}_1\tilde{C}_{g(\delta_1)}\tilde{Q}_0)^\dagger\\
&=& \tilde{Q}_i\tilde{C}_{g(\delta_{i})}\rho^k_{g(\Delta_{i-1})} (\tilde{Q}_i\tilde{C}_{g(\delta_{i})})^\dagger
\end{eqnarray}

In order to replace $\rho^k_{g(\Delta_{i-1})}$ with $\rho^k_{g(\Delta_{i})}$, we need to commute $\tilde{Q}_i\tilde{C}_{g(\delta_i)}$ past $T_{\Delta_i,\hat{z},\hat{x},k}$. First observe that $\tilde{Q}_i\tilde{C}_{g(\delta_i)}$ operates on the register held by the prover, $R_{i+1}$, and therefore commutes with operators acting on register $M_{i+1}$. However, it does not commute with $U_{g(\Delta_{i-1})}$. To take care of this issue, observe that:
\begin{equation}
\tilde{Q}_i\tilde{C}_{g(\delta_i)}U_{g(\Delta_{i-1})} = (\tilde{Q}_i\tilde{C}_{g(\delta_i)})U_{g(\Delta_{i-1})} (\tilde{Q}_i\tilde{C}_{g(\delta_i)})^\dagger(\tilde{Q}_i\tilde{C}_{g(\delta_i)})
\end{equation}

Now the rightmost part of the above expression, $\tilde{Q}_i\tilde{C}_{g(\delta_i)}$, is acting on the Pauli operator $(Z^zX^x)_{R_{i+1}}$ which is acting on $\rho^k_{g(\Delta_{i-1})}$. If the verifier updates his Pauli keys  for register $R_{i+1}$ (as is part of the protocol for performing a Clifford operation, described in Section \ref{des:secapp}), $\tilde{Q}_i\tilde{C}_{g(\delta_i)}$ can be commuted past the Pauli operator:
\begin{equation}
\tilde{Q}_i\tilde{C}_{g(\delta_i)} Z^zX^x = \tilde{Q}_i\tilde{C}_{g(\delta_i)} Z^zX^x (\tilde{Q}_i\tilde{C}_{g(\delta_i)})^\dagger \tilde{Q}_i\tilde{C}_{g(\delta_i)} 
\end{equation}
As described in more detail in Section \ref{des:secapp}, applying the Clifford operator $\tilde{Q}_i\tilde{C}_{g(\delta_i)}$ involves both the prover applying the operator to the authenticated states and the verifier updating his Pauli keys from $Z^zX^x$ to $\tilde{Q}_i\tilde{C}_{g(\delta_i)} Z^zX^x (\tilde{Q}_i\tilde{C}_{g(\delta_i)})^\dagger$ (this is a Pauli since $\tilde{Q}_i\tilde{C}_{g(\delta_i)}$ is a Clifford). 

Now $\tilde{Q}_i\tilde{C}_{g(\delta_i)}$, is acting directly on $\rho^k_{g(\Delta_{i-1})}$ so we have:
\begin{equation}
\rho_{g(\Delta_i)}^k = \tilde{Q}_i\tilde{C}_{g(\delta_i)}\rho_{g(\Delta_{i-1})}^k\tilde{C}_{g(\delta_i)}^\dagger \tilde{Q}_i^\dagger
\end{equation}
Note that this is still a state encoded with the signed polynomial code (hence the superscript $k$), since the Clifford operators are logical operators on the signed polynomial encoding. Finally, the prover can apply another attack $V_{g(\Delta_i)}$. Note that this attack acts only on the registers held by the prover ($R_{i+1}$ and $E$) and therefore can be shifted past operators acting on register $M_{i+1}$ in $T_{\Delta_i,\hat{z},\hat{x},k}$. We now set:
\begin{equation}
U_{g(\Delta_i)} = V_{g(\Delta_i)}\tilde{Q}_i\tilde{C}_{g(\delta_i)}U_{g(\Delta_{i-1})}\tilde{C}_{g(\delta_i)}^\dagger \tilde{Q}_i^\dagger
\end{equation}
The prover's state at the end of the round is then:
\begin{equation}
\frac{1}{2^m|\mbP_{m'}|}\sum_{\substack{z,x\in F_q^{|R_{i+1}|}\\ k\in\{-1,1\}^m}}\tau_{i+1}(z,x,k)\otimes\sum_{\Delta_i,z_1,x_1\in F_q^{|M_{i+1}|}} W^{i+1}_{\Delta_i,\hat{z},\hat{x},k}(\rho^k_{g(\Delta_i)}\otimes\rho_{\mathcal{E}})(W^{i+1}_{\Delta_i,\hat{z},\hat{x},k})^\dagger
\end{equation}
where
$$
W^{i+1}_{\Delta_i,\hat{z},\hat{x},k} = (\ket{\Delta_i}\bra{\Delta_i}(D_k^\dagger)^{\otimes |M_{i+1}|}(Z^{z_1}X^{x_1})^\dagger\otimes\mcI_{R_{i+1},E}) U_{g(\Delta_i)} ((Z^{z_1}X^{x_1})_{M_{i+1}}\otimes (Z^{z}X^{x})_{R_{i+1}}\otimes \mcI_{\mathcal{E}})
$$
\end{proof}
\subsection{Proof of \Cl{claim:statesimplify} (Final State)}
\begin{proof}
Recall that we start with the state given in equation \ref{corfinalstateequation} in Corollary \ref{finalstate}:
\begin{equation}
\rho_{L+1} \EqDef \frac{1}{2^m|\mbP_{m'}|}\sum_{\substack{z,x\in F_q^{|\mathcal{P}_{final}|}\\k\in\{-1,1\}^m}}\tau_{final}(z,x,k)\otimes\sum_{\substack{\Delta_L\in F_q^{|\mathcal{V}_{L+1}|}\\z_1,x_1\in F_q^{|\mathcal{V}_{final}|}}}  V_{\Delta_L,\hat{z},\hat{x},k}(\rho^k_{g(\Delta_L)}\otimes\rho_{\mathcal{E}})V_{\Delta_L,\hat{z},\hat{x},k}^\dagger
\end{equation}
where $\hat{z} = (z_1,z)$, $\hat{x} = (x_1,x)$ and
$$
V_{\Delta_L,\hat{z},\hat{x},k} = ((\ket{\Delta_L}\bra{\Delta_L}\otimes\mcI_{\mathcal{F}})(D_k^\dagger)^{\otimes |\mathcal{V}_{final}|}(Z^{z_1}X^{x_1})_{\mathcal{V}_{final}}^\dagger\otimes\mcI_{\mathcal{P}_{final},\mathcal{E}}) U_{g(\Delta_L)} ((Z^{z_1}X^{x_1})_{\mathcal{V}_{final}}\otimes (Z^zX^x)_{\mathcal{P}_{final}}\otimes \mcI_{\mathcal{E}})
$$
Our goal is to determine the form of the state after tracing out the first register (the key register) and registers $\mathcal{P}_{final}$ and $E$. We begin by tracing out the key register, which allows us to sum over $z,x\in F_q^{|\mathcal{P}_{final}|}$. We are also allowed to sum over $k$, but we will keep $k$ fixed while we simplify the state. The state can then be written as:
\begin{equation}
\sum_{z,x\in F_q^{|\mathcal{P}_{final}|}}\sum_{\substack{\Delta_L\in F_q^{|\mathcal{V}_{L+1}|}\\z_1,x_1\in F_q^{|\mathcal{V}_{final}|}}}  V_{\Delta_L,\hat{z},\hat{x},k}(\rho^k_{g(\Delta_L)}\otimes\rho_{\mathcal{E}})V_{\Delta_L,\hat{z},\hat{x},k}^\dagger
\end{equation}
By \Le{paulimix}, this has the effect of mixing register $\mathcal{P}_{final}$. The state is now:
\begin{equation}
\frac{1}{q^{|\mathcal{P}_{final}|}}\sum_{\substack{\Delta_L\in F_q^{|\mathcal{V}_{L+1}|}\\z_1,x_1\in F_q^{|\mathcal{V}_{final}|}}} V_{\Delta_L,z_1,x_1,k}' (\sigma^k_{g(\Delta_L)}\otimes\mcI_{\mathcal{P}_{final}}\otimes\rho_{\mathcal{E}})  V_{\Delta_L,z_1,x_1,k}'^\dagger
\end{equation}
where 
$$
V_{\Delta_L,z_1,x_1,k}' =((\ket{\Delta_L}\bra{\Delta_L}\otimes\mcI_{\mathcal{F}})(D_k^\dagger)^{\otimes |\mathcal{V}_{final}|}(Z^{z_1}X^{x_1})^\dagger\otimes\mcI_{\mathcal{P}_{final},\mathcal{E}}) U_{g(\Delta_L)} (Z^{z_1}X^{x_1}\otimes \mcI_{\mathcal{P}_{final},\mathcal{E}})
$$
and
$$
\tr_{\mathcal{P}_{final}}(\rho^k_{g(\Delta_L)}) = \sigma^k_{g(\Delta_L)}
$$
Next, we observe that the Pauli encoding/decoding of $Z^{z_1}X^{x_1}$ on register $\mathcal{V}_{final}$ has the effect of decohering (removing cross terms of) the part of $U_{g(\Delta_L)}$ that is acting on $\mathcal{V}_{final}$, as shown in \Le{PauliDiscretization}.

Applying the lemma with:
\begin{equation}
U = U_{g(\Delta_L)} = \sum_{P\in\mbP_{|\mathcal{V}_{final}|}} P\otimes U_{g(\Delta_L)}^P 
\end{equation}
we can simplify the prover's decoded state to:
\begin{equation}
\sum\limits_{\substack{\Delta_L\in F_q^{|\mathcal{V}_{L+1}|}\\P\in\mbP_{|\mathcal{V}_{final}|}}} ((\ket{\Delta_L}\bra{\Delta_L}\otimes\mcI_{\mathcal{F}})(D_k^\dagger)^{\otimes |\mathcal{V}_{final}|}P)\sigma_{g(\Delta_L)}^k((\ket{\Delta_L}\bra{\Delta_L}\otimes\mcI_{\mathcal{F}})D_k^{\otimes |\mathcal{V}_{final}|}P)^\dagger\otimes U_P^{g(\Delta_L)}(\mcI_{\mathcal{P}_{final}}\otimes\rho_{\mathcal{E}})(U_P^{g(\Delta_L)})^\dagger
 \end{equation}
 where the above state also has a factor of $\frac{1}{q^{|\mathcal{P}_{final}|}}$.
 
We trace out registers $\mathcal{P}_{final}$ and $E$ since the verifier will not look at these registers:
\begin{eqnarray}
\sigma_k &=&\sum\limits_{\substack{\Delta_L\in F_q^{|\mathcal{V}_{L+1}|}\\P\in\mbP_{|\mathcal{V}_{final}|}}} \alpha_{P,\Delta_L}\cdot ((\ket{\Delta_L}\bra{\Delta_L}\otimes\mcI_{\mathcal{F}})D_k^{\otimes |\mathcal{V}_{final}|}P)\sigma_{g(\Delta_L)}^k((\ket{\Delta_L}\     \bra{\Delta_L}\otimes\mcI_{\mathcal{F}})D_k^{\otimes |\mathcal{V}_{final}|}P)^\dagger\nonumber
 \end{eqnarray}
 where $\alpha_{P,\Delta_L} = \frac{1}{q^{|\mathcal{P}_{final}|}}\tr(U_P^{g(\Delta_L)}(\mcI_{\mathcal{P}_{final}}\otimes\rho_{\mathcal{E}})(U_P^{g(\Delta_L)})^\dagger)$. 
\end{proof}

\subsection{Proofs of \Cl{claim:trivialsoundness} and \Cl{claim:nontrivialsoundness} (Trivial and Nontrivial Deviation)}
\subsubsection{Necessary Claims}
For both proofs, we require the  three following claims. The first is regarding the state $\sigma_a^k = \tr_{\mathcal{P}_{final}}(\rho_a^k)$ (for $a\in F_q^{3L})$, where $\rho_a^k$ is defined in \Cl{claim:polystateformat} (in equation \ref{eq:finalstate}) and $\sigma_a^k$ is first defined in Claim \ref{claim:statesimplify}. The claim below considers the unauthenticated version of the state ($\sigma_a$). In other words, if $\sigma_a'$ is $\sigma_a$ with $m-1$ auxiliary 0 qudits appended to each individual qudit, then:
\begin{equation}
\sigma_a^k = E_k\sigma_a'E_k^\dagger
\end{equation}
This equality follows from the definition of the encoding circuit (Definition \ref{def:encodingcircuit}).
\begin{claim}\label{cl:finalstatedescription}
For $a\in F_q^{3L}$ and $\sigma_a$ as defined in Claim \ref{claim:statesimplify} we claim that
\begin{equation}
\sigma_a = \frac{1}{q^{3L}}\sum_{l,l'\in F_q^{3L}}\ket{l}\bra{l'}\otimes\tr_{\mathcal{P}_{final}}(\ket{\psi}_{a,l}\bra{\psi}_{a,l'})
\end{equation}
where $\ket{\psi}_{a,a} = U\ket{0}^{\otimes n}$.
\end{claim}
The second claim involves conjugation properties of the encoding circuit $E_k$ (see Definition \ref{def:encodingcircuit}) with respect to trivial Pauli operators: 
\begin{claim}\label{claim:trivialcommute}
For a trivial Pauli operator $P\in\mbP_{|\mathcal{V}_{final}|}$, 
\begin{equation}
(\Pi_{G_a}\otimes(\hat{\Pi}_0)_{\mathcal{F}} )(E_k^\dagger)^{\otimes |\mathcal{V}_{final}|}P E_k^{\otimes |\mathcal{V}_{final}|} = (E_k^\dagger)^{\otimes |\mathcal{V}_{final}|}P E_k^{\otimes |\mathcal{V}_{final}|}(\Pi_{G_a}\otimes(\hat{\Pi}_0)_{\mathcal{F}})
\end{equation}
where $(\hat{\Pi}_0)_{\mathcal{F}}$ is defined in equation \ref{acceptingprojection} as
\begin{equation}
(\hat{\Pi}_0)_{\mathcal{F}} = \ket{1}\bra{1}\otimes\mcI^{\otimes d}\ket{0}\bra{0}^{\otimes d}
\end{equation}
and for $a = (a(1),\ldots,a(3L))\in F_q^{3L}$
\begin{equation}\label{simplifiedprojection}
\Pi_{G_a} = (\ket{a(1)}\bra{a(1)}\otimes\mcI^{\otimes d}\otimes\ket{0}\bra{0}^{\otimes d})\otimes\cdots\otimes (\ket{a(3L)}\bra{a(3L)}\otimes\mcI^{\otimes d}\otimes\ket{0}\bra{0}^{\otimes d})
\end{equation}
\end{claim}
The final claim simplifies the expression for $\tr(\hat{\Pi}_0\sigma_k^P)$:
\begin{claim}\label{cl:simplifyingsoundness}
For all $P = Z^zX^x\in\mbP_{|\mathcal{V}_{final}|}$, and for $\sigma_k^P$ as defined in \Cl{claim:statesimplify} (equation \ref{eq:sigmadef}),
\begin{equation}
\tr(\hat{\Pi}_0\sigma_k^P) = \sum\limits_{a\in F_q^{3L}} \alpha_{P,a}\cdot\tr((\Pi_{G_a}\otimes (\hat{\Pi}_0)_{\mathcal{F}})(E_k^\dagger)^{\otimes |\mathcal{V}_{final}|}X^x\sigma_{a}^k(X^x)^\dagger(E_k)^{\otimes |\mathcal{V}_{final}|})
\end{equation}
where $(\hat{\Pi}_0)_{\mathcal{F}}$ is defined in equation \ref{acceptingprojection} as
\begin{equation}
(\hat{\Pi}_0)_{\mathcal{F}} = \ket{1}\bra{1}\otimes\mcI^{\otimes d}\ket{0}\bra{0}^{\otimes d}
\end{equation}
and for $a = (a(1),\ldots,a(3L))\in F_q^{3L}$
\begin{equation}
\Pi_{G_a} = (\ket{a(1)}\bra{a(1)}\otimes\mcI^{\otimes d}\otimes\ket{0}\bra{0}^{\otimes d})\otimes\cdots\otimes (\ket{a(3L)}\bra{a(3L)}\otimes\mcI^{\otimes d}\otimes\ket{0}\bra{0}^{\otimes d})
\end{equation}
\end{claim}
We now proceed to proving Claims \ref{claim:trivialsoundness} and \ref{claim:nontrivialsoundness}, and then we prove the claims listed above.
\subsubsection{Proof of \Cl{claim:trivialsoundness} (Trivial Deviation)}
\begin{proofof}{\Cl{claim:trivialsoundness}}
Our goal in this proof is to show
\begin{equation}
\frac{1}{2^m}\tr(\hat{\Pi}_0\sum_k\sigma_k^P) \leq \frac{\gamma}{q^{3L}}\sum_{a\in F_q^{3L}}\alpha_{P,a}
\end{equation}
for a trivial Pauli operator $P$ acting on $\mathcal{V}_{final}$. We will show that for all $k$, 
\begin{equation}
\tr(\hat{\Pi}_0\sigma_k^P) \leq \frac{\gamma}{q^{3L}}\sum_{a\in F_q^{3L}}\alpha_{P,a}
\end{equation}
By \Cl{cl:simplifyingsoundness} (and by the fact that $P$ is a trivial Pauli operator and therefore has no $X$ operator), we have
\begin{equation}
\tr(\hat{\Pi}_0\sigma_k^P) = \sum\limits_{a\in F_q^{3L}} \alpha_{P,a}\cdot\tr((\Pi_{G_a}\otimes (\hat{\Pi}_0)_{\mathcal{F}})(E_k^\dagger)^{\otimes |\mathcal{V}_{final}|}\sigma_{a}^k(E_k)^{\otimes |\mathcal{V}_{final}|})
\end{equation}
where (as defined in equation \ref{simplifiedprojection}) for $a = (a(1),\ldots,a(3L))\in F_q^{3L}$
\begin{equation}
\Pi_{G_a} = (\ket{a(1)}\bra{a(1)}\otimes\mcI^{\otimes d}\otimes\ket{0}\bra{0}^{\otimes d})\otimes\cdots\otimes (\ket{a(3L)}\bra{a(3L)}\otimes\mcI^{\otimes d}\otimes\ket{0}\bra{0}^{\otimes d})
\end{equation}
and (as defined in equation \ref{acceptingprojection}):
\begin{equation}
\hat{\Pi}_0 = (\mcI^{\otimes d+1}\otimes\ket{0}\bra{0}^{\otimes d})^{\otimes 3L}_{\mathcal{V}_{L+1}}\otimes (\ket{1}\bra{1}\otimes\mcI^{\otimes d}\otimes\ket{0}\bra{0}^{\otimes d})_{\mathcal{F}} = (\hat{\Pi}_0)_{\mathcal{V}_{L+1}}\otimes(\hat{\Pi}_0)_{\mathcal{F}}
\end{equation}
We note that $\sigma_a^k$ is the density matrix $\sigma_a$ encoded with the signed polynomial code; i.e. if $\sigma_a'$ is the density matrix $\sigma_a$ with $m-1$ auxiliary 0 qudits appended to each individual qudit, we have:
\begin{equation}
\sigma_a^k = E_k^{\otimes|\mathcal{V}_{final}|}\sigma_a' (E_k^{\otimes|\mathcal{V}_{final}|})^\dagger
\end{equation}
It follows that
\begin{eqnarray}
\tr(\hat{\Pi}_0\sigma_k^P)
&=& \sum\limits_{a\in F_q^{3L}} \alpha_{P,a}\cdot\tr((\Pi_{G_a}\otimes (\hat{\Pi}_0)_{\mathcal{F}})\sigma_{a}')
\end{eqnarray}
Observe that the projection $\Pi_{G_a}\otimes (\hat{\Pi}_0)_{\mathcal{F}}$ does not alter the auxiliary 0 qudits; it acts as $\mcI^{\otimes d}\otimes\ket{0}\bra{0}^{\otimes d}$ on each set of $m-1$ auxiliary qudits. Therefore, we can trace out all of the auxiliary qubits (and also remove the corresponding operators from the projections). Tracing out the auxiliary qudits from $\sigma_a'$ simply results in $\sigma_a$. $\Pi_{G_a}$ can be replaced by $\ket{a}\bra{a}$, and $(\hat{\Pi}_0)_{\mathcal{F}}$ can be replaced by $\ket{1}\bra{1}$. The resulting expression is:
\begin{equation}
\tr(\hat{\Pi}_0\sigma_k^P) = \sum_{a\in F_q^{3L}}\alpha_{P,a}\cdot\tr(\ket{a1}\bra{a1}\sigma_a)
\end{equation}
\Cl{cl:finalstatedescription} implies that the state $\sigma_a$ can be written as
\begin{equation}
\sigma_a = \frac{1}{q^{3L}}\sum_{l,l'\in F_q^{3L}}\ket{l}\bra{l'}\otimes\tr_{\mathcal{P}_{final}}(\ket{\psi}_{a,l}\bra{\psi}_{a,l'})
\end{equation}
where $\ket{\psi}_{a,a} = U\ket{0}^{\otimes n}$. Given this information about $\sigma_a$, we can continue:
\begin{eqnarray}
\tr(\hat{\Pi}_0\sigma_k^P) &=& \sum_{a\in F_q^{3L}}\alpha_{P,a}\cdot\tr(\ket{a1}\bra{a1}\sigma_a)\\
&=& \frac{1}{q^{3L}}\sum_{a,l,l'\in F_q^{3L}}\alpha_{P,a}\cdot\tr(\ket{a}\bra{a}\ket{l}\bra{l'})\tr(\ket{1}\bra{1}(\tr_{\mathcal{P}_{final}}(\ket{\psi}_{a,l}\bra{\psi}_{a,l'})))\\
&=& \frac{1}{q^{3L}}\sum_{a\in F_q^{3L}}\alpha_{P,a}\cdot\tr(\ket{1}\bra{1}(\tr_{\mathcal{P}_{final}}(U\ket{0}\bra{0}^{\otimes n}U^\dagger)))\label{eq:fortrivialcorol}\\
&\leq& \frac{\gamma}{q^{3L}}\sum_{a\in F_q^{3L}}\alpha_{P,a}
\end{eqnarray}
The last inequality follows because measuring $U\ket{0}$ results in 0 with probability $\geq 1 - \gamma$ and 1 otherwise. 
\begin{remark}\label{remark:fortrivialcorol}
Note that if $U\ket{0}$ resulted in 0 with probability exactly $1 - \gamma$, the last inequality would be replaced by an equality (which would replace the inequality in the statement of the claim with an equality). Also note that if in equation \ref{eq:fortrivialcorol} the projection $\ket{1}\bra{1}$ was replaced by $\mcI$, the next line would be the same, except $\gamma$ would be replaced by 1 and the inequality by an equality. These two facts will be useful in the proof of Claim \ref{cl:trivialdensitymatrix}, which is required for the proof of the polynomial version of Corollary \ref{corol:confidence} (Corollary \ref{corol:polyconfidence})
\end{remark}
\end{proofof}


\subsubsection{Proof of \Cl{claim:nontrivialsoundness} (Nontrivial Deviation)}
We now prove \Cl{claim:nontrivialsoundness}.

\begin{proofof}{\Cl{claim:nontrivialsoundness}}
Our goal in this proof is to show that for a non trivial Pauli operator $P = Z^zX^x\in\mbP_{|\mathcal{V}_{final}|}$,
\begin{equation}
\frac{1}{2^m}\tr(\hat{\Pi}_0\sum_k\sigma_k^P)\leq\frac{1}{q^{3L}2^{m-1}}\sum_{a\in F_q^{3L}}\alpha_{P,a}
\end{equation}
To do this, we will show that at most 2 terms in the above sum over $k$ can be non zero; each of those terms can be at most
\begin{equation}
\frac{1}{2^mq^{3L}}\sum_{a\in F_q^{3L}}\alpha_{P,a}
\end{equation}
The claim follows. We begin by using \Cl{cl:simplifyingsoundness} to write:
\begin{equation}\label{eq:startingstatetrivial}
\tr(\hat{\Pi}_0\sigma_k^P) = \sum\limits_{a\in F_q^{3L}} \alpha_{P,a}\cdot\tr((\Pi_{G_a}\otimes (\hat{\Pi}_0)_{\mathcal{F}})(E_k^\dagger)^{\otimes |\mathcal{V}_{final}|}X^x\sigma_{a}^k(X^x)^\dagger(E_k)^{\otimes |\mathcal{V}_{final}|})
\end{equation}
where (as defined in equation \ref{simplifiedprojection}) for $a = (a(1),\ldots,a(3L))\in F_q^{3L}$
\begin{equation}
\Pi_{G_a} = (\ket{a(1)}\bra{a(1)}\otimes\mcI^{\otimes d}\otimes\ket{0}\bra{0}^{\otimes d})\otimes\cdots\otimes (\ket{a(3L)}\bra{a(3L)}\otimes\mcI^{\otimes d}\otimes\ket{0}\bra{0}^{\otimes d})
\end{equation}
and (as defined in equation \ref{acceptingprojection}):
\begin{equation}
\hat{\Pi}_0 = (\mcI^{\otimes d+1}\otimes\ket{0}\bra{0}^{\otimes d})^{\otimes 3L}_{\mathcal{V}_{L+1}}\otimes (\ket{1}\bra{1}\otimes\mcI^{\otimes d}\otimes\ket{0}\bra{0}^{\otimes d})_{\mathcal{F}} = (\hat{\Pi}_0)_{\mathcal{V}_{L+1}}\otimes(\hat{\Pi}_0)_{\mathcal{F}}
\end{equation}
Note that the projection $\Pi_{G_a}\otimes (\hat{\Pi}_0)_{\mathcal{F}}$ includes the projection of each block of $m$ qudits onto $\mcI^{\otimes d+1}\otimes\ket{0}\bra{0}^{\otimes d}$; it can be written as:
\begin{equation}
 \Pi_{G_a}\otimes (\hat{\Pi}_0)_{\mathcal{F}} = \hat{\Pi}_0^L\hat{\Pi}_0^A
\end{equation}
where
\begin{equation}\label{eq:corollaryprojection}
\hat{\Pi}_0^A = (\mcI^{\otimes d+1}\otimes\ket{0}\bra{0}^{\otimes d})^{\otimes |\mathcal{V}_{final}|}
\end{equation}
Intuitively, this implies that the non trivial Pauli operator $P$ must preserve the authenticated state (up to trivial operators) on every block of $m$ qudits in order end up in the subspace defined by $\hat{\Pi}_0^A$. Using  similar reasoning as used in the proof of Lemma \ref{lem:PauliPolyAuthSec}, we should be able to say that $P$ can only do this for at most 2 sign keys at a time. This intuition is formalized in Corollary \ref{soundnesstrivialpaulis}, which follows from Lemma \ref{lem:PauliPolyAuthSec} and is proven immediately after this proof:
\begin{corol}\label{soundnesstrivialpaulis}
For a non trivial Pauli operator $X^x\in\mbP_{tm}$ and a density matrix $\sigma$ on $m$ qudits, there exist at most 2 sign keys $k\in\{-1,1\}^m$ (which are the same regardless of $\sigma$) for which the following expression
\begin{equation}
(\mcI^{\otimes d+1}\otimes\ket{0}\bra{0}^{\otimes d})^{\otimes tm}(E_k^\dagger)^{\otimes t} X^x \sigma^k (X^x)^\dagger E_k^{\otimes t} (\mcI^{\otimes d+1}\otimes\ket{0}\bra{0}^{\otimes d})^{\otimes tm}
\end{equation}
can be non zero. If $X^x = X^{x_1}\otimes\cdots\otimes X^{x_t}$ for $x_i\in F_q^{m}$ and the expression above is non zero, $X^{x_i}$ must be $k$-correlated for all $i$.
\end{corol}
This corollary implies that $\tr(\hat{\Pi}_0^A\sigma_k^P)$  (where $\hat{\Pi}_0^A$ is defined in equation \ref{eq:corollaryprojection}) is non zero for at most 2 sign keys. Fix one sign key for which $\tr(\hat{\Pi}_0^A\sigma_k^P)$ is non zero. We will now simplify the expression in equation \ref{eq:startingstatetrivial} for this fixed sign key $k$.
\begin{eqnarray}
\tr(\hat{\Pi}_0\sigma_k^P)&=& \sum\limits_{a\in F_q^{3L}} \alpha_{P,a}\cdot\tr((\Pi_{G_a}\otimes (\hat{\Pi}_0)_{\mathcal{F}})(E_k^\dagger)^{\otimes |\mathcal{V}_{final}|}X^x\sigma_{a}^k(X^x)^\dagger(E_k)^{\otimes |\mathcal{V}_{final}|})
\end{eqnarray}
Now due to Corollary \ref{soundnesstrivialpaulis}, we know that $X^x$ is $k$-correlated (see Definition \ref{correlatedx}). Because $X^x$ maps an authenticated state to a different authenticated state, it is by definition equal to a logical Pauli operator $\tilde{X}^{x_k}$, so it maps $\sigma_a^k$ to $\sigma_{a,x_k}^k = (X^{x_k}\sigma_a(X^{x_k})^\dagger)^k$:
\begin{eqnarray}
\ldots &=& \sum\limits_{a\in F_q^{3L}} \alpha_{P,a}\cdot\tr((\Pi_{G_a}\otimes (\hat{\Pi}_0)_{\mathcal{F}})(E_k^\dagger)^{\otimes |\mathcal{V}_{final}|}\sigma_{a,x_k}^k(E_k)^{\otimes |\mathcal{V}_{final}|})\\
&=& \sum\limits_{a\in F_q^{3L}} \alpha_{P,a}\cdot\tr(
(\Pi_{G_a}\otimes(\hat{\Pi}_0)_{\mathcal{F}}))\sigma_{a,x_k}'(\Pi_{G_a}\otimes\mcI_{\mathcal{F}}))
\end{eqnarray}
where the equality follows because $\sigma_{a,x_k}^k$ is the density matrix $\sigma_{a,x_k}$ encoded with the signed polynomial code, and $\sigma_{a,x_k}'$ is the density matrix $\sigma_{a,x_k}$ with $m-1$ auxiliary 0 qudits appended to each qudit of $\sigma_{a,x_k}$. Note that the projection $\Pi_{G_a}\otimes(\hat{\Pi}_0)_{\mathcal{F}}$ does not alter the auxiliary 0 qudits, as it acts on each set of $m-1$ 0 qudits as $\mcI^{\otimes d}\otimes\ket{0}\bra{0}^{\otimes d}$. Then we can trace out all the auxiliary 0 qudits and also remove them from the projection $\Pi_{G_a} \otimes (\hat{\Pi}_0)_{\mathcal{F}}$. The projection is now simply $\ket{a}\bra{a}\otimes\ket{1}\bra{1}$ and we have:
\begin{eqnarray}
\tr(\hat{\Pi}_0\sigma_k^P) &=& \frac{1}{2^m}\sum\limits_{a\in F_q^{3L}} \alpha_{P,a}\cdot\tr(
(\ket{a1}\bra{a1})X^{x_k}\sigma_{a} (X^{x_k})^\dagger)
\end{eqnarray}
Now we use \Cl{cl:finalstatedescription} to write $\sigma_a$ as
\begin{equation}
\sigma_a = \frac{1}{q^{3L}}\sum_{l,l'\in F_q^{3L}}\ket{l}\bra{l'}\otimes\tr_{\mathcal{P}_{final}}(\ket{\psi}_{a,l}\bra{\psi}_{a,l'})
\end{equation}
where $\ket{\psi}_{a,a} = U\ket{0}^{\otimes n}$. Let $\sigma_{a,l,l'} = \tr_{\mathcal{P}_{final}}(\ket{\psi}_{a,l}\bra{\psi}_{a,l'})$ and let $X^{x_k} = X^{x_k}_{\mathcal{V}_{L+1}}\otimes X^{x_k}_{\mathcal{F}}$. Now we have:
\begin{eqnarray}
\tr(\hat{\Pi}_0\sigma_k^P) &=& \frac{1}{q^{3L}}\sum\limits_{a,l,l'\in F_q^{3L}} \alpha_{P,a}\cdot\tr(
(\ket{a1}\bra{a1})X^{x_k}(\ket{l}\bra{l'}\otimes\sigma_{a,l,l'}) (X^{x_k})^\dagger)\\
&=& \frac{1}{q^{3L}}\sum\limits_{a,l,l'\in F_q^{3L}} \alpha_{P,a}\cdot\tr
(\ket{a}\bra{a}(X^{x_k}_{\mathcal{V}_{L+1}}\ket{l}\bra{l'}(X^{x_k}_{\mathcal{V}_{L+1}})^\dagger))\tr(\ket{1}\bra{1} X^{x_k}_{\mathcal{F}}\sigma_{a,l,l'}(X^{x_k}_{\mathcal{F}})^\dagger))\nonumber\\
&=& \frac{1}{q^{3L}}\sum\limits_{a\in F_q^{3L}} \alpha_{P,a}\cdot\tr(\ket{1}\bra{1} X^{x_k}_{\mathcal{F}}\sigma_{a,a-x_k,a-x_k}(X^{x_k}_{\mathcal{F}})^\dagger))\\\label{eq:nontrivialclaimprojection}
&\leq& \frac{1}{q^{3L}}\sum\limits_{a\in F_q^{3L}} \alpha_{P,a}
\end{eqnarray}
where we have obtained the second to last equality because $l= l' = a-x_k$ in order for $\tr
(\ket{a}\bra{a}(X^{x_k}\ket{l}\bra{l'}(X^{x_k})^\dagger))$ to be 1 (otherwise the trace will be 0). The last equality is obtained because $\sigma_{a,a-x_k,a-x_k}$ is a density matrix. 
\begin{remark}\label{remark:claimcorol}
Note that the final inequality would still hold if the projection $\ket{1}\bra{1}$ was replaced by $\mcI$; this fact will be useful in the proof of Claim \ref{cl:nontrivialdensitymatrix}, which is required for the proof of the polynomial version of Corollary \ref{corol:confidence} (Corollary \ref{corol:polyconfidence}). 
\end{remark}
\end{proofof}

\begin{proofof}{Corollary \ref{soundnesstrivialpaulis}}
Let $X^x = X^{x_1}\otimes\cdots\otimes X^{x_t}$. We show that if there exists an $i$ for which  $X^{x_i}$ is not $k$-correlated,
\begin{equation}
(\mcI^{\otimes d+1}\otimes\ket{0}\bra{0}^{\otimes d})^{\otimes tm}(E_k^\dagger)^{\otimes t} X^x \sigma^k (X^x)^\dagger E_k^{\otimes t} (\mcI^{\otimes d+1}\otimes\ket{0}\bra{0}^{\otimes d})^{\otimes tm}
\end{equation}
\begin{equation}\label{eqsoundnesstrivialpaulis}
=(\mcI^{\otimes d+1}\otimes\ket{0}\bra{0}^{\otimes d})^{\otimes tm}((E_k^\dagger)^{\otimes t} X^x E_k^{\otimes t})\sigma((E_k^\dagger)^{\otimes t} X^x E_k^{\otimes t})^\dagger (\mcI^{\otimes d+1}\otimes\ket{0}\bra{0}^{\otimes d})^{\otimes tm} = 0
\end{equation}
where the equality follows due to the fact that $\sigma^k$ is the density matrix $\sigma$ encoded with the signed polynomial code. It follows by Claim \ref{correlatedx} that the expression can be non zero for at most 2 sign keys, which completes the proof. Note that it is implied that the choice of $k$ for which the expression is non zero is independent of $\sigma$; it is dependent only on whether $X^x$ is $k$-correlated or not. Assume there exists an $i$ for which $X^{x_i}$ is not $k$-correlated. By \Cl{correlateddecomposition} we can break down $X^{x_i}$ into a product of a $k$-correlated operator $Q_k$ and an uncorrelated operator $\hat{Q}_k$. By equation \ref{equation:uncorrelatedform}, we know that 
$\hat{Q}_k$ can be written as
\begin{equation}
\hat{Q}_k = \mcI^{\otimes d+1}\otimes X^{\hat{x_i}_{d+2}}\otimes\cdots\otimes X^{\hat{x_i}_m}
\end{equation}
where $(\hat{x_i}_{d+2},\ldots,\hat{x_i}_m)\neq 0^d$ since $X^{x_i}$ is not $k$-correlated. Now we can refer to equation \ref{conjugationproperties} to write:
\begin{equation}\label{uncorrelatedconjugation}
E_k^\dagger \hat{Q}_k E_k = \mcI^{\otimes d+1}\otimes X^{\hat{x_i}_{d+2}}\otimes\cdots\otimes X^{\hat{x_i}_m}
\end{equation}
Returning to the expression above (equation \ref{eqsoundnesstrivialpaulis}), but just the leftmost part of the expression which operates on the $i^{th}$ register of $\sigma$, we have:
\begin{eqnarray}
(\mcI^{\otimes d+1}\otimes\ket{0}\bra{0}^{\otimes d})E_k^\dagger X^{x_i} E_k &=& (\mcI^{\otimes d+1}\otimes\ket{0}\bra{0}^{\otimes d})E_k^\dagger \hat{Q}_k Q_k E_k\\
&=& (\mcI^{\otimes d+1}\otimes\ket{0}\bra{0}^{\otimes d})E_k^\dagger \hat{Q}_k E_k E_k^\dagger Q_k E_k
\end{eqnarray}
Plugging in the expression for $E_k^\dagger\hat{Q}_k E_k$ from equation \ref{uncorrelatedconjugation}, we have
\begin{eqnarray}\label{eq:beforeinnerproduct}
\ldots &=& (\mcI^{\otimes d+1}\otimes\ket{0}\bra{0}^{\otimes d})(\mcI^{\otimes d+1}\otimes X^{\hat{x_i}_{d+2}}\otimes\cdots\otimes X^{\hat{x_i}_m})E_k^\dagger Q_k E_k
\end{eqnarray}
Observe that $Q_k$ is $k$-correlated; by definition, it preserves a state authenticated with a sign key. Then the rightmost part of equation \ref{eq:beforeinnerproduct} is simply $E_k^\dagger$ acting on an authenticated state. When $E_k^\dagger$ acts on an arbitrary authenticated density matrix, it performs the following map:
\begin{equation}
E_k^\dagger \sigma^k = \sigma\otimes\ket{0}\bra{0}^{m-1}
\end{equation}
It follows that the above expression (equation \ref{eq:beforeinnerproduct}) contains the inner product below:
\begin{equation}
\bra{0^d}(X^{\hat{x_i}_{d+2}}\otimes\cdots\otimes X^{\hat{x_i}_m})\ket{0^d}
\end{equation}
which must be equal to 0 since $(\hat{x_i}_{d+2},\ldots,\hat{x_i}_m)\neq 0^d$. 
\end{proofof}

\subsubsection{Proofs of Necessary Claims}
We begin with \Cl{cl:finalstatedescription}.

\begin{proofof}{\Cl{cl:finalstatedescription}}
Recall that we are considering the unauthenticated state $\sigma_a$ (defined in \Cl{claim:statesimplify}), where $\sigma_a = \tr_{\mathcal{P}_{final}}(\rho_a)$, for $a = (a_1,\ldots,a_L)$ ($a_i\in F_q^3$) and, as defined in \Cl{claim:polystateformat} (in equation \ref{eq:finalstate}), 
\begin{equation}
\rho_a = (Q_{L}C_{a_L}\cdots Q_1C_{a_1}Q_0)\rho (Q_{L}C_{a_L}\cdots Q_1C_{a_1}Q_0)^\dagger
\end{equation}
Note that $\rho_a$ is a pure state and can be written as $\ket{\psi_a}\bra{\psi_a}$ (since $\rho$ consists of $n$ 0 qudits and $L$ resource states). Fact \ref{fact:teleportationstate} states that $\ket{\psi_a}$ can be written as:
\begin{equation}
\ket{\psi_a} = \frac{1}{\sqrt{q^{3L}}}\sum_{l\in F_q^{3L}}\ket{l}\ket{\psi}_{a,l}
\end{equation}
where $\ket{\psi}_{a,a} = U\ket{0}^{\otimes n}$ (recall that $U$ is the \qc\ instance which the prover was asked to apply). It follows that the state $\sigma_a$ can be written as
\begin{equation}
\sigma_a = \tr_{\mathcal{P}_{final}}(\ket{\psi_a}\bra{\psi_a}) = \frac{1}{q^{3L}}\sum_{l,l'\in F_q^{3L}}\ket{l}\bra{l'}\otimes\tr_{\mathcal{P}_{final}}(\ket{\psi}_{a,l}\bra{\psi}_{a,l'})
\end{equation}
\end{proofof}
We proceed to proving Claim \ref{claim:trivialcommute}.

\begin{proofof}{\Cl{claim:trivialcommute}}
Recall that we would like to prove
\begin{equation}
(\Pi_{G_a}\otimes(\hat{\Pi}_0)_{\mathcal{F}} )(E_k^\dagger)^{\otimes |\mathcal{V}_{final}|}P E_k^{\otimes |\mathcal{V}_{final}|} = (E_k^\dagger)^{\otimes |\mathcal{V}_{final}|}P E_k^{\otimes |\mathcal{V}_{final}|}(\Pi_{G_a}\otimes(\hat{\Pi}_0)_{\mathcal{F}})
\end{equation}
for a trivial Pauli operator $P\in\mbP_{|\mathcal{V}_{final}|}$, where $\Pi_{G_a}$ is defined in equation \ref{simplifiedprojection} and $(\hat{\Pi}_0)_{\mathcal{F}}$ is defined in \ref{acceptingprojection}. In other words, we want to show that the Pauli operator $(E_k^\dagger)^{\otimes |\mathcal{V}_{final}|}P E_k^{\otimes |\mathcal{V}_{final}|}$ commutes with the projection $\Pi_{G_a}\otimes(\hat{\Pi}_0)_{\mathcal{F}}$ when $P$ is trivial. Observe that
\begin{equation}
(E_k^\dagger)^{\otimes |\mathcal{V}_{final}|}P E_k^{\otimes |\mathcal{V}_{final}|}
\end{equation}
must be trivial in the registers $1,d+2,\ldots,m$; only registers $d+2,\ldots,m$ can be non trivial. This follows from the definition of $E_k$ (see Section \ref{sec:encodingcircuit}). In more detail, $E_k^\dagger$ consists of \textit{SUM} and multiplication operations (which compose $D_k^\dagger$) followed by Fourier operations. As shown in Section \ref{sec:conjugationproperties} (in equations \ref{sumconjugation} and \ref{multiplyconjugation}), the \textit{SUM} and multiplication operators in $D_k^\dagger$ map trivial operators to trivial operators by conjugation. The Fourier operators, which flip $Z$ and $X$ operators, act only on registers $2,\ldots,d+1$, so only these registers can be mapped to non trivial operators. 

Note that trivial operators commute with standard basis projections, and $\Pi_{G_a}\otimes(\hat{\Pi}_0)_{\mathcal{F}}$ acts with standard basis projections on registers $1,d+2,\ldots,m$ for each block of $m$ registers. Since $\Pi_{G_a}\otimes(\hat{\Pi}_0)_{\mathcal{F}}$ acts as $\mcI$ on registers $2,\ldots,d+1$ for each block of $m$ registers, the non trivial portion of $(E_k^\dagger)^{\otimes |\mathcal{V}_{final}|}P E_k^{\otimes |\mathcal{V}_{final}|}$ also commutes with $\Pi_{G_a}\otimes(\hat{\Pi}_0)_{\mathcal{F}}$. 
\end{proofof}

Finally, we prove \Cl{cl:simplifyingsoundness}.

\begin{proofof}{\Cl{cl:simplifyingsoundness}}
Recall that our goal is to prove the following equality for all $P = Z^zX^x\in \mbP_{|\mathcal{V}_{final}|}$:
\begin{equation}
\tr(\hat{\Pi}_0\sigma_k^P) = \sum\limits_{a\in F_q^{3L}} \alpha_{P,a}\cdot\tr((\Pi_{G_a}\otimes (\hat{\Pi}_0)_{\mathcal{F}})(E_k^\dagger)^{\otimes |\mathcal{V}_{final}|}X^x\sigma_{a}^k(X^x)^\dagger(E_k)^{\otimes |\mathcal{V}_{final}|})
\end{equation}
where $\sigma_k^P$ was defined as follows in \Cl{claim:statesimplify}:
\begin{equation}
\sigma_k^P = \sum\limits_{\Delta_L\in F_q^{|\mathcal{V}_{L+1}|}} \alpha_{P,g(\Delta_L)}\cdot((\ket{\Delta_L}\bra{\Delta_L}\otimes\mcI_{\mathcal{F}})(D_k^\dagger)^{\otimes|\mathcal{V}_{final}|}P)\sigma_{g(\Delta_L)}^k((\ket{\Delta_L}\bra{\Delta_L}\otimes\mcI_{\mathcal{F}})(D_k^\dagger)^{\otimes|\mathcal{V}_{final}|}P)^\dagger
\end{equation}
and the projection $\hat{\Pi}_0$ acting on $\mathcal{V}_{final} = \mathcal{V}_{L+1}\cup F$ was defined in equation \ref{acceptingprojection} as follows:
\begin{equation}
\hat{\Pi}_0 = (\mcI^{\otimes d+1}\otimes\ket{0}\bra{0}^{\otimes d})^{\otimes 3L}_{\mathcal{V}_{L+1}}\otimes (\ket{1}\bra{1}\otimes\mcI^{\otimes d}\otimes\ket{0}\bra{0}^{\otimes d})_{\mathcal{F}} = (\hat{\Pi}_0)_{\mathcal{V}_{L+1}}\otimes(\hat{\Pi}_0)_{\mathcal{F}}
\end{equation}
Before continuing, note that the projection $\hat{\Pi}_0$ makes it unnecessary to sum over all $\Delta_L$ in the expression for $\sigma_k$; we can instead sum over a subset of $\Delta_L$, and split the sum according to the value $a\in F_q^{3L}$ of $g(\Delta_L)$:
\begin{equation}\label{eqduringsimplification}
\tr(\hat{\Pi}_0\sigma_k^P) = \tr(\hat{\Pi}_0\sum\limits_{\substack{a\in F_q^{3L}\\\Delta_L\in G_a}} \alpha_{P,a}\cdot((\ket{\Delta_L}\bra{\Delta_L}\otimes\mcI_{\mathcal{F}})(D_k^\dagger)^{\otimes|\mathcal{V}_{final}|}P)\sigma_{a}^k((\ket{\Delta_L}\bra{\Delta_L}\otimes\mcI_{\mathcal{F}})(D_k^\dagger)^{\otimes|\mathcal{V}_{final}|}P)^\dagger)
\end{equation}
where 
\begin{equation}\label{goodmeasurementspace}
G_a \EqDef \{((s_1,0^{d}),\ldots,(s_{3L},0^{d}))|s_1,\ldots,s_{3L}\in F_q^{d+1},g(s_1,\ldots,s_{3L}) = a\}
\end{equation}
Instead of summing over all $\Delta_L\in F_q^{|\mathcal{V}_{L+1}|}$, we have restricted to $\Delta_L\in \cup_a G_a$. This is because $\Delta_L$ must equal $((s_1,0^{d}),\ldots,(s_{3L},0^{d}))$ to give a non zero trace when projected onto $\hat{\Pi}_0$. Next, since $\ket{\Delta_L}\bra{\Delta_L}\otimes\mcI_{\mathcal{F}}$ commutes with $\hat{\Pi}_0$, we can remove $\ket{\Delta_L}\bra{\Delta_L}\otimes\mcI_{\mathcal{F}}$ from the right hand side of equation \ref{eqduringsimplification} (due to the cyclic nature of trace), obtaining:
\begin{eqnarray}
\tr(\hat{\Pi}_0\sigma_k^P) &=& \tr(\hat{\Pi}_0\sum\limits_{\substack{a\in F_q^{3L}\\\Delta_L\in G_a}} \alpha_{P,a}\cdot((\ket{\Delta_L}\bra{\Delta_L}\otimes\mcI_{\mathcal{F}})(D_k^\dagger)^{\otimes|\mathcal{V}_{final}|}P)\sigma_{a}^k((D_k^\dagger)^{\otimes|\mathcal{V}_{final}|}P)^\dagger)\\
&=& \tr(\hat{\Pi}_0\sum\limits_{a\in F_q^{3L}} \alpha_{P,a}\cdot((\Pi_{G_a}\otimes\mcI_{\mathcal{F}})(D_k^\dagger)^{\otimes|\mathcal{V}_{final}|}P)\sigma_{a}^k((D_k^\dagger)^{\otimes|\mathcal{V}_{final}|}P)^\dagger)\label{restrictedsum}
\end{eqnarray}
where for
$a = (a(1),\ldots,a(3L))\in F_q^{3L}$ 
\begin{equation}
\Pi_{G_a} = \sum_{\Delta_L\in G_a}\ket{\Delta_L}\bra{\Delta_L} = (\ket{a(1)}\bra{a(1)}\otimes\mcI^{\otimes d}\otimes\ket{0}\bra{0}^{\otimes d})\otimes\cdots\otimes (\ket{a(3L)}\bra{a(3L)}\otimes\mcI^{\otimes d}\otimes\ket{0}\bra{0}^{\otimes d})
\end{equation}
We can further simplify to:
\begin{equation}
\tr(\hat{\Pi}_0\sigma_k^P) = \sum\limits_{a\in F_q^{3L}} \alpha_{P,a}\cdot\tr((\Pi_{G_a}\otimes (\hat{\Pi}_0)_{\mathcal{F}})(D_k^\dagger)^{\otimes |\mathcal{V}_{final}|}P\sigma_{a}^kP^\dagger(D_k)^{\otimes |\mathcal{V}_{final}|})
\end{equation}
because
\begin{equation}
(\hat{\Pi}_0)_{\mathcal{V}_{L+1}}\Pi_{G_a} = \Pi_{G_a}
\end{equation}
Recall (from \ref{equation:encodingcircuit}) that
\begin{equation}
D_k^\dagger = (\mcI\otimes F^{\otimes d}\otimes \mcI^{\otimes d}) E_k^\dagger
\end{equation}
It follows that:
\begin{eqnarray}\label{equation:switchtorealdecoding}
(\Pi_{G_a}\otimes (\hat{\Pi}_0)_{\mathcal{F}})(D_k^\dagger)^{\otimes |\mathcal{V}_{final}|} &=& (\Pi_{G_a}\otimes (\hat{\Pi}_0)_{\mathcal{F}})((\mcI\otimes F^{\otimes d}\otimes \mcI^{\otimes d})^\dagger)^{\otimes |\mathcal{V}_{final}|}(E_k^\dagger)^{\otimes |\mathcal{V}_{final}|}\\
&=& ((\mcI\otimes F^{\otimes d}\otimes \mcI^{\otimes d})^\dagger)^{\otimes |\mathcal{V}_{final}|}(\Pi_{G_a}\otimes (\hat{\Pi}_0)_{\mathcal{F}})(E_k^\dagger)^{\otimes |\mathcal{V}_{final}|}\label{switchdecoding}
\end{eqnarray}
The commutation in the final equality occurs because of the structure of $(\Pi_{G_a}\otimes (\hat{\Pi}_0)_{\mathcal{F}})$; for each set of $m$ registers, it acts as identity on registers $2,\ldots, d+1$ in the set. Plugging in the equality obtained above, we obtain:
\begin{eqnarray}
\tr(\hat{\Pi}_0\sigma_k^P)
&=& \sum\limits_{a\in F_q^{3L}} \alpha_{P,a}\cdot\tr(((\Pi_{G_a}\otimes(\hat{\Pi}_0)_{\mathcal{F}})(E_k^\dagger)^{\otimes|\mathcal{V}_{final}|}P)\sigma_{a}^k((E_k^\dagger)^{\otimes|\mathcal{V}_{final}|}P)^\dagger)
\end{eqnarray}
Note that we have removed the terms corresponding to $((\mcI\otimes F^{\otimes d}\otimes \mcI^{\otimes d})^\dagger)^{\otimes |\mathcal{V}_{final}|}$; these terms canceled due to the cyclic nature of trace, since they were present on both ends of the above expression.

To complete the claim, we need to show that $P = Z^zX^x$ can be replaced by $X^x$. To begin, observe that in the expression for $\sigma_k^P$, we can replace
\begin{equation}
(E_k^\dagger)^{\otimes |\mathcal{V}_{final}|}P 
\end{equation}
with
\begin{equation}
(E_k^\dagger)^{\otimes |\mathcal{V}_{final}|}Z^z E_k^{\otimes |\mathcal{V}_{final}|} (E_k^\dagger)^{\otimes |\mathcal{V}_{final}|}X^x
\end{equation}
By \Cl{claim:trivialcommute}, we know that the Pauli operator $(E_k^\dagger)^{\otimes |\mathcal{V}_{final}|}Z^z E_k^{\otimes |\mathcal{V}_{final}|}$ can be commuted past the projection $\Pi_{G_a}\otimes(\hat{\Pi}_0)_{\mathcal{F}}$. Note that since we have now pulled the Pauli operator 
\begin{equation}
(E_k^\dagger)^{\otimes |\mathcal{V}_{final}|}Z^z E_k^{\otimes |\mathcal{V}_{final}|}
\end{equation}
past the projection, we can remove it from the expression, due to the cyclic nature of trace.
\end{proofof}

\section{Blind \QPIP}\label{app:blind}
In this section, we will prove the following theorem: 

\begin{statement}{\Th{thm:blind}} \textit{\Th{thm:qcircuit} holds also in a blind setting,
namely, the prover does not get any
 information regarding the function being computed and its input.}
\end{statement}\newline
To begin, we define blindness: 
\begin{deff} \cite{blind,broadbent2008ubq,childs2001saq}
  Secure blind quantum computation is a
  process where a server computes a function for a client and the following
  properties hold:
\begin{itemize}
\item \textbf{Blindness}: The prover gets no information beyond an upper bound on the size of the circuit.
    Formally, in a blind computation scheme for a set of circuits  $\mathfrak{C}_n$ which take as input strings in $\{0,1\}^n$, the prover's reduced density matrix is identical for every
      $C\in\mathfrak{C}_n$ and input $x\in\{0,1\}^n$.
\item \textbf{Security}: Completeness and soundness hold the same way as in
      \QAS\ (\Def{def:qas}).
\end{itemize}

\end{deff}
We use the \QPIP\ protocols for \qc\ in order to provide a blind \QPIP\ for
any language in \BQP. To do this, we require a universal circuit. Roughly, a universal circuit acts on
input bits and control bits. The control bits can be thought of as a description of a circuit that should be applied to the input bits. Universal circuits can be formally defined as follows:
\begin{deff} For a circuit $U$ acting on $n$ qubits, let $c(U)\in\{0,1\}^k$ be the canonical (classical) description of $U$. The universal circuit $\mathfrak{U}_{n,k}$ acts in the
  following way:
  \begin{eqnarray}
    \mathfrak{U}_{n,k} \ket\phi\otimes\ket{c(U)}
    \longrightarrow U\ket\phi\ket{c(U)}
  \end{eqnarray}
\end{deff}
Constructing such a circuit is an easy exercise. We would like the universal circuit to simulate any circuit made of Toffoli and Hadamard gates on $n$ qubits - it is well known that such circuits are quantum universal. Assume there is an upper bound of $m$ gates in the circuit. The universal circuit is split up into $m$ layers. Each layer $i$ consists of every possible gate on the $n$ input qubits, and each such gate is controlled by 1 qubit. Only 1 of these control qubits will be set to 1, based on which is the $i^{th}$ gate applied in $U$. 
\\~\\
To perform a blind computation, the verifier will compute, with the
prover's help, the result of the universal circuit acting on input and control
bits. It follows that to prove blindness, we need to show that the input to the universal circuit is hidden from the prover. To do this, we show that the prover's density matrix in both the Clifford and polynomial schemes remains independent of the input (to the universal circuit) throughout the computation. We begin with the Clifford scheme.

\begin{statement}{\Th{thm:blindclifford} (Blindness of the Clifford based \QPIP)}\textit{The state of the prover in the Clifford based \QPIP\ (Protocol \ref{prot:CliffordIP}) is independent of the input to the circuit which is being computed throughout the protocol.}\end{statement}

\begin{proofof}{\Th{thm:blindclifford}}
We do not need to consider the prover's extra space, since that contains no information about the input. We need only consider the qubits sent to the prover by the verifier. Whenever the prover receives a state from the verifier (at the beginning of the protocol and during the protocol), the verifier has chosen new, independent keys to encode the state using the Clifford \QAS. Therefore, each density matrix sent to the prover by the verifier is the maximally mixed state (by \Le{cliffordmix}) regardless of the input. This remains true throughout the protocol, since when the prover sends a register to the verifier, the verifier returns a completely mixed state, independent of the remaining registers. It follows that the prover's state (other than his extra space) is always described by the completely mixed state.
\end{proofof}
\\~\\
We now consider the polynomial \QPIP.

\begin{statement}{\Th{thm:blindpolynomial} (Blindness of the Polynomial Based \QPIP)}\textit{The state of the prover in the polynomial based \QPIP\ (Protocol \ref{prot:PolynomialIP}) remains independent of the input to the circuit which is being computed throughout the protocol.}\end{statement}\newline
We remark that the proof of this fact turns out to be rather cumbersome, because we are relying on the randomness provided by the measurement results to prove blindness. In fact, the proof can be greatly simplified by adding additional randomness to the Toffoli states (intuitively, this adds a one time pad to the decoded measurement results sent to the prover). However, it is interesting to note that this additional randomness is not needed for blindness, and that the randomness of the measurement results is indeed enough. 
\\~\\
\begin{proofof}{\Th{thm:blindpolynomial}}
This case is more complicated, due to the classical interaction in each round. Without the classical interaction, the prover's initial state is just the maximally mixed state, due to the Pauli keys (by Lemma \ref{paulimix}), so blindness follows easily in this case. Returning to the case in which there is classical interaction, we need to show that the joint quantum state and classical information of the prover are independent of the input to the computation. 
Recall that each message sent back by the verifier is a decoded measurement result which is required in order to apply the Toffoli gate. We will argue that due to the way the Toffoli gate is applied (see Section \ref{app:toffoli} and Fact \ref{fact:teleportationstate}), the decoded measurement results are uniformly random regardless of the input. This implies that revealing the decoded measurement result leaks no information about the input. We now proceed to prove this formally.
\\~\\
We will show that at the start of the final round (round $L+1$) the prover's state (which includes the classical messages from the verifier) is independent of the input. As given in \Cl{claim:polystateformat} and equation \ref{startoffinalround} the joint state of the prover's registers, $\mathcal{P}_{L+1}$, the environment $\mathcal{E}$, the verifier's registers, $\mathcal{V}_{L+1}$ and the key register containing the sign key and Pauli keys of qudits in $\mathcal{P}_{L+1}$, is:
\begin{equation}
\frac{1}{2^m|\mbP_{m'}|}\sum_{\substack{z,x\in F_q^{|\mathcal{P}_{L+1}|}\\k\in\{-1,1\}^m}}\tau_{L+1}(z,x,k)\otimes\sum_{\Delta_L,z_1,x_1\in F_q^{|\mathcal{V}_{L+1}|}} W_{\Delta_L,\hat{z},\hat{x},k}^{L+1} (\rho^k_{g(\Delta_L)}\otimes\rho_{\mathcal{E}} )(W_{\Delta_L,\hat{z},\hat{x}}^{L+1})^\dagger
\end{equation}
where 
\begin{equation}
\hat{z} = (z_1,z), \hat{x} = (x_1,x),
\end{equation}
\begin{equation}
W_{\Delta_{L},\hat{z},\hat{x},k}^{L+1} = (\ket{\Delta_{L}}\bra{\Delta_{L}}(D_k^\dagger)^{\otimes |\mathcal{V}_{L+1}|}(Z^{z_1}X^{x_1})^\dagger_{\mathcal{V}_{L+1}}\otimes\mcI_{\mathcal{P}_{L+1},E}) U_{g(\Delta_{L})} ((Z^{z_1}X^{x_1})_{\mathcal{V}_{L+1}}\otimes (Z^{z}X^{x})_{\mathcal{P}_{L+1}}\otimes \mcI_{\mathcal{E}})
\end{equation}
where $U_{g(\Delta_{L})}$ is a unitary operator dependent on $g(\Delta_{L})$, $\rho_{\mathcal{E}}$ is the initial state of the prover's environment and 
\begin{equation}
\rho_{g(\Delta_{L})}^k = (\tilde{Q}_{L}\tilde{C}_{g(\delta_{L})}\cdots \tilde{Q}_1\tilde{C}_{g(\delta_1)}\tilde{Q}_0)\rho^k (\tilde{Q}_{L}\tilde{C}_{g(\delta_{L})}\cdots \tilde{Q}_1\tilde{C}_{g(\delta_1)}\tilde{Q}_0)^\dagger
\end{equation}
for $\Delta_{L} = (\delta_1,\ldots,\delta_{L})\in F_q^{|\mathcal{V}_{L+1}|}$ and where $\rho^k$ is the initial density matrix, containing an authentication of the input state on $n$ qudits and authentications of $L$ Toffoli states on 3 qudits each. 
\\~\\
First, since we are only considering the prover's state, we can trace out the verifier's registers $\mathcal{V}_{L+1}$ and the first register containing the keys. This gives the following state:
\begin{equation}
\frac{1}{2^m|\mbP_{m'}|}\sum_{\substack{z,x\in F_q^{|\mathcal{P}_{L+1}|}\\k\in\{-1,1\}^m}}\sum_{\Delta_L,z_1,x_1\in F_q^{|\mathcal{V}_{L+1}|}} \tr_{\mathcal{V}_{L+1}}(W_{\Delta_L,\hat{z},\hat{x},k}^{L+1} (\rho^k_{g(\Delta_L)}\otimes\rho_{\mathcal{E}} )(W_{\Delta_L,\hat{z},\hat{x}}^{L+1})^\dagger)
\end{equation}
Next, the sum over $z,x\in F_q^{|\mathcal{P}_{L+1}|}$ allows us to use the Pauli mixing lemma (\Le{paulimix}) to replace $\rho_{g(\Delta_{L})}^k$ with 
\begin{equation}
 \tr_{\mathcal{P}_{L+1}}(\rho_{g(\Delta_{L})}^k) \otimes\frac{1}{q^{|\mathcal{P}_{L+1}|}}\mcI^{\otimes |\mathcal{P}_{L+1}|}
\end{equation}
We can now rewrite the prover's state as:
\begin{equation}\label{eq:blindnessproversstate}
\frac{1}{2^m|\mbP_{|\mathcal{V}_{L+1}|}|}\sum_{\substack{k\in\{-1,1\}^m\\\Delta_L,z_1,x_1\in F_q^{|\mathcal{V}_{L+1}|}}} \tr_{\mathcal{V}_{L+1}}(W_{\Delta_L,z_1,x_1,k}'  (\sigma_{g(\Delta_L),k})  (W_{\Delta_L,z_1,x_1,k}')^\dagger)
\end{equation}
where 
\begin{equation}
\sigma_{g(\Delta_L),k} =  \tr_{\mathcal{P}_{L+1}}(\rho_{g(\Delta_{L})}^k) \otimes\frac{1}{q^{|\mathcal{P}_{L+1}|}}\mcI^{\otimes |\mathcal{P}_{L+1}|}\otimes\rho_{\mathcal{E}} 
\end{equation}
\begin{equation}
W_{\Delta_{L},z_1,x_1,k}' = (\ket{\Delta_{L}}\bra{\Delta_{L}}(D_k^\dagger)^{\otimes |\mathcal{V}_{L+1}|}(Z^{z_1}X^{x_1})^\dagger_{\mathcal{V}_{L+1}}\otimes\mcI_{\mathcal{P}_{L+1},E}) U_{g(\Delta_{L})} ((Z^{z_1}X^{x_1})_{\mathcal{V}_{L+1}}\otimes \mcI_{\mathcal{P}_{L+1}}\otimes \mcI_{\mathcal{E}})
\end{equation}
Since $\rho_{g(\Delta_L)}^k$ is a pure state, it can be written as $\ket{\psi_{g(\Delta_L)}}^k\bra{\psi_{g(\Delta_L)}}^k$. Recall from Fact \ref{fact:teleportationstate} that the unauthenticated state $\ket{\psi_{g(\Delta_L)}}$ can be written as
\begin{equation}
\ket{\psi_{g(\Delta_{L})}} = \frac{1}{\sqrt{q^{3L}}}\sum_{l\in F_q^{3L}} \ket{l}\ket{\psi}_{g(\Delta_{L}),l}
\end{equation}
Given this form of the unauthenticated state, we revert back to analyzing the authenticated state, which is:
\begin{equation}
\rho_{g(\Delta_L)}^k = (\ket{\psi_{g(\Delta_{L})}}\bra{\psi_{g(\Delta_{L})}})^k = \frac{1}{q^{3L}}\sum_{l,l'\in F_q^{3L}} (\ket{l}\bra{l'})^k\otimes(\ket{\psi}_{g(\Delta_{L}),l}\bra{\psi}_{g(\Delta_{L}),l'})^k
\end{equation}
Since $\mathcal{V}_{L+1}$ corresponds to the first register in the above sum (containing the authentication of $l$) and $\mathcal{P}_{L+1}$ corresponds to the register containing the authentication of $\ket{\psi}_{g(\Delta_{L}),l}$, we have:
\begin{equation}\label{eq:authenticatedtracedoutstate}
\tr_{\mathcal{P}_{L+1}}(\rho_{g(\Delta_{L})}^k) = \frac{1}{q^{3L}}\sum_{l,l'\in F_q^{3L}}\tr( \ket{\psi}_{g(\Delta_{L}),l}\bra{\psi}_{g(\Delta_{L}),l'})^k\cdot(\ket{l}\bra{l'})^k 
\end{equation}
In the following claim (which we prove after this proof), we show that once we plug in the above expression into the prover's state as given in Equation \ref{eq:blindnessproversstate}, the state is only non zero when $l = l'$. To see this, observe that summing over $z_1,x_1\in F_q^{|\mathcal{V}_{L+1}|}$ results in decohering (removing cross terms of) the part of $U_{g(\Delta_L)}$ acting on register $\mathcal{V}_{L+1}$ by the Pauli decoherence lemma (Lemma \ref{PauliDiscretization}). Then due to the standard basis projection onto $\ket{\Delta_L}\bra{\Delta_L}$, the state will be zero unless $l = l'$. 
\begin{claim}\label{cl:blindnessdecoherence}
The following expression (from equation \ref{eq:blindnessproversstate}), which represents the prover's state at the start of round $L+1$
\begin{equation}
\frac{1}{2^m|\mbP_{|\mathcal{V}_{L+1}|}|}\sum_{\substack{k\in\{-1,1\}^m\\\Delta_L,z_1,x_1\in F_q^{|\mathcal{V}_{L+1}|}}} \tr_{\mathcal{V}_{L+1}}(W_{\Delta_L,z_1,x_1,k}' ( \sigma_{g(\Delta_L),k} )(W_{\Delta_L,z_1,x_1,k}')^\dagger)
\end{equation}
is equal to
\begin{equation}
\frac{1}{q^{3L}2^m}\sum_{\substack{k\in\{-1,1\}^m\\P\in \mbP_{|\mathcal{V}_{L+1}|}\\ \Delta_L\in F_q^{|\mathcal{V}_{L+1}|},l\in F_q^{3L}}}  \mu_{\Delta_L,l,P,k}\cdot U_{P,g(\Delta_L)} (\frac{1}{q^{|\mathcal{P}_{L+1}|}}\mcI^{\otimes |\mathcal{P}_{L+1}|}\otimes\rho_{\mathcal{E}} )U_{P,g(\Delta_L)}^\dagger 
\end{equation}
where
\begin{equation}\label{blindnessoperatorbreakdown}
U_{g(\Delta_L)} = \sum_{P\in\mbP_{|\mathcal{V}_{L+1}|}} P\otimes U_{P,g(\Delta_L)}
\end{equation}
and
\begin{equation}
\mu_{\Delta_L,l,P,k} = \tr(\ket{\Delta_{L}}\bra{\Delta_{L}}(D_k^\dagger)^{\otimes |\mathcal{V}_{L+1}|}P (\ket{l}\bra{l})^k P^\dagger (D_k)^{\otimes |\mathcal{V}_{L+1}|})
\end{equation}
\end{claim}
The above state is the same regardless of the input density matrix $\rho^k$; therefore, we have shown blindness for the polynomial scheme.
\end{proofof}
\\~\\
We proceed to the proof of the claim.
\\~\\
\begin{proofof}{\Cl{cl:blindnessdecoherence}}
We begin with
\begin{equation}\label{eq:startingblindnessstate}
\frac{1}{2^m|\mbP_{|\mathcal{V}_{L+1}|}|}\sum_{\substack{k\in\{-1,1\}^m\\\Delta_L,z_1,x_1\in F_q^{|\mathcal{V}_{L+1}|}}} \tr_{\mathcal{V}_{L+1}}(W_{\Delta_L,z_1,x_1,k}' (\sigma_{g(\Delta_L),k})(W_{\Delta_L,z_1,x_1,k}')^\dagger)
\end{equation}
where
\begin{equation}
\sigma_{g(\Delta_L),k} =  \tr_{\mathcal{P}_{L+1}}(\rho_{g(\Delta_{L})}^k) \otimes\frac{1}{q^{|\mathcal{P}_{L+1}|}}\mcI^{\otimes |\mathcal{P}_{L+1}|}\otimes\rho_{\mathcal{E}} 
\end{equation}
and
\begin{equation}
W_{\Delta_{L},z_1,x_1,k}' = (\ket{\Delta_{L}}\bra{\Delta_{L}}(D_k^\dagger)^{\otimes |\mathcal{V}_{L+1}|}(Z^{z_1}X^{x_1})^\dagger_{\mathcal{V}_{L+1}}\otimes\mcI_{\mathcal{P}_{L+1},E}) U_{g(\Delta_{L})} ((Z^{z_1}X^{x_1})_{\mathcal{V}_{L+1}}\otimes \mcI_{\mathcal{P}_{L+1}}\otimes \mcI_{\mathcal{E}})
\end{equation}
We can now apply the Pauli decoherence lemma (\Le{PauliDiscretization}) with the decomposition of $U_{g(\Delta_L)}$ as given in equation \ref{blindnessoperatorbreakdown} (and with $Z^{z_1}X^{x_1}$ playing the role of $Q$ in the lemma) to simplify the state in equation \ref{eq:startingblindnessstate} to: 
\begin{equation}\label{eq:blindnessaffterpaulitw}
\frac{1}{2^m}\sum_{\substack{k\in\{-1,1\}^m, \Delta_L\in F_q^{|\mathcal{V}_{L+1}|}\\ P\in\mbP_{|\mcV_{L+1}|}}} \tr_{\mathcal{V}_{L+1}}(W_{\Delta_L,P,k}'' ( \sigma_{g(\Delta_L),k})(W_{\Delta_L,P,k}'')^\dagger)
\end{equation}
where
\begin{equation}
W_{\Delta_{L},P,k}'' = \ket{\Delta_{L}}\bra{\Delta_{L}}(D_k^\dagger)^{\otimes |\mathcal{V}_{L+1}|}P\otimes U_{P,g(\Delta_L)} 
\end{equation}
Now we can plug in the expression for $\tr_{\mathcal{P}_{L+1}}(\rho_{g(\Delta_{L})}^k)$ from equation \ref{eq:authenticatedtracedoutstate} to obtain:
\begin{equation}
\sigma_{g(\Delta_L),k} = \frac{1}{q^{3L}}\sum_{l,l'\in F_q^{3L}}\tr( \ket{\psi}_{g(\Delta_{L}),l}\bra{\psi}_{g(\Delta_{L}),l'})^k\cdot(\ket{l}\bra{l'})^k  \otimes\frac{1}{q^{|\mathcal{P}_{L+1}|}}\mcI^{\otimes |\mathcal{P}_{L+1}|}\otimes\rho_{\mathcal{E}} 
\end{equation}
Plugging this in to equation \ref{eq:blindnessaffterpaulitw} we obtain:
\begin{equation}\label{eq:blindnessafterpluggingin}
\frac{1}{q^{3L}2^m}\sum_{\substack{k\in\{-1,1\}^m\\P\in \mbP_{|\mathcal{V}_{L+1}|}\\ \Delta_L\in F_q^{|\mathcal{V}_{L+1}|},l,l'\in F_q^{3L}}}  \mu_{\Delta_L,l,l',P,k}\tr(( \ket{\psi}_{g(\Delta_{L}),l}\bra{\psi}_{g(\Delta_{L}),l'})^k)\cdot U_{P,g(\Delta_L)} (\frac{1}{q^{|\mathcal{P}_{L+1}|}}\mcI^{\otimes |\mathcal{P}_{L+1}|}\otimes\rho_{\mathcal{E}} )U_{P,g(\Delta_L)}^\dagger 
\end{equation}
where
\begin{equation}
\mu_{\Delta_L,l,l',P,k} = \tr(\ket{\Delta_{L}}\bra{\Delta_{L}}(D_k^\dagger)^{\otimes |\mathcal{V}_{L+1}|}P (\ket{l}\bra{l'})^k P^\dagger (D_k)^{\otimes |\mathcal{V}_{L+1}|})
\end{equation}
Observe that $\mu_{\Delta_L,l,l',P,k}$ is only non zero when $l = l'$. This follows due to two observations. First, note that 
\begin{equation}
(D_k)^{\otimes |\mathcal{V}_{L+1}|}\ket{\Delta_{L}}\bra{\Delta_{L}}(D_k^\dagger)^{\otimes |\mathcal{V}_{L+1}|}
\end{equation}
is a standard basis projection, because $\ket{\Delta_L}\bra{\Delta_L}$ is a standard basis projection, and $D_k$ consists only of sum and multiplication operations (see Claim \ref{decodingoperations}). Next, note that if we have a standard basis projection $S$ acting on a matrix $P\ket{\psi}\bra{\psi'}P^\dagger$ in register $\mathcal{V}_{L+1}$, where $\ket{\psi} = \sum_i\alpha_i\ket{i}$ and $\ket{\psi'} = \sum_i\beta_i\ket{i}$, we need not consider the cross terms of $\ket{\psi}\bra{\psi'}$: 
\begin{eqnarray}
\tr(SP\ket{\psi}\bra{\psi'}P^\dagger) &=& \sum_{i,j}\alpha_i\beta_j^*\tr(SP\ket{i}\bra{j} P^\dagger)\\
&=& \sum_{i}\alpha_i\beta_i^*\tr(SP\ket{i}\bra{i} P^\dagger)
\end{eqnarray}
Since the authenticated states $(\ket{l}\bra{l'})^k$ consists only of cross terms unless $l = l'$, this implies that $l$ must equal $l'$ in order for $\mu_{\Delta_L,l,l',P,k}$ to be non zero. Therefore, we can now write equation \ref{eq:blindnessafterpluggingin} as
\begin{equation}
\frac{1}{q^{3L}2^m}\sum_{\substack{k\in\{-1,1\}^m\\P\in \mbP_{|\mathcal{V}_{L+1}|}\\ \Delta_L\in F_q^{|\mathcal{V}_{L+1}|},l\in F_q^{3L}}}  \mu_{\Delta_L,l,P,k}\tr(( \ket{\psi}_{g(\Delta_{L}),l}\bra{\psi}_{g(\Delta_{L}),l})^k)\cdot U_{P,g(\Delta_L)} (\frac{1}{q^{|\mathcal{P}_{L+1}|}}\mcI^{\otimes |\mathcal{P}_{L+1}|}\otimes\rho_{\mathcal{E}} )U_{P,g(\Delta_L)}^\dagger 
\end{equation}
Since $\tr(( \ket{\psi}_{g(\Delta_{L}),l}\bra{\psi}_{g(\Delta_{L}),l})^k) = 1$, the above expression is equal to
\begin{equation}
\frac{1}{q^{3L}2^m}\sum_{\substack{k\in\{-1,1\}^m\\P\in \mbP_{|\mathcal{V}_{L+1}|}\\ \Delta_L\in F_q^{|\mathcal{V}_{L+1}|},l\in F_q^{3L}}}  \mu_{\Delta_L,l,P,k}\cdot U_{P,g(\Delta_L)} (\frac{1}{q^{|\mathcal{P}_{L+1}|}}\mcI^{\otimes |\mathcal{P}_{L+1}|}\otimes\rho_{\mathcal{E}} )U_{P,g(\Delta_L)}^\dagger 
\end{equation}

\end{proofof}

\section{Interpretation of Results}\label{app:interpretation}
In this section, we prove Corollary \ref{corol:confidence}. 
\subsection{Clifford \QPIP}
\begin{corol} \label{corol:cliffordconfidence}
For the Clifford \QPIP\ protocol (Protocol \ref{prot:CliffordIP})  with security parameter $\epsilon$ (where $\epsilon = \frac{1}{2^e}$ by definition), if the
verifier does not abort  with probability $\ge \beta$
then the trace distance between the final density matrix conditioned on the verifier's acceptance and that of the correct state
is at most $\frac {\epsilon} \beta$
\end{corol}
\begin{proofof}{Corollary \ref{corol:cliffordconfidence}}
The final state of the protocol before the verifier's cheat detection can be
  written as (see \Eq{eq:cliffordfinalstate}):
\begin{equation}\label{eq:cliffordfinalstateint}
s\tr_A(\rho_{N+1}) + \frac{1-s}
{4^m -1}\sum_{Q\in\mbP_m\setminus\{\mcI\}}Q(\tr_A(\rho_{N+1}))
Q^\dagger
\end{equation}
where $s$ represents the weight of the prover's attack on the identity, $A$ is the space of all computational qubits other than the first, and $\rho_{N+1}$ is the correct final state of the protocol:
\begin{equation}
\rho_{N+1} = (U_{N}\cdots U_1) \rho (U_{N}\cdots U_1)^\dagger
\end{equation}
where $\rho$ is equal to the initial density matrix. Note that $\rho_{N+1}$ includes the auxiliary 0 states, but the circuit does not act on the auxiliary 0 states (so it includes $\mcI$ operators which we have not included for ease of notation). We can instead write
\begin{equation}
\rho_{N+1}' = (U_{N}\cdots U_1)\rho' (U_{N}\cdots U_1)^\dagger
\end{equation}
where $\rho'$ is the input state $\rho$ without the auxiliary 0's. Then
\begin{equation}
\tr_A(\rho_{N+1}) = \tr_{A'}(\rho'_{N+1})\otimes\ket{0}\bra{0}^{\otimes e} \EqDef \rho_C \otimes\ket{0}\bra{0}^{\otimes e}
\end{equation}
where $A'$ is the space of all compuational qubits other than the first (but excluding the auxiliary 0's). 
\\~\\
We now rewrite the state from equation \Eq{eq:cliffordfinalstateint}:
\begin{equation}
s\rho_C\otimes\ket{0}\bra{0}^{\otimes e} + \frac{1-s}
{4^m -1}\sum_{Q_1\otimes Q_2\in\mbP_m\setminus\{\mcI\}}Q_1\rho_C
Q_1^\dagger\otimes Q_2\ket{0}\bra{0}^{\otimes e} Q_2^\dagger
\end{equation}
Let $V\subset \mbP_m\setminus\{\mcI\}$ be the set of Pauli operators which pass the cheat detection procedure (i.e. preserve the auxiliary 0 states). Assume the verifier declares the computation valid with probability $\beta$. After he declares the computation valid (and we trace out the auxiliary 0 states), his state is:
\begin{equation}
\sigma = \frac{1}{\beta}(s\rho_C + \frac{1-s}{4^m-1}\sum_{Q_1\otimes Q_2\in V} Q_1\rho_C
Q_1^\dagger)
\end{equation}
Then the trace distance to the correct state $\rho_C$ is: 
\begin{eqnarray}
T(\sigma,\rho_C) &\leq& \frac{1}{\beta}(sT(\rho_C,\rho_C) + \frac{1-s}{4^m-1}\sum_{Q_1\otimes Q_2\in V} T(Q_1\rho_C
Q_1^\dagger,\rho_C))\\
&\leq& \frac{1}{\beta}\cdot\frac{|V|}{4^m-1}\\
&\leq& \frac{\epsilon}{\beta}
\end{eqnarray}
where the first inequality follows by convexity of trace distance and the final inequality follows from the same argument used to prove the security of the Cliford \QAS- more specifically, see the explanation preceding equation \ref{aftertr}.
\end{proofof}
\subsection{Polynomial \QPIP}
We now continue to the polynomial \QPIP. In this setting, we are concerned with the trace distance between density matrices \textit{after} measurement. This is because the verifier in the polynomial \QPIP\ protocol only performs a classical verification circuit; therefore, he cannot detect phase attacks on the correct state. Once the state is measured, phase attacks have no effect on the state. For this purpose, let $\sigma_M$ represent a density matrix $\sigma$ on one qudit after measurement:
\begin{equation}
\sigma_M = \sum_{i\in F_q}\ket{i}\bra{i}\sigma\ket{i}\bra{i}
\end{equation}
We will require the following fact:
\begin{fact}\label{fact:measureddensity}
For all density matrices $\rho,\sigma$ on 1 qudit,
\begin{equation}
T(\sigma_M,\rho_M)\leq T(\sigma,\rho)
\end{equation}
\end{fact}
We require a bit more notation before stating the corollary. Recall that the register $\mathcal{V}_{final}$, which is the register of containing the $3mL + m$ qudits held by the verifier at the end of the protocol, is equal to $\mathcal{V}_{L+1}\cup \mathcal{F}$. We will only be interested in the first qudit of register $\mathcal{F}$; this is the qudit which contains the final result of the circuit. For this purpose, we introduce the following notation. For a density matrix $\sigma$ on $\mathcal{V}_{final}$, let $\sigma' = \tr_{2,\ldots,m}(\tr_{\mathcal{V}_{L+1}}(\sigma))$. 
\\~\\
\begin{statement}{Corollary \ref{corol:polyconfidence}}\textit{
For the polynomial \QPIP\ protocol (Protocol \ref{prot:PolynomialIP}) with security parameter $\epsilon$ (where $\epsilon = \frac{1}{2^{m-1}}$ by definition), assume
the verifier aborts  with probability at most $1 - \beta$. Then the trace distance between the final measured density matrix conditioned on the verifier's acceptance ($\sigma_M'$) and that of the correct measured state ($\rho_C$)
is at most $\frac {2\epsilon} \beta$.}
\end{statement}
\\~\\
\begin{proofof}{Corollary \ref{corol:polyconfidence}}
Recall the final state on register $\mathcal{V}_{final}$ held by the verifier (before the verifier checks for errors, but after he decodes) from \Cl{claim:statesimplify}:
\begin{equation}\label{eq:startingsigmamaincor0}
\rho_{L+1}|_{\mathcal{V}_{final}} = \frac{1}{2^m}\sum_{k\in\{-1,1\}^m}\sigma_k = \frac{1}{2^m}\sum_{k\in\{-1,1\}^m}\sum\limits_{P\in\mbP_{|\mathcal{V}_{final}|}}\sigma_k^P
\end{equation}
where
\begin{equation}\label{eq:startingsigmaincor}
\sigma_k^P = \sum\limits_{\Delta_L\in F_q^{|\mathcal{V}_{L+1}|}} \alpha_{P,g(\Delta_L)}\cdot((\ket{\Delta_L}\bra{\Delta_L}\otimes\mcI_{\mathcal{F}})(D_k^\dagger)^{\otimes|\mathcal{V}_{final}|}P)\sigma_{g(\Delta_L)}^k((\ket{\Delta_L}\bra{\Delta_L}\otimes\mcI_{\mathcal{F}})(D_k^\dagger)^{\otimes|\mathcal{V}_{final}|}P)^\dagger
\end{equation}
and $\tr_{\mathcal{P}_{final}}(\rho_{g(\Delta_L)}^k) = \sigma_{g(\Delta_L)}^k$,
\begin{equation}\label{eq:positivecoefficients}
\alpha_{P,g(\Delta_L)} = \frac{1}{q^{|\mathcal{P}_{final}|}}\tr(U_{g(\Delta_L)}^P(\mcI_{\mathcal{P}_{final}}\otimes\rho_{\mathcal{E}})(U_{g(\Delta_L)}^P)^\dagger)
\end{equation}
and
\begin{equation}
U_{g(\Delta_L)} = \sum_{P\in\mbP_{|\mathcal{V}_{final}|}}P\otimes U_{g(\Delta_L)}^P
\end{equation}
We begin by conditioning on the verifier's acceptance, by applying the projection
\begin{equation}\label{eq:interpretationprojection}
(\hat{\Pi}_0)_{\mathcal{V}_{final}}\EqDef (\mcI^{\otimes d+1}\otimes \ket{0}\bra{0}^{\otimes d})^{\otimes 3L + 1} 
\end{equation}
on the state and then re-normalizing. The projection above represents the verifier's test of checking that the last $d$ qudits of each block of $m$ qudits are 0 (for a reminder of why this is the test and how the protocol works, see Protocol \ref{prot:PolynomialIP}). The resulting state after conditioning on acceptance (where $\beta$ is the probability of acceptance) is
\begin{eqnarray}
 \sigma &=& \frac{1}{\beta} (\hat{\Pi}_0)_{\mathcal{V}_{final}}(\rho_{L+1}|_{\mathcal{V}_{final}})(\hat{\Pi}_0)_{\mathcal{V}_{final}}\\
 &=& \frac{1}{2^m\beta}\sum_{k\in\{-1,1\}^m}\sum\limits_{P\in\mbP_{|\mathcal{V}_{final}|}} \gamma_k^P\hat{\sigma}_k^P\label{eq:stateafteraccept}
\end{eqnarray}
where
\begin{equation}\label{eq:defintsigma}
\hat{\sigma}_k^P = \frac{1}{\gamma_k^P}(\hat{\Pi}_0)_{\mathcal{V}_{final}}(\sigma_k^P)(\hat{\Pi}_0)_{\mathcal{V}_{final}}
\end{equation}
and
\begin{equation}\label{eq:defgamma}
\gamma_k^P = \tr((\hat{\Pi}_0)_{\mathcal{V}_{final}}(\sigma_k^P))
\end{equation}
Our goal is to show that the trace distance $T(\sigma_M',\rho_C)$ between $\rho_C$ and $\sigma_M'$ is at most $\frac{2\epsilon}{\beta}$. Note that the expression in equation \ref{eq:stateafteraccept} is a convex sum over density matrices; this is because $\sigma_{k}^{P}$ is an unnormalized density matrix. To see this, observe that $\sigma_k^P$ (written in equation \ref{eq:startingsigmaincor}) is a sum over terms of the following form: there is a density matrix ($\sigma_{g(\Delta_L)}^k$), followed by a unitary operation $(D_k^\dagger)^{\otimes |\mathcal{V}_{final}|}P$, followed by a projection. Each term has a non negative coefficient ($\alpha_{P,g(\Delta_L)}$ from equation \ref{eq:positivecoefficients}). Therefore, by convexity of trace distance, we can upper bound the trace distance as follows:
\begin{eqnarray}\label{eq:tracedistancebound}
  T(\sigma_M',\rho_C)&\leq& \frac{1}{2^m\beta}\sum_{k\in\{-1,1\}^m}
  \sum\limits_{P\in\mbP_{|\mathcal{V}_{final}|}}\gamma_k^P T((\hat{\sigma}_{k}^P)_M',\rho_C)\\
  &=&  \frac{1}{\beta}\sum\limits_{P\in\mbP_{|\mathcal{V}_{final}|}}\frac{1}{2^m} \sum_{k\in\{-1,1\}^m}\gamma_{k}^PT((\hat{\sigma}_{k}^P)_M',\rho_C)
\end{eqnarray}
Next, we will require two claims (which we prove immediately after the current proof). The first claim shows that when $P$ is a trivial Pauli operator, it preserves the correct final state on the first qudit of register $\mathcal{F}$:
\begin{claim}\label{cl:trivialdensitymatrix}
For all trivial $P\in\mbP_{|\mathcal{V}_{final}|}$ (i.e. $P$ consisting of only $Z$ and $\mcI$ operators):
\begin{equation}
 \frac{1}{2^m}\sum_{k\in\{-1,1\}^m}\gamma_k^PT((\hat{\sigma}_{k}^P)_M',\rho_C) = 0
\end{equation}
\end{claim}
The next claim shows that if $P$ is a non trivial Pauli operator, the trace distance will still be small; intuitively, this is because the state with attack operator $P$ can only pass the verifier's test for 2 (out of $2^m$) sign keys:
\begin{claim}\label{cl:nontrivialdensitymatrix}
 For all non trivial $P\in\mbP_{|\mathcal{V}_{final}|}$:
 \begin{equation}
  \frac{1}{2^m}\sum_{k\in\{-1,1\}^m}\gamma_{k}^P T((\hat{\sigma}_{k}^P)_M',\rho_C) \leq \frac{1}{2^{m-1}q^{3L}}\sum_{a\in F_q^{3L}}\alpha_{P,a}
\end{equation}
\end{claim}
Given the two claims, we can further simplify the bound in equation \ref{eq:tracedistancebound} by first using \Cl{cl:trivialdensitymatrix} to remove trivial Pauli operators from the expression (let $P\in\mbP_{|\mathcal{V}_{final}|}^T$ be the set of all non trivial Pauli operators):
\begin{eqnarray}
    T(\sigma_M',\rho_C) &\leq& \frac{1}{\beta}\sum\limits_{P\in\mbP_{|\mathcal{V}_{final}|}^{NT}}\frac{1}{2^m} \sum_{k\in\{-1,1\}^m}\gamma_{k}^PT((\hat{\sigma}_{k}^P)'_M,\rho_C)\\
    &\leq&  \frac{1}{2^{m-1}q^{3L}\beta}\sum\limits_{P\in\mbP_{|\mathcal{V}_{final}|}^{NT}}\sum_{a\in F_q^{3L}}\alpha_{P,a} \\
    &=& \frac{1}{2^{m-1}q^{3L}\beta}\sum_{a\in F_q^{3L}}\sum\limits_{P\in\mbP_{|\mathcal{V}_{final}|}^{NT}}\alpha_{P,a} \\
  \end{eqnarray}
The second inequality follows from \Cl{cl:nontrivialdensitymatrix}. Next, we can use  \Le{Decompose}, which provides the following equality:
\begin{equation}
\sum_{P\in \mbP_{|\mathcal{V}_{final}}|} \alpha_{P,a} = 1
\end{equation}
Continuing with the upper bound, we obtain:
\begin{eqnarray}
    T(\sigma_M',\rho_C) &\leq& \frac{1}{2^{m-1} q^{3L}\beta}\sum_{a\in F_q^{3L}} 1\\
    &=& \frac{1}{2^{m-1}\beta}
    \end{eqnarray}
  which completes the proof.
  \end{proofof}
 
 \subsubsection{Proof of \Cl{cl:trivialdensitymatrix}}
  \begin{proofof}{\Cl{cl:trivialdensitymatrix}}
We would like to show that for all trivial $P\in\mbP_{|\mathcal{V}_{final}|}$ (i.e. $P$ consisting of only $Z$ and $\mcI$ operators):
\begin{equation}\label{eq:trivialdensitymatrix}
 \frac{1}{2^m}\sum_{k\in\{-1,1\}^m}\gamma_k^PT((\hat{\sigma}_{k}^P)_M',\rho_C) = 0
\end{equation}
To do so, we need to show that for all $k$, 
\begin{equation}\label{eq:needtoprovetrivial}
\tr(\ket{1}\bra{1}(\hat{\sigma}_{k}^P)') = \tr(\ket{1}\bra{1}\rho_C)\EqDef p_1
\end{equation}
where $p_1$ is the probability that the correct state outputs 1 when measured. The reason the above statement is equivalent to equation \ref{eq:trivialdensitymatrix} is because we are considering the measured density matrices; therefore, we need only prove that both density matrices obtain 1 with the same probability to prove that the trace distance is 0.
\\~\\
Recall from equation \ref{eq:defintsigma} that:
\begin{equation}
\hat{\sigma}_k^P = \frac{1}{\gamma_k^P}(\hat{\Pi}_0)_{\mathcal{V}_{final}}(\sigma_k^P)(\hat{\Pi}_0)_{\mathcal{V}_{final}}
\end{equation}
Plugging this in to equation \ref{eq:needtoprovetrivial}, we obtain:
\begin{eqnarray}
\tr(\ket{1}\bra{1}(\hat{\sigma}_{k}^P)') &=& \frac{1}{\gamma_k^P}\tr((\mcI^{\otimes |\mathcal{V}_{final}|}\otimes\ket{1}\bra{1}\otimes\mcI^{\otimes m-1}) (\hat{\Pi}_0)_{\mathcal{V}_{final}}(\sigma_k^P)(\hat{\Pi}_0)_{\mathcal{V}_{final}})\\
&=&  \frac{1}{\gamma_k^P}\tr((\mcI^{\otimes |\mathcal{V}_{final}|}\otimes\ket{1}\bra{1}\otimes\mcI^{\otimes m-1}) (\hat{\Pi}_0)_{\mathcal{V}_{final}}(\sigma_k^P))\\
&=& \frac{1}{\gamma_k^P}\tr(\hat{\Pi}_0\sigma_k^P)\label{eq:trivialdensitymatrix2}
\end{eqnarray}
where the second equality follows because $(\mcI^{\otimes |\mathcal{V}_{final}|}\otimes\ket{1}\bra{1}\otimes\mcI^{\otimes m-1})$ commutes with $(\hat{\Pi}_0)_{\mathcal{V}_{final}}$ and due to the cyclic nature of trace. The third equality follows due to the following equality: 
\begin{equation}
\hat{\Pi}_0 = (\mcI^{\otimes |\mathcal{V}_{final}|}\otimes\ket{1}\bra{1}\otimes\mcI^{\otimes m-1}) (\hat{\Pi}_0)_{\mathcal{V}_{final}}
\end{equation}
where we recall (from equation \ref{acceptingprojection}) that
\begin{equation}
\hat{\Pi}_0 = (\mcI^{\otimes d+1}\otimes\ket{0}\bra{0}^{\otimes d})^{\otimes 3L}_{\mathcal{V}_{L+1}}\otimes (\ket{1}\bra{1}\otimes\mcI^{\otimes d}\otimes \ket{0}\bra{0}^{\otimes d})_{\mathcal{F}}\EqDef (\hat{\Pi}_0)_{\mathcal{V}_{L+1}}\otimes(\hat{\Pi}_0)_{\mathcal{F}}
\end{equation}
and (from equation \ref{eq:interpretationprojection}) that 
\begin{equation}
(\hat{\Pi}_0)_{\mathcal{V}_{final}}= (\mcI^{\otimes d+1}\otimes \ket{0}\bra{0}^{\otimes d})^{\otimes 3L + 1} 
\end{equation}
Next, Claim \ref{claim:trivialsoundness} gives the following inequality:
\begin{equation}
\frac{1}{\gamma_k^P}\tr(\hat{\Pi}_0\sigma_k^P)= \frac{1}{\gamma_k^P}\frac{p_1}{q^{3L}}\sum_{a\in F_q^{3L}}\alpha_{P,a}
\end{equation}
Note that Claim \ref{claim:trivialsoundness} is not stated in this format. The first difference is that instead of $p_1$, the claim has $\gamma$; this is simply the probability that the circuit applied by the honest prover outputs 1, which is $p_1$ in this case. The second is that the statement of the claim also has a summation over $k$. However, in the proof of the claim, that summation is not used at all; the claim is shown for individual elements of the sum over $k$, and the sum over $k$ is only used in the statement. Finally, Claim \ref{claim:trivialsoundness} has an inequality rather than an equality; see Remark \ref{remark:fortrivialcorol} for why it is okay to use equality in this setting. 
\\~\\
Finally, we claim that $\gamma_k^P = \frac{1}{q^{3L}}\sum\limits_{a\in F_q^{3L}}\alpha_{P,a}$. Given this claim, we have
\begin{eqnarray}
\tr(\ket{1}\bra{1}(\hat{\sigma}_{k}^P)')&=& \frac{1}{\gamma_k^P}\tr(\hat{\Pi}_0\sigma_k^P)\\
&=& \frac{1}{\gamma_k^P}\frac{p_1}{q^{3L}}\sum_{a\in F_q^{3L}}\alpha_{P,a}\\
&=& p_1
\end{eqnarray}
which completes the proof, since we have proven the equality in equation \ref{eq:needtoprovetrivial}.To see why $\gamma_k^P$ satisfies the above equality, recall from equation \ref{eq:defgamma} that
\begin{equation}
\gamma_k^P = \tr((\hat{\Pi}_0)_{\mathcal{V}_{final}}(\sigma_k^P))
\end{equation}
We can again apply Claim \ref{claim:trivialsoundness} here; the two differences are that we are using $(\hat{\Pi}_0)_{\mathcal{V}_{final}}$ rather than $\hat{\Pi}_0$ and that we require equality rather than inequality. The difference between these two projections is that the latter projects the first qudit of $\mathcal{F}$ onto $\ket{1}\bra{1}$. See Remark \ref{remark:fortrivialcorol} for why the claim still applies, but with $\gamma$ in the statement of the claim replaced by 1 and the inequality replaced by equality, which gives 
\begin{equation}
\gamma_k^P = \tr((\hat{\Pi}_0)_{\mathcal{V}_{final}}(\sigma_k^P)) = \frac{1}{q^{3L}}\sum_{a\in F_q^{3L}}\alpha_{P,a}
\end{equation}

  \end{proofof}

  \subsubsection{Proof of \Cl{cl:nontrivialdensitymatrix}}
  \begin{proofof}{\Cl{cl:nontrivialdensitymatrix}}
   We would like to show that for all non trivial $P\in\mbP_{|\mathcal{V}_{final}|}$:
 \begin{equation}
  \frac{1}{2^m}\sum_{k\in\{-1,1\}^m}\gamma_{k}^P T((\hat{\sigma}_{k}^P)_M',\rho_C) \leq \frac{1}{2^{m-1}q^{3L}}\sum_{a\in F_q^{3L}}\alpha_{P,a}
\end{equation}
First, we use Fact \ref{fact:measureddensity} and then upper bound the trace distance by 1 to obtain:
\begin{eqnarray}
  \frac{1}{2^m}\sum_{k\in\{-1,1\}^m}\gamma_{k}^P T((\hat{\sigma}_{k}^P)_M',\rho_C) &\leq&
  \frac{1}{2^m}\sum_{k\in\{-1,1\}^m}\gamma_{k}^P T((\hat{\sigma}_{k}^P)',\rho_C) \\
  &\leq & \frac{1}{2^m}\sum_{k\in\{-1,1\}^m}\gamma_{k}^P 
  \end{eqnarray}
  Next, note that by the definition of $\gamma_k^P$ in equation \ref{eq:defgamma}
\begin{eqnarray}
  \frac{1}{2^m}\sum_{k\in\{-1,1\}^m}\gamma_{k}^P  = \frac{1}{2^m}\sum_{k\in\{-1,1\}^m}\tr((\hat{\Pi}_0)_{\mathcal{V}_{final}}\sigma_k^P)\leq \frac{1}{2^{m-1}q^{3L}}\sum_{a\in F_q^{3L}}\alpha_{P,a}
\end{eqnarray}
where the final inequality follows from Claim \ref{claim:nontrivialsoundness}. Note that in Claim \ref{claim:nontrivialsoundness} the projection is $\hat{\Pi}_0$ rather than $(\hat{\Pi}_0)_{\mathcal{V}_{final}}$. Recall (from equation \ref{acceptingprojection}) that
\begin{equation}
\hat{\Pi}_0 = (\mcI^{\otimes d+1}\otimes\ket{0}\bra{0}^{\otimes d})^{\otimes 3L}_{\mathcal{V}_{L+1}}\otimes (\ket{1}\bra{1}\otimes\mcI^{\otimes d}\otimes \ket{0}\bra{0}^{\otimes d})_{\mathcal{F}}\EqDef (\hat{\Pi}_0)_{\mathcal{V}_{L+1}}\otimes(\hat{\Pi}_0)_{\mathcal{F}}
\end{equation}
and (from equation \ref{eq:interpretationprojection})
\begin{equation}
(\hat{\Pi}_0)_{\mathcal{V}_{final}}= (\mcI^{\otimes d+1}\otimes \ket{0}\bra{0}^{\otimes d})^{\otimes 3L + 1} 
\end{equation}
The difference between the two is that the first qudit of register $\mathcal{F}$ is projected onto $\ket{1}\bra{1}$ in $\hat{\Pi}_0$ and onto $\mcI$ in $(\hat{\Pi}_0)_{\mathcal{V}_{final}}$. However, the proof of Claim \ref{claim:nontrivialsoundness} holds if $\hat{\Pi}_0$ is replaced with $(\hat{\Pi}_0)_{\mathcal{V}_{final}}$; see Remark \ref{remark:claimcorol}. 
  \end{proofof}

\section{Acknowledgements}
D.A. thanks Oded Goldreich, Madhu sudan and Guy Rothblum
for exciting and inspiring conversations that eventually led to this work.
E.E. thanks Avinatan Hassidim for
stimulating and refining ideas, particularly about fault tolerance.
We also thank Gil Kalai, David DiVincenzo and Ari Mizel, for stimulating
questions and clarifications, and Daniel Gottesman for many helpful ideas and
remarks, and in particular, for his help in
proving \Th{thm:CliffordAuth}.

\bibliographystyle{alpha}
\bibliography{qpip}

\newcommand{\etalchar}[1]{$^{#1}$}
\begin{thebibliography}{BOCG{\etalchar{+}}06}

\bibitem[AA06]{jonesHardness}
D.~Aharonov and I.~Arad.
\newblock {The BQP-hardness of approximating the Jones Polynomial}.
\newblock {\em Arxiv preprint quant-ph/0605181}, 2006.

\bibitem[Aar09]{aaronson2009}
S.~Aaronson.
\newblock {BQP and the Polynomial Hierarchy}.
\newblock {\em Arxiv preprint quant-ph/0910.4698}, 2009.

\bibitem[AAV13]{qpcp}
D.~Aharonov, I.~Arad, and T.~Vidick.
\newblock {The quantum PCP conjecture}.
\newblock {\em ACM SIGACT News}, 44:47--49, 2013.

\bibitem[AB09]{arora:ccm}
S.~Arora and B.~Barak.
\newblock {\em {Computational Complexity: A Modern Approach}}.
\newblock Cambridge Univ. Press, 2009.

\bibitem[ABO97]{aharonov1997ftq}
D.~Aharonov and M.~Ben-Or.
\newblock {Fault-tolerant quantum computation with constant error}.
\newblock {\em Proceedings of the twenty-ninth annual ACM symposium on Theory
  of computing}, pages 176--188, 1997.

\bibitem[ABOE08]{abe2008}
D.~Aharonov, M.~Ben-Or, and E.~Eban.
\newblock {Interactive Proofs For Quantum Computations}.
\newblock {\em Arxiv preprint arXiv:0810.5375}, 2008.

\bibitem[ABW08]{ambainis2008tre}
A.~Ambainis, J.~Bouda, and A.~Winter.
\newblock {Tamper-resistant encryption of quantum information}.
\newblock {\em Arxiv preprint arXiv:0808.0353}, 2008.

\bibitem[AE07]{ambainis2007qtd}
A.~Ambainis and J.~Emerson.
\newblock {Quantum t-designs: t-wise independence in the quantum world}.
\newblock {\em Arxiv preprint quant-ph/0701126}, 2007.

\bibitem[AFK87]{abadi1987hio}
M.~Abadi, J.~Feigenbaum, and J.~Kilian.
\newblock {On hiding information from an oracle}.
\newblock In {\em Proceedings of the nineteenth annual ACM conference on Theory
  of computing}, pages 195--203. ACM New York, NY, USA, 1987.

\bibitem[AJL06]{aharonov2006pqa}
D.~Aharonov, V.~Jones, and Z.~Landau.
\newblock {A polynomial quantum algorithm for approximating the Jones
  polynomial}.
\newblock In {\em Proceedings of the thirty-eighth annual ACM symposium on
  Theory of computing}, pages 427--436. ACM New York, NY, USA, 2006.

\bibitem[AS06]{blind}
P.~Arrighi and L.~Salvail.
\newblock {Blind Quantum Computation}.
\newblock {\em International Journal of Quantum Information}, 4(5):883--898,
  2006.

\bibitem[ASMZ17]{verifiableftqpip}
D.~Aharonov, F.~Song, U.~Mahadev, and J.~Zhengfeng.
\newblock {Blind or verifiable fault tolerant delegated quantum computation}.
\newblock {\em In progress}, 2017.

\bibitem[AV12]{av2012}
D.~Aharonov and U.~Vazirani.
\newblock {Is Quantum Mechanics Falsifiable? A computational perspective on the
  foundations of Quantum Mechanics}.
\newblock {\em Arxiv preprint arXiv:1206.3686}, 2012.

\bibitem[BCG{\etalchar{+}}02]{barnum2002aqm}
H.~Barnum, C.~Cr{\'e}peau, D.~Gottesman, A.~Smith, and A.~Tapp.
\newblock {Authentication of Quantum Messages}.
\newblock {\em Proceedings of the 43rd Symposium on Foundations of Computer
  Science}, pages 449--458, 2002.

\bibitem[BFK08]{broadbent2008ubq}
A.~Broadbent, J.~Fitzsimons, and E.~Kashefi.
\newblock {Universal blind quantum computation}.
\newblock {\em Arxiv preprint arXiv:0807.4154}, 2008.

\bibitem[BFKW13]{bfkw2013}
S.~Barz, J.F. Fitzsimons, E.~Kashefi, and P.~Walther.
\newblock {Demonstration of measurement-only blind quantum computing}.
\newblock {\em Nature Physics 9, 727}, 2013.

\bibitem[BFLW09]{bfld2009}
M.~Bordewich, M.~Freedman, L.~Lov\'{a}sz, and D.~Welsh.
\newblock {Approximate Counting and Quantum Computation}.
\newblock {\em Arxiv preprint 0908.2122}, 2009.

\bibitem[BGS12]{broadbent2012}
A.~Broadbent, G.~Gutoski, and D.~Stebila.
\newblock {Quantum one-time programs}.
\newblock {\em Arxiv preprint arXiv:1211.1080}, 2012.

\bibitem[BJ14]{bj2014}
A.~Broadbent and S.~Jeffery.
\newblock {Quantum homomorphic encryption for circuits of low T-gate
  complexity}.
\newblock {\em Arxiv preprint arXiv:1412.8766}, 2014.

\bibitem[BK05]{bravyi2005uqc}
S.~Bravyi and A.~Kitaev.
\newblock {Universal quantum computation with ideal Clifford gates and noisy
  ancillas}.
\newblock {\em Physical Review A}, 71(2):22316, 2005.

\bibitem[BKB{\etalchar{+}}12]{bkbfzw2012}
S.~Barz, E.~Kashefi, A.~Broadbent, J.~F. Fitzsimons, A.~Zeilinger, and
  P.~Walther.
\newblock {Demonstration of blind quantum computing}.
\newblock {\em Science 335, 303}, 2012.

\bibitem[BOCG{\etalchar{+}}06]{benor2006smq}
M.~Ben-Or, C.~Cr\'epeau, D.~Gottesman, A.~Hassidim, and A.~Smith.
\newblock {Secure Multiparty Quantum Computation with (Only) a Strict Honest
  Majority}.
\newblock {\em Foundations of Computer Science, 2006. FOCS'05. 47th Annual IEEE
  Symposium on}, pages 249--260, 2006.

\bibitem[Bro15]{broadbent2016}
A.~Broadbent.
\newblock {How to Verify a Quantum Computation}.
\newblock {\em Arxiv preprint arXiv:1509.09180}, 2015.

\bibitem[Chi01]{childs2001saq}
A.M. Childs.
\newblock {Secure assisted quantum computation}.
\newblock {\em Arxiv preprint quant-ph/0111046}, 2001.

\bibitem[DLT02]{dlt2002}
D.~DiVincenzo, D.W. Leung, and B.M. Terhal.
\newblock {Quantum data hiding}.
\newblock {\em IEEE Trans. Inf. Th., 48(3):580–599}, 2002.

\bibitem[DSS16]{dss2016}
Y.~Dulek, C.~Schaffner, and F.~Speelman.
\newblock {Quantum homomorphic encryption for polynomial-sized circuits}.
\newblock {\em Arxiv preprint arXiv:1603.09717}, 2016.

\bibitem[FK12]{fk2012}
J.~Fitzsimons and E.~Kashefi.
\newblock {Unconditionally verifiable blind computation}.
\newblock {\em Arxiv preprint arXiv:1203.5217}, 2012.

\bibitem[FKLW01]{fklw2001}
M.~Freedman, A.~Kitaev, M.~Larsen, and Z.~Wang.
\newblock {Topological Quantum Computation}.
\newblock {\em Arxiv preprint quant-ph/0101025}, 2001.

\bibitem[GKW15]{gkw2015}
A.~Gheorghiu, E.~Kashefi, and P.~Wallden.
\newblock {Robustness and device independence of verifiable blind quantum
  computing}.
\newblock {\em New Journal of Physics 17, 083040}, 2015.

\bibitem[GMR85]{goldwasser1985kci}
S.~Goldwasser, S.~Micali, and C.~Rackoff.
\newblock {The knowledge complexity of interactive proof-systems}.
\newblock In {\em Proceedings of the seventeenth annual ACM symposium on Theory
  of computing}, pages 291--304. ACM New York, NY, USA, 1985.

\bibitem[Got04]{aaronsonpost}
Daniel Gottesman, 2004.
\newblock As referenced in [\url{http://www.scottaaronson.com/blog/?p=284};
  accessed 13-Apr-2017].

\bibitem[GRB{\etalchar{+}}16]{grbmw2016}
C.~Greganti, MC. Roehsner, s.~Barz, T.~Morimae, and P.~Walther.
\newblock {Demonstration of measurement-only blind quantum computing}.
\newblock {\em New Journal of Physics 18, 013020}, 2016.

\bibitem[HH16]{hh2016}
M.~Hayashi and M.~Hajdu\v{s}ek.
\newblock {Self-guaranteed measurement-based quantum computation}.
\newblock {\em Arxiv preprint arXiv:1603.02195}, 2016.

\bibitem[HM15]{hm2015}
M.~Hayashi and T.~Morimae.
\newblock {Verifiable Measurement-Only Blind Quantum Computating with
  Stabilizer Testing}.
\newblock {\em Physical Review Letters 115}, 2015.

\bibitem[HPDF15]{hpf2015}
M.~Hajdu\v{s}ek, C.~P\'{e}rez-Delgado, and J.~Fitzsimons.
\newblock {Device-Independent Verifiable Blind Quantum Computation}.
\newblock {\em Arxiv preprint arXiv:1502.02563}, 2015.

\bibitem[KSV02]{kitaev2002caq}
A.Y. Kitaev, A.~Shen, and M.N. Vyalyi.
\newblock {\em {Classical and Quantum Computation}}.
\newblock American Mathematical Society, 2002.

\bibitem[Mck13]{mckague2013}
M.~Mckague.
\newblock {Interactive proofs for BQP via self-tested graph states}.
\newblock {\em Arxiv preprint arXiv:1309.5675}, 2013.

\bibitem[MF16]{mf2016}
T.~Morimae and J.~Fitzsimons.
\newblock {Post hoc verification with a single prover}.
\newblock {\em Arxiv preprint arXiv:1603.06046}, 2016.

\bibitem[Mor14]{morimae2014}
T.~Morimae.
\newblock {Verification for measurement-only blind computation}.
\newblock {\em Physical Review A 89}, 2014.

\bibitem[Roo03]{roodman2003bap}
A.~Roodman.
\newblock {Blind Analysis in Particle Physics}.
\newblock In {\em Statistical Problems in Particle Physics, Astrophysics, and
  Cosmology, Proceedings of the PHYSTAT 2003 Conference held 8-11 September,
  2003 at the Stanford Linear Accelerator Center. SLAC eConf C030908.
  \url{http://www.slac.stanford.edu/econf/C030908}., p. 166}, 2003.

\bibitem[RUV12]{ruv2012}
B.~Reichardt, F.~Unger, and U.~Vazirani.
\newblock {A classical leash for a quantum system}.
\newblock {\em Arxiv preprint arXiv:1209.0448}, 2012.

\bibitem[Sho96]{shor1996}
P.~Shor.
\newblock {Fault-tolerant quantum computation}.
\newblock {\em Arxiv preprint quant-ph/9605011}, 1996.

\bibitem[Sho97]{shor1997pta}
PW~Shor.
\newblock {Polynomial-time algorithms for prime factorization and discrete
  logarithms on a quantum computer}.
\newblock {\em SIAM journal on computing(Print)}, 26(5):1484--1509, 1997.

\bibitem[TFMI16]{tfmi2016}
Y.~Takeuchi, K.~Fujii, T.~Morimae, and N.~Imoto.
\newblock {Practically verifiable blind quantum computation with acceptance
  rate amplification}.
\newblock {\em Arxiv preprint arXiv:1607.01568}, 2016.

\bibitem[Vaz07]{vaziranitalk}
Umesh Vazirani, 2007.
\newblock Talk given in a conference in Japan.

\bibitem[Wat03]{watrous2003phc}
J.~Watrous.
\newblock {PSPACE has constant-round quantum interactive proof systems}.
\newblock {\em Theoretical Computer Science}, 292(3):575--588, 2003.

\bibitem[Wik08]{BlindWiki}
Wikipedia.
\newblock Blind experiment --- {W}ikipedia{,} the free encyclopedia, 2008.
\newblock [\url{https://en.wikipedia.org/wiki/Blinded_experiment}; accessed
  20-Oct-2008].

\bibitem[Yaa08]{JonathanThesis}
Jonathan Yaari.
\newblock {\em \emph{Preprint:} Interactive Proofs as a Theory of
  Confirmation}.
\newblock PhD thesis, The Hebrew University of Jerusalem, 2008.

\bibitem[YPDF14]{ypf2014}
L.~Yu, C.~Perez-Delgado, and J.~Fitzsimons.
\newblock {Limitations on information theoretically secure quantum homomorphic
  encryption}.
\newblock {\em Arxiv preprint arXiv:1406.2456}, 2014.

\end{thebibliography}
\newpage \appendix

\section{A Symmetric Definition of \QPIP}\label{sec:symmetric}  
Here we provide the definition of $\QPIP_\kappa^{sym}$ and then prove Corollary \ref{thm:mainsym} and Corollary \ref{corol:bqpsepbpp}. We begin with the definition of $\QPIP_{\kappa}^{sym}$: 
\begin{deff}\label{def:QPIPsym} A language $\mcL$ is in the class symmetric quantum prover
  interactive proof $(\QPIP_{\kappa}^{sym})$ with completeness $c$ 
and soundness $s$ (where $c-s$ is constant) if there exists an interactive protocol with
  the following properties:
\begin{itemize}
\item The prover $\mathds{P}$ and verifier $\mathds{V}$ are exactly the same as in the
      definition of $\QPIP_{\kappa}$ (\Def{def:QPIP}). Namely, a \BQP\ machine and
      quantum-classical hybrid machine for the prover and verifier respectively.
\item Communication is identical to the $\QPIP_{\kappa}$ definition.
\item The verifier has three possible outcomes: \textbf{1}, \textbf{0}, and \textbf{ABORT}:
      \begin{itemize}
  \item\textbf{1}: The verifier is convinced that $x\in \mcL$.
  \item\textbf{0}: The verifier is convinced that $x\notin \mcL$.
  \item\textbf{ABORT}: The verifier caught the prover cheating.
      \end{itemize}

\item \textbf{Completeness}: $\forall
      x\in\{0,1\}^*$, after interacting with $\mathds{P}$, the verifier's outcome is correct with high probability:
\[ \Pr_r({\left[\mathds{V},\mathds{P} \right](x,r) = \mathbbm{1}_\mcL })  \ge c \]
where $\mathbbm{1}_\mcL$ is the indicator function of $\mcL$, $r$ represents the randomness used by the verifier, and $\left[\mathds{V},\mathds{P} \right](x,r)$ is the verifier's outcome after using randomness $r$ and interacting with $\mathds{P}$ on input $x$.
\item \textbf{Soundness}: For all provers $\mathds{P}'$ (with the same description as $\mathds{P}$) and for {\bf all}
      $x\in\{0,1\}^*$, the verifier is mistaken with bounded
      probability, that is:
\[ \Pr_r({\left[\mathds{V},\mathds{P} \right](x,r) = 1- \mathbbm{1}_\mcL })  \le s\]
\end{itemize}
\end{deff}
We now prove Corollary \ref{thm:mainsym}: 

\begin{proofof}{Corollary \ref{thm:mainsym}}
We will prove that $\QPIP_c = \QPIP_c^{sym}$. It follows from the definitions of $\QPIP_c$ (Definition \ref{def:QPIP}) and $\QPIP_c^{sym}$ (Definition \ref{def:QPIPsym}) that $\QPIP_c\subseteq \QPIP_c^{sym}$. We obtain the other direction by showing that for any language $\mcL$, if $\mcL$ is in $\QPIP_c$ then $\mcL,\mcL^c \in \QPIP_c^{sym}$. First note that if $\mcL$ is in $\QPIP_c$, then so is $\mcL^c$, since \BQP\ is closed under complement and $\BQP\ = \QPIP_c$ by \Th{thm:main}. Let $\mathds{V}_{\mcL},\mathds{P}_{\mcL}$ denote the $\QPIP_c$ verifier and prover for the
  language $\mcL$. By the assumption, there exists such a pair for both $\mcL$
  and $\mcL^c$. We define the pair $\wt{\mathds{P}}$ and $\wt{\mathds{V}}$ to be
  $\QPIP_c^{sym}$ verifier and prover in the following way: on the first round the
  prover $\wt{\mathds{P}}$ sends to $\wt{\mathds{V}}$\; ``yes'' if $x\in \mcL$ and ``no''
  otherwise. Now, both $\wt{\mathds{P}}$ and $\wt{\mathds{V}}$ behave according to
  $\mathds{V}_{\mcL},\mathds{P}_{\mcL}$ if ``yes'' was sent or according to
  $\mathds{V}_{\mcL^c},\mathds{P}_{\mcL^c}$ otherwise. Soundness and completeness follow
  immediately from the definition.

\end{proofof}
\\~\\
Finally, we prove Corollary \ref{corol:bqpsepbpp}: 

\begin{proofof}{Corollary \ref{corol:bqpsepbpp}}
This corollary uses $\QPIP^{sym}$ rather than \QPIP\ (recall from Corollary \ref{thm:mainsym} that $\QPIP^{sym} = \BQP$). First note that if we run $\QPIP^{sym}$ (either the polynomial based or Clifford based protocol) on an instance $x$ drawn from $D$ and the verifier does not abort with probability $\beta$, the probability that the verifier outputs the incorrect answer is at most $\frac{2\epsilon}{\beta} + \gamma$, by Corollary \ref{corol:confidence} and Corollary \ref{corol:polyconfidence}. We need to amplify this probability so that the verifier outputs the incorrect answer with probability which is at most inverse exponential in $n$. If we can do so, the corollary follows since any \BPP\ machine would err with non-negligible probability on instances drawn from $D$, by assumption. Therefore, the prover cannot be simulated by a \BPP\ machine. 
\\~\\
To amplify the probability of outputting an incorrect answer, we run $\QPIP^{sym}$ polynomially many times (in $n$) and take the output to be the majority of the output values, ignoring runs on which the verifier aborted. Since $\beta$ is constant, if we repeat the protocol polynomially many times, we expect to collect polynomially many output values (we fail to do so with probability $p$ which is inverse exponential $n$). Since each output is correct with probability $1 - (\frac{2\epsilon}{\beta} + \gamma) > \frac{1}{2}$, by taking the majority of these output values, we can reduce the error of the output to be $p'$, which is inverse exponential in $n$. The overall probability of error is then $(1-p)p' + p$, which is inverse exponential in $n$. 
\end{proofof}

\section{Clifford and Pauli Operators}\label{app:backgroundcliffordpauli}
Here are some useful lemmas about Clifford/Pauli operators. We first prove the Pauli mixing lemma:

\begin{proofof}{\Le{paulimix} \textbf{ (Pauli Mixing)}}
First, we write $\rho$ as:
$$
\sum_{ij}\ket{i}\bra{j}_A\otimes \rho_{ij}
$$
It follows that:
$$
\tr_A(\rho) = \sum_i\rho_{ii}
$$
Next, observe that:
\begin{eqnarray}
\sum_{P\in\mbP_n} P\ket{i}\bra{j}P^\dagger &=& \sum_{zx}Z^zX^x \ket{i}\bra{j} (Z^zX^x)^\dagger\\
&=& \sum_{zx}\omega_q^{z(i-j)} X^x\ket{i}\bra{j}(X^x)^\dagger\
\end{eqnarray}
This expression is 0 if $i\neq j$. If $i = j$, we obtain $q^n\mcI$. Plugging in this observation to the above expression, we have:
\begin{eqnarray}
\frac{1}{|\mbP_n|}\sum_{P\in\mbP_n} (P\otimes\mcI_B)\rho(P\otimes\mcI_B)^\dagger &=&
\frac{1}{|\mbP_n|}\sum_{ij}\sum_{P\in\mbP_n} P\ket{i}\bra{j}_A P^\dagger\otimes \rho_{ij}\\
&=& \frac{1}{|\mbP_n|}\sum_{i}\sum_{P\in\mbP_n} P\ket{i}\bra{i}_A P^\dagger\otimes \rho_{ij}\\
&=& \frac{q^n}{|\mbP_n|}\mcI\otimes\sum_i \rho_{ii}\\
&=&\frac{1}{q^n}\mcI\otimes\tr_A(\rho)
\end{eqnarray}
\end{proofof}

Now we prove the Clifford mixing lemma:

\begin{proofof}{\Le{cliffordmix} \textbf{ (Clifford Mixing)}}
To prove this lemma, we observe that applying a random Clifford includes applying a random Pauli, and the lemma then follows from \Le{paulimix}. In more detail, we have the following equality for all $Q\in\mbP_n$:
\begin{eqnarray}
\frac{1}{|\mfC_n|}\sum_{C\in\mfC_n} (C\otimes\mcI_B) \rho (C\otimes\mcI_B)^\dagger &=& \frac{1}{|\mfC_n|}\sum_{C\in\mfC_n} (CQ\otimes\mcI_B) \rho (CQ\otimes\mcI_B)^\dagger 
\end{eqnarray}
Now we have:
\begin{eqnarray}
\frac{1}{|\mfC_n|}\sum_{C\in\mfC_n} (C\otimes\mcI_B) \rho (C\otimes\mcI_B)^\dagger &=& \frac{1}{|\mfC_n||\mbP_n|}\sum_{C\in\mfC_n} \sum_{Q\in\mbP_n} (CQ\otimes\mcI_B) \rho (CQ\otimes\mcI_B)^\dagger
\end{eqnarray}
Regrouping terms, we have
\begin{eqnarray}
\ldots &=& \frac{1}{|\mfC_n|}\sum_{C\in\mfC_n} (C\otimes\mcI_B)(\frac{1}{|\mbP_n|}\sum_{Q\in\mbP_n} (Q\otimes\mcI_B) \rho (Q\otimes\mcI_B)^\dagger)(C\otimes\mcI_B)^\dagger
\end{eqnarray}
By \Le{paulimix} (with $q = 2$) the above expression is equal to:
\begin{eqnarray}
\ldots &=& \frac{1}{|\mfC_n|2^n}\sum_{C\in\mfC_n}(C\otimes\mcI_B)(\mcI\otimes\tr_A(\rho))(C\otimes\mcI_B)^\dagger\\
&=& \frac{1}{2^n}\mcI\otimes\tr_A(\rho)
\end{eqnarray}
\end{proofof}

\section{Clifford Technical Details}\label{sec:cliffordtechnical}
Here we prove \Le{mix}, \Le{pauliTw} and \Le{clifTw}, which were used to prove \Le{CliffordDiscretization} (and \Le{pauliTw} is also used to prove \Le{PauliDiscretization}). 
\\~\\
\begin{statement}{\Le{mix}}\textbf{(Pauli Partitioning by Cliffords)} 
\textit{For every
  $P,Q\in\mbP_m\setminus\{\mcI\}$ it holds that : $\left|\left\{C\in\mfC_m |
      C^\dagger PC =Q\right\}\right| = \frac {\left|\mfC_m\right|} {\left|\mbP_m
    \right| -1}= \frac {\left|\mfC_m\right|} {4^m -1}$.}
    \end{statement}
\begin{proofof}{\Le{mix}}
  We first claim that for every
  $Q,P\in\mbP_m\setminus \mcI$ there exists $D\in\mfC_m$ such that $D^\dagger P
  D=Q$. We will prove this claim by induction. Specifically, we show that
  starting from any non identity Pauli operator one can, using conjunction by
  Clifford group operator reach the Pauli operator $X\otimes \mcI^{\otimes
    m-1}$.

We first notice that the swap operation is in $\mfC_2$ since it holds that:
\begin{eqnarray}
    SWAP_{k,k+1} &= &CNOT_{k\rightarrow (k+1)}CNOT_{(k+1)\rightarrow k}CNOT_{k\rightarrow(k+1)}
\end{eqnarray}
Furthermore, we recall that $K^\dagger(XZ)K\propto X$ and $H^\dagger
ZH=X$. Therefore, any non identity Pauli $P=P_1\odots P_m$ can be transformed using $SWAP,H$
and $K$ to the form: $X^{\otimes k}\otimes \mcI^{\otimes m-k}$ (up to a
phase and for some $k\geq 1$). To conclude we use:
\begin{eqnarray}
CNOT_{1\rightarrow2}^\dagger (X_1\otimes X_2)CNOT_{1\rightarrow2}& =& X\otimes \mcI
\end{eqnarray}
which reduces the number of $X$ operations at hand. Applying this sufficiently many times results in reaching the desired form. Since this holds for any non-identity Pauli operators: $P,Q$ we know there are $C,D\in \mfC_m$ such that:
\begin{eqnarray}
X\otimes\mcI^{\otimes m-1}&=&C^\dagger PC=D^\dagger QD \\
&\Rightarrow& DC^\dagger PCD^\dagger= Q
\end{eqnarray}
therefore $CD^\dagger$ is the operator we looked for.  

Given $P'\in\mbP_m\setminus\mcI$, define $A_{P',Q}$ as follows $A_{P',Q}\EqDef\left\{C\in\mfC_m| C^\dagger P'C=Q\right\}$. We will show that $|A_{P',Q}|$ is independent of $P'$. Now fix $P'\in\mbP_m\setminus\mcI$ and let $D\in\mfC_m$ be one of the operators for which the following equality holds: $D^\dagger P'D=Q$. Then it holds
that for all $Q'\in\mbP_m\setminus\mcI$, $D^\dagger C\in A_{Q',Q} \iff C\in A_{P',Q'}$. Therefore $|A_{Q',Q}| = |A_{P',Q'}|$ for all non identity $P',Q'$ and $Q$. Using the fact that $|A_{P',Q'}| = |A_{Q',P'}|$, it follows that $|A_{P',Q'}|$ is independent of $Q'$ and $P'$. 

Now note that the sets $\{A_{P',Q'}\,:\forall P' \}$ form a partition of $\mfC_m$. These sets clearly do not intersect. Observe that for each $C\in \mfC_m$, there exists $P'\in\mbP_m\setminus\mcI$ such that $P' = CQ'C^\dagger$. Since all the sets in the partition have the same size, we obtain:
\begin{eqnarray}
\left|\mfC_m\right| =\sum_{P'\in\mbP_m\setminus\mcI}\left|A_{P',Q}\right|
=& (4^m -1)\left|A_{P,Q}\right|
\end{eqnarray}
which concludes the proof.
\end{proofof}
\\~\\
\begin{statement}{\Le{pauliTw}}\textbf{(Pauli Twirl)} \textit{Let $P \ne P'$ be generalized Pauli operators. For
  any density matrix $\rho'$ on $m'>m$ qubits it holds that
  \begin{equation}
\sum\limits_{Q\in \mbP_m} (Q^\dagger P Q\otimes\mcI) \rho' (Q^\dagger (P')^\dagger Q\otimes\mcI) = 0\nonumber
\end{equation}}
\end{statement}
\begin{proofof}{\Le{pauliTw}}
Let  $P\ne P'$ be generalized Pauli operator $P=X^aZ^b$ and $P'=X^{a'}Z^{b'}$.
\begin{eqnarray}
\sum_{Q\in \mbP_m} (Q^\dagger P Q\otimes\mcI) \rho'
(Q^\dagger P'^\dagger Q\otimes\mcI)
&=&
\sum_{d,c=0}^{q-1} ((X^cZ^d)^\dagger X^aZ^b (X^cZ^d)\otimes\mcI) \rho'
((X^cZ^d)^\dagger (X^{a'}Z^{b'})^\dagger (X^cZ^d)\otimes\mcI)\nonumber
\end{eqnarray}
We use the fact that $Z^dX^c=\omega_q^{dc} X^cZ^d$ (see Definition \ref{defpauli}) and some algebra:
\begin{eqnarray}
{\ldots} &=&
\sum_{d,c=0}^{q-1} \omega_q^{d(a-a')+c(b-b')} (X^aZ^b\otimes\mcI) \rho'
(Z^{-b'}X^{-a'}\otimes\mcI)\\
&=& (X^aZ^b\otimes\mcI) \rho'
(Z^{-b'}X^{-a'}\otimes\mcI) \sum_{c=0}^{q-1} \omega_q^{c(b-b')} \sum_{d=0}^{q-1}
\omega_q^{d(a-a')}
\end{eqnarray}
To conclude the proof we recall that $a\ne a'$ or $b\ne b'$, hence
one of the above sums vanishes.

\end{proofof}
\\~\\
\begin{statement}{\Le{clifTw}}\textbf{(Clifford Twirl)} \textit{Let $P\ne P'$ be Pauli operators. For
  any density matrix $\rho'$ on $m'>m$ qubits it holds that
  \begin{equation}
  \sum\limits_{C\in\mfC_m}(C^\dagger P C\otimes\mcI)\rho' (C^\dagger
    (P')^\dagger C \otimes\mcI)=0\nonumber
    \end{equation}}
    \end{statement}
\begin{proofof}{\Le{clifTw}}
Notice that applying a random Clifford operator ``includes'' the application
  of a random Pauli:
\begin{eqnarray}
  \sum_{c\in \mfC_m} (C^\dagger PC\otimes\mcI)\rho' (C^\dagger (P')^\dagger C\otimes\mcI) &=& \sum_{c\in \mfC_m} ((CQ)^\dagger P(CQ)\otimes\mcI)\rho' ((CQ)^\dagger (P')^\dagger(CQ)\otimes\mcI)
\end{eqnarray}
Equality holds for any $Q\in \mfC_n$ since it is nothing but a change of order
of summation.
\begin{eqnarray}
\ldots &=& \sum_{Q\in \mbP_m}\frac 1 {|\mbP_m|}\sum_{c\in \mfC_m} ((CQ)^\dagger P(CQ)\otimes\mcI)\rho' ((CQ)^\dagger (P')^\dagger(CQ)\otimes\mcI)\\
&=& \sum_{c\in \mfC_m} \frac 1 {|\mbP_m|}\sum_{Q\in \mbP_m}((CQ)^\dagger P(CQ)\otimes\mcI)\rho' ((CQ)^\dagger (P')^\dagger(CQ)\otimes\mcI)\\
&=& \sum_{c\in \mfC_m} \frac 1 {|\mbP_m|}\sum_{Q\in \mbP_m}(Q^\dagger(C^\dagger PC)Q\otimes\mcI)\rho'(Q^\dagger (C^\dagger (P')^\dagger C)Q\otimes\mcI)
\end{eqnarray}
By \Le{pauliTw} (with $q = 2$), we know that this expression is 0 if $C^\dagger PC \neq C^\dagger P'C$. 
\end{proofof}

\section{Logical Gates on Signed Polynomial Codes}\label{app:poly}
In this section we prove that the logical operators given in Section \ref{sec:logicalgates} behave as claimed. We first prove that the logical $X$ operator is correct.

\begin{proofof}{\Cl{claim:logicalx}}
We can easily verify that applying $X^{k_1x}\otimes\cdots\otimes X^{k_mx}$ is the logical $\wt{X}_k^x$
operation:
\begin{equation}\begin{split}
\wt{X}_k^x \ket{S_a^k} =& (X^{k_1x}\otimes\cdots\otimes X^{k_mx})  \frac{1}{\sqrt{q^d}}\sum_{f:def(f)\le d ,
f(0)=a}\ket{ k_1f(\alpha_1),{\ldots} , k_mf(\alpha_m)} \\
=&  \frac{1}{\sqrt{q^d}}\sum_{f:def(f)\le d ,
f(0)=a}\ket{ k_1(f(\alpha_1)+x),{\ldots} , k_m(f(\alpha_m)+x)}
\end{split}
\end{equation}
Setting $f'(\alpha) = f(\alpha)+x$:
\begin{equation}\begin{split}
{\ldots} =& \frac{1}{\sqrt{q^d}}\sum_{f':deg(f')\le d ,
f'(0)=a+1}\ket{ k_1f'(\alpha_1),{\ldots} , k_mf'(\alpha_m)}\\
=&\ket{S_{a+x}^k}
\end{split}\end{equation}
\end{proofof}

We now prove that the logical \textit{SUM} operator is correct.

\begin{proofof}{\Cl{claim:logicalsum}}
\begin{equation}\begin{split}
\wt{\textit{SUM}} \ket{S_a} \ket{S_b} =\ & (\textit{SUM})^{\otimes
m}\frac{1}{{q^d}}\sum_{ f(0)=a}\ket{ k_1f(\alpha_1),{\ldots} ,
k_mf(\alpha_m)} \sum_{ h(0)=b}\ket{ k_1h(\alpha_1),{\ldots} ,k_mh(\alpha_m)}\\
=& \frac{1}{{q^d}}\sum_{f(0)=a,h(0)=b}\ket{ k_1f(\alpha_1),{\ldots} ,
k_mf(\alpha_m)}\ket{ k_1(h(\alpha_1)+f(\alpha_1)),{\ldots} ,
k_m(h(\alpha_m)+f(\alpha_m))}
\end{split}\end{equation}
We set $g(\alpha)= f(\alpha)+h(\alpha)$
\begin{equation}\begin{split}
{\ldots} =& \frac{1}{{q^d}}\sum_{f(0)=a,g(0)=a+b}\ket{ k_1f(\alpha_1),{\ldots} ,
k_mf(\alpha_m)} \ket{ k_1g(\alpha_1),{\ldots} ,
k_mg(\alpha_m)} \\
=& \ket{S_a^k}\ket{S_{a+b}^k}
\end{split}\end{equation}
\end{proofof}

We proceed to the proof of the lemma needed for the logical Fourier transform.

\begin{proofof}{\Le{inter}}
A polynomial $p$ of degree $\le m-1$ is completely determined by it's
values in the points $\alpha_i$. We write $p$ as in the form of the
Lagrange interpolation polynomial:
$f(x) = \sum_i \prod_{j\ne i} \frac{x-\alpha_j}{\alpha_i
-\alpha_j}f(\alpha_j)
$.
Therefore, we set $c_i = \prod_{j\ne i}\frac
{-\alpha_j}{\alpha_i-\alpha_j}$ and notice that it is independent of $p$,
and the claim follows.
\end{proofof}

We continue to the proof of the logical Fourier transform.

\begin{proofof}{\Cl{claim:fourier}}
We denote $\ket{kf} = \ket{k_1f(\alpha_1),{\ldots} ,k_mf(\alpha_m)}$
\begin{eqnarray}
F_{c_1}\otimes F_{c_2}{\ldots} \otimes F_{c_m}\ket{S_a^k} &= &
q^{-d/2} F_{c_1}\otimes F_{c_2}{\odots} F_{c_m}\sum_{f:def(f)\le
d,f(0)=a}\ket{kf} \\
& = & q^{-d/2} q^{-m/2}\sum_{f:def(f)\le d,f(0)=a} \sum_{b_1,{\ldots} ,b_m}
\omega_q^{\sum_i c_ik_if(\alpha_i)b_i}\ket{b_1,{\ldots} ,b_m}
\end{eqnarray}
We think of the $b_i$'s as defining a signed polynomial $g$ of degree $\le
m-1$ that is $k_ig(\alpha_i)=b_i$ and split the sum according to $g(0)$:
\begin{eqnarray}\label{ignore}
{\ldots} & = & q^{-(m+d)/2} \sum_{\substack{f:def(f)\le d\\f(0)=a}} \sum_b
\sum_{\substack{g:deg(g)\le m-1\\g(0)=b }}
\omega_q^{\sum_i c_ik_if(\alpha_i)k_ig(\alpha_i)}\ket{kg}\\
&=& q^{-(m+d)/2} \sum_{\substack{f:def(f)\le d\\f(0)=a}} \sum_b
\sum_{\substack{g:deg(g)\le m-1\\g(0)=b }}
\omega_q^{\sum_i c_if(\alpha_i)g(\alpha_i)}\ket{kg} \label{pre}
\end{eqnarray}
We temporarily restrict our view to polynomials $g$ with degree at most
$m-d-1$ and therefore  the polynomial $fg$ has degree at most $m-1$.
We use \Le{inter} on $fg$:
\begin{eqnarray}
\sum_{i=1}^m c_i (fg)(\alpha_i) &= fg(0) &=ab
\end{eqnarray}
Going back to \Eq{pre}:
\begin{eqnarray}
q^{-(m+d)/2} \sum_{f,g} \sum_{b\in F_q} \omega_q^{\sum_i c_i
(fg)(\alpha_i)}\ket{kg}  &=& q^{-(m+d)/2}
\sum_{b\in F_q}\sum_{f,g} \omega_q^{ab}\ket{kg}
\end{eqnarray}
Where the summation is over all $f,g$ such that $f(0)=a$ and
$g(0)=b$ while the degrees of $f$ and $g$ are at most $d$ and
$m-d-1$ respectively.

The sum does not depend on $f$ and there are exactly $q^d$
polynomials $f$ in the sum, therefore, we can write the expression as :
\begin{equation}\begin{split}
{\ldots} =&\quad  q^{-(m+d)/2} \sum_{b\in F_q}q^d\sum_{g}
\omega_q^{ab}\ket{kg} \\
\ =&\quad  \frac{ 1}{ \sqrt{q}} \sum_{b\in F_q} \omega_q^{ab}\frac {1}
{ \sqrt{q^{m-d-1}}} \sum_{g:deg(g)\le
m-d-1, g(0)=b}\ket{kg} \\
\ =&\quad  \frac 1 {\sqrt{q}} \sum_{b\in F_q} \omega_q^{ab}\ket{\wt{S_b^k}}
\end{split}\end{equation}
Since the above expression has norm 1, if follows that the
coefficients that we temporally ignored at \Eq{ignore} all vanish.
\end{proofof}

Finally, we prove that the logical $Z$ operator is correct.

\begin{proofof}{\Cl{claim:logicalz}}
\begin{eqnarray}
\wt{Z}_k^z\ket{S_a^k} &=& (Z^{k_1c_1z}\otimes\cdots\otimes Z^{k_mc_mz})\frac{1}{\sqrt{q^d}}\sum_{f:def(f)\le d ,
f(0)=a}\ket{ k_1f(\alpha_1),{\ldots} , k_mf(\alpha_m)} \\
&=& \frac{1}{\sqrt{q^d}}\sum_{f:def(f)\le d ,
f(0)=a}\omega_q^{\sum_i k_ic_izk_if(\alpha_i)}\ket{ k_1f(\alpha_1),{\ldots} , k_mf(\alpha_m)}\\
&=& \frac{1}{\sqrt{q^d}}\sum_{f:def(f)\le d ,
f(0)=a}\omega_q^{z\sum_i c_if(\alpha_i)}\ket{ k_1f(\alpha_1),{\ldots} , k_mf(\alpha_m)}\\
&=& \frac{1}{\sqrt{q^d}}\sum_{f:def(f)\le d ,
f(0)=a}\omega_q^{zf(0)}\ket{ k_1f(\alpha_1),{\ldots} , k_mf(\alpha_m)}\\
&=& \omega_q^{za}\ket{S_a^k}
\end{eqnarray}
\end{proofof}


\section{Notation Tables}\label{app:tables}
We begin with notation used for both the Clifford and polynomial protocols, and then proceed to notation used only in the polynomial protocol (beginning with the sign key $k$, Definition \ref{def:SignedPolynomial}). 
\begin{center}
    \begin{tabular}{ | l | l | p{12cm} |}
    \hline
    Notation & Reference & Explanation  \\ \hline
    $\gamma$ & Protocol \ref{prot:CliffordIP}, \ref{prot:PolynomialIP}  & Error of circuit which is being applied; used in both protocols \\\hline
    
    $\epsilon$ & Protocol \ref{protocol:cliffordqas}, \ref{protocol:polynomialqas} & Security parameter \\\hline
    
    $\Pi_0,\Pi_1$ & Definition \ref{def:qas} & Projections used to define security of a \QAS \\\hline
    
    $L$ & Protocol \ref{prot:PolynomialIP} & Number of Toffoli gates in the circuit \\\hline
    
    $n$ & Protocol \ref{prot:CliffordIP},\ref{prot:PolynomialIP} & Number of qubits in the circuit which is being applied\\\hline
    
            $N$ & Definition \ref{def:qcircuit} & Number of gates in the circuit which is being applied\\\hline
            
    $E$ & \Th{thm:CliffordAuth},\ref{thm:PolynomialAuth},\ref{thm:CliffordIP} & Eve's environment \\\hline
    
     $\mathcal{E}$ & Section \ref{sec:notation} & Prover's environment register in \QPIP\ protocol \\\hline

        $k$ & Definition \ref{def:SignedPolynomial} & Sign key for signed polynomial code; used as superscript for encoded states \\\hline

        $\tilde{U}$ & Section \ref{sec:logicalgates} & Logical version of a gate $U$ for the signed polynomial code \\\hline
        
        $m,d$ & Definition \ref{def:SignedPolynomial} & Length and degree of polynomial code \\\hline
    $E_k, D_k$ & Definition \ref{def:encodingcircuit}, \ref{def:interpolationcircuit} & Encoding circuit for signed polynomial code ($E_k = D_k (F^{\otimes d}\otimes\mcI)$) \\ \hline
    
    $\mbP_m$ & Definition \ref{defgeneralized} & Group of generalized Pauli operators \\ \hline
    
    $g$ & Equation \ref{deffunctiong} & $g$ takes as input strings in $F_q^m$ and returns the first coordinate of each string \\\hline

     $\rho^k$ & Definition \ref{def:SignedPolynomial} & Initial state in polynomial \QPIP containing $n$ authenticated 0 states and $L$ authenticated magic states \\\hline

    $\mathcal{P}_i$ & Section \ref{sec:notation}  & Register containing prover's $3m(L-i + 1) + mn$ qubits at the end of round $i-1$ \\\hline
    
    $\mathcal{V}_i$ & Section \ref{sec:notation} & Register containing the $3m(i-1)$ qubits which have been sent to the verifier in rounds $1,\ldots,i-1$ \\\hline

    $\mathcal{F}$ & After equation \ref{startoffinalround} & Register of $m$ qudits containing the final authenticated qudit given to the verifier in the final round  \\\hline
    
        $\mathcal{P}_{final}$ & After equation \ref{startoffinalround} & Register of authenticated qudits remaining with the prover at the end of the protocol ($\mathcal{P}_{L+1} = \mathcal{P}_{final}\cup \mathcal{F}$)  \\\hline
        
            $\mathcal{V}_{final}$ & After equation \ref{startoffinalround} & Register containing all $m(3L + 1)$ qudits sent to the verifier during the protocol ($\mathcal{V}_{final} = \mathcal{F}\cup \mathcal{V}_{L+1}$)  \\\hline
    
    $\tau_i(z,x,k)$ & Section \ref{sec:notation} & Keys held by verifier at the start of round $i$ ($z,x\in F_q^{|\mathcal{P}_i|}$) \\\hline
    
    $\delta_i$ & Section \ref{sec:notation} & Measurement result in $F_q^{3m}$ of prover in round $i$ \\\hline

    $\Delta_i$ & Claim \ref{claim:polystateformat} & Measurement results (composing a string in $F_q^{3mi}$) from rounds $1$ to $i$ \\\hline
    
    $\tilde{Q}_i$ & Section \ref{sec:conversiontological} & Logical Clifford operators applied in round $i$ \\\hline
    
    $\tilde{C}_{\beta_i}$ & Section \ref{sec:conversiontological} & Clifford correction operators for Toffoli gate $i$ if measurement result is $\beta_i\in F_q^{3}$ \\\hline
    
    $\rho_{g(\Delta_{i-1})}^k$& Claim \ref{claim:polystateformat} & The state on $mn + 3mL$ qudits resulting from applying operations requested in rounds $1,\ldots,i-1$ \\\hline
    
    $l\in F_q^{3L}$& Fact \ref{fact:teleportationstate} & Used in the sum over all possible teleportation measurement results \\\hline
    
    $\beta\in F_q^{3L}$& Fact \ref{fact:teleportationstate} & Used to denote one fixed measurement result \\\hline

    $\hat{\Pi}_0$& Equation \ref{acceptingprojection} & Used to denote the accepting subspace on $3mL + m$ qudits in the polynomial \QPIP \\\hline
    
    $\Pi_{G_a}$& Equation \ref{simplifiedprojection} & Projection (on $3mL$ qudits) onto a valid, decoded measurement result $a\in F_q^{3L}$ \\\hline

    \end{tabular}
\end{center}

\end{document}